\documentclass[acmsmall,screen,nonacm]{acmart}

\usepackage[T1]{fontenc}
\usepackage{syntax}
\usepackage{mathtools}
\usepackage{caption}
\usepackage{enumitem}
\usepackage{wasysym}
\usepackage{listings}
\usepackage{multido}
\usepackage{stmaryrd}
\usepackage{galois}
\usepackage{tikz-cd}
\usepackage{thm-restate}
\usepackage{longtable}
\usepackage{mathtools}
\usepackage{subcaption}
\usepackage[most]{tcolorbox}
\usepackage{tikz}
\usetikzlibrary{arrows.meta,positioning}
\usepackage{wrapfig}
\usepackage[math]{blindtext}
\usepackage{setspace}
\usepackage{placeins}
\usepackage{afterpage}

\usepackage{float}

\usepackage[textsize=tiny,colorinlistoftodos,textwidth=1.3cm]{todonotes}

\definecolor{azure}{HTML}{0080FF}
\definecolor{dgreen}{HTML}{228B22}

\newcommand{\newextmathcommand}[2]{\newcommand{#1}{\ensuremath{#2}\xspace}
}

\newcommand{\magmoid}{magmoid\xspace}
\newcommand{\magmoids}{magmoids\xspace}

\newcommand{\IR}{\textsc{ImpADT}\xspace}

\newcommand{\pcf}{\textsc{PCF}\xspace}

\newcommand{\nn}{\mathbb{N}}

\newcommand{\zz}{\mathbb{Z}}
\newcommand{\zzi}{\overline{\mathbb{Z}}}
\newcommand{\qq}{\mathbb{Q}}

\newcommand{\bb}{\mathbb{B}}

\newcommand{\pp}{\mathcal{P}}
\newcommand{\dd}{\mathcal{D}}
\newcommand{\ww}{\mathcal{W}}
\newcommand{\ii}{\mathcal{I}}

\newextmathcommand{\Data}{\mathrm{Data}}
\newextmathcommand{\Metrics}{\mathcal{M}}
\newextmathcommand{\ListInt}{\mathtt{ListInt}}
\newextmathcommand{\Cons}{\mathtt{Cons}}
\newextmathcommand{\Nil}{\mathtt{Nil}}
\newextmathcommand{\Some}{\mathtt{Some}}
\newextmathcommand{\None}{\mathtt{None}}
\newextmathcommand{\unit}{1}

\newextmathcommand{\Cst}{\mathtt{Cst}}
\newextmathcommand{\Var}{\mathtt{Var}}
\newextmathcommand{\Add}{\mathtt{Add}}
\newextmathcommand{\Mul}{\mathtt{Mul}}

\newextmathcommand{\sconcat}{\mathtt{f}}

\newextmathcommand{\cardm}{\mathtt{card}}
\newextmathcommand{\len}{\mathtt{len}}

\newextmathcommand{\abs}{\mathrm{abs}}

\newextmathcommand{\catex}{\mathbb{C}}
\newextmathcommand{\catexbis}{\mathbb{D}}

\newextmathcommand{\costm}{\mathcal{C}}
\newextmathcommand{\statespace}{\mathrm{Hw}}
\newextmathcommand{\statem}{{\mathcal{S}t}}\newextmathcommand{\counterm}{\mathcal{K}}

\DeclareMathOperator{\lfp}{\mathrm{lfp}\xspace}
\DeclareMathOperator{\gfp}{\mathrm{gfp}\xspace}
\newextmathcommand{\pleq}{\mathbin{\dot\leq}}
\newextmathcommand{\pgeq}{\mathbin{\dot\geq}}

\newcommand{\letin}[4]{\mathrm{let}_{#1}~#2\shortleftarrow#3~\mathrm{in}~#4}

\newcommand{\eqdef}{\overset{\Delta}{=}}

\newcommand{\Fix}{\mathrm{Fix}}

\newcommand{\candf}{\hat{f}}

\newcommand{\id}{\mathrm{id}}
\newcommand{\Id}{\mathrm{Id}}
\newcommand{\Ob}{\mathrm{Ob}}

\newextmathcommand{\Hom}{\mathrm{Hom}}
\newextmathcommand{\End}{\mathrm{End}}

\newextmathcommand{\Set}{\mathbf{Set}}
\newextmathcommand{\Ord}{\mathbf{Ord}}
\newextmathcommand{\Pos}{\mathbf{Pos}}
\newextmathcommand{\CPO}{\mathbf{CPO}}
\newextmathcommand{\CLat}{\mathbf{CLat}}
\newextmathcommand{\Kleisli}{\mathcal{K}\ell}
\newextmathcommand{\syntcat}{\mathcal{L}}

\newextmathcommand{\Vars}{\mathcal{V}}
\newextmathcommand{\Locs}{\mathcal{L}}
\newextmathcommand{\wnodes}{\mathcal{N}}
\newextmathcommand{\Prog}{\mathtt{Prog}}
\newextmathcommand{\Funs}{\mathcal{F}}
\newextmathcommand{\Types}{\mathcal{T}}
\newextmathcommand{\BasicTypes}{\mathcal{B}}
\newextmathcommand{\Conss}{\mathcal{C}}

\newextmathcommand{\ite}{\mathrm{ite}}

\newextmathcommand{\Mem}{\mathrm{Mem}}
\newextmathcommand{\aMem}{\mathrm{Mem}^\sharp}

\makeatletter
\newcommand{\sesquisharp}{{\mathbin{\mathpalette\sesquisharp@{}}}}
\newcommand{\sesquisharp@}[2]{\vcenter{\hbox{\begin{tikzpicture}[baseline=(current bounding box.center), x=1ex,y=1ex, line width=0.08ex, line cap=round]
      \node[inner sep=0pt] at (0,0) {$\m@th#1\sharp$};
      \node[inner sep=0pt] at (.35,.12) {$\m@th#1\sharp$};
    \end{tikzpicture}}}}
\makeatother

\newextmathcommand{\funvec}{\vec{\mathbf{f}}}\newextmathcommand{\argvec}{\vec{\mathbf{x}}}

\NewCommandCopy{\rawPhi}{\Phi}
\renewcommand{\Phi}{\mathrm{\rawPhi}}
\NewCommandCopy{\rawPsi}{\Psi}
\renewcommand{\Psi}{\mathrm{\rawPsi}}

\newcommand{\sem}[1]{\llbracket#1\rrbracket}
\newcommand{\semop}[1]{\llbracket#1\rrbracket_{\Phi}}
\newcommand{\semstmt}[1]{\llbracket#1\rrbracket_{\lambda}}
\newcommand{\widesemstmt}[1]{\left\llbracket#1\right\rrbracket_{\lambda}}
\newcommand{\semexpr}[1]{\llbracket#1\rrbracket_{\mathbb{E}}}
\newcommand{\semaexp}[1]{\llbracket#1\rrbracket_{\mathrm{AExp}}}
\newcommand{\sembexp}[1]{\llbracket#1\rrbracket_{\mathbb{B}}}

\newcommand{\semtypes}[1]{\llbracket#1\rrbracket_{\mathcal{T}}}
\newcommand{\asem}[1]{\llbracket#1\rrbracket^\sharp}
\newcommand{\asemop}[1]{\llbracket#1\rrbracket_{\Phi}^\sharp}
\newcommand{\asemstmt}[1]{\llbracket#1\rrbracket_{\lambda}^\sharp}
\newcommand{\asemexpr}[1]{\llbracket#1\rrbracket_{\mathbb{E}}^\sharp}

\newcommand{\asembexp}[1]{\llbracket#1\rrbracket_{\mathbb{B}}^\sharp}

\newcommand{\mathcmt}[1]{\ensuremath{\Lbag\text{#1}\Rbag}\xspace}

\tikzset{
  no line/.style={draw=none,
    commutative diagrams/every label/.append style={/tikz/auto=false}}}

\AtEndPreamble{\theoremstyle{acmdefinition}
\newtheorem{remark}[theorem]{Remark}}
\AtEndPreamble{\theoremstyle{acmplain}
\newtheorem*{theorem*}{Theorem}}
\AtEndPreamble{\theoremstyle{acmplain}
\newtheorem*{proposition*}{Proposition}}
\AtEndPreamble{\theoremstyle{acmplain}
\newtheorem{fact}{Fact}}

\begin{document}

\title[Abstract Compilation as Abstraction of Operator Semantics, applied to Cost Analysis]
{Abstract Compilation as Abstraction of Operator Semantics,\\ applied to Cost Analysis}

\author{Louis Rustenholz}
\email{louis.rustenholz@imdea.org}
\orcid{0000-0002-1599-2431}
\affiliation{\institution{Universidad Polit\'{e}cnica de Madrid}
  \city{Madrid}
  \country{Spain}
}
\affiliation{\institution{IMDEA Software Institute}
  \city{Pozuelo de Alarc\'{o}n}
  \country{Spain}
}

\author{Alessio Mansutti}
\orcid{0000-0002-1104-7299}
\affiliation{\institution{IMDEA Software Institute}
  \city{Pozuelo de Alarc\'{o}n}
  \country{Spain}
}

\author{Pedro L\'{o}pez-Garc\'{i}a}
\orcid{0000-0002-1092-2071}
\affiliation{\institution{Spanish Council for Scientific Research (CSIC)}
  \city{Madrid}
  \country{Spain}
}
\affiliation{\institution{IMDEA Software Institute}
  \city{Pozuelo de Alarc\'{o}n}
  \country{Spain}
}

\author{Félix Ridoux}
\authornote{Work performed under previous IMDEA affiliation.}
\orcid{0009-0002-6312-0483}
\affiliation{\institution{Computer Science Laboratory of Sorbonne University (LIP6)}
  \city{Paris}
  \country{France}
}
\affiliation{\institution{STMicroelectronics}
  \city{Crolles}
  \country{France}
}

\author{Niki Vazou}
\orcid{0000-0003-0732-5476}
\affiliation{\institution{IMDEA Software Institute}
  \city{Pozuelo de Alarc\'{o}n}
  \country{Spain}
}

\author{Manuel V. Hermenegildo}
\orcid{0000-0002-7583-323X}
\affiliation{\institution{Universidad Polit\'{e}cnica de Madrid}
  \city{Madrid}
  \country{Spain}
}
\affiliation{\institution{IMDEA Software Institute}
  \city{Pozuelo de Alarc\'{o}n}
  \country{Spain}
}

\renewcommand{\shortauthors}{Louis Rustenholz et al.}
\authorsaddresses{}

\begin{abstract}
Least fixpoints are fundamental to program semantics, but they
abstract away the recursive structure that generated them.
We introduce \emph{operator semantics}: a semantic intermediate representation between syntax and classical denotational semantics, which treats programs as operators.
Abstract compilation is then understood as the act of abstracting such operators.
We develop higher-order abstract domains for functions, operators, and
programs themselves, in which composition is the
key novel primitive, together with a categorical framework for constructing sound, precise,
and modular abstract compilers.  We instantiate this framework in the
context of recurrence-based static cost analysis, developing
solver-independent, optimal recurrence extraction techniques for
recursive programs over algebraic data types,
that support general function unknowns and catamorphic metrics,
a broad class of size metrics beyond traditional approaches.

\end{abstract}

\ccsdesc[500]{Theory of computation~Program analysis}
\ccsdesc[500]{Theory of computation~Program semantics}
\ccsdesc[300]{Theory of computation~Categorical semantics}
\ccsdesc[300]{Theory of computation~Abstraction}
\ccsdesc[300]{Software and its engineering~Automated static analysis}
\ccsdesc[300]{Software and its engineering~Compilers}

\keywords{Abstract interpretation; least fixpoints; operator
  semantics; abstract compilation; higher-order abstract domains;
  static cost analysis; recurrence extraction; catamorphic metrics.}

\maketitle

\section{Introduction}
\label{sec:introduction}

Least fixpoints are central to the semantics of programs.  They
provide elegant characterisations of recursive
behaviours, and form the basis of foundational reasoning techniques such as abstract interpretation~\cite{Cousot79}.
However, from the perspective of program analysis, least fixpoints
also discard valuable information: the recursive structure that
generated them,
which may be exploited to produce concise representations of expressive invariants
and guide analyses.
This paper explores a different semantic viewpoint: rather than taking
fixpoints as the primary semantic objects, and discarding the semantic
operator that generated them as a mere technical auxiliary, we
promote it to a first-class object.
Our viewpoint can be summarised as: \ \textbf{Programs $\equiv$ Operators.}

Once this viewpoint has been adopted, it is natural to view Galois connections between operator spaces,
i.e. abstractions of operators, as \emph{abstractions of programs themselves}, recursive structure included.
We may view this as a form of \emph{abstract compilation},
applicable to a wide variety of notions of programs and abstractions.
The viewpoint is in fact symmetric: operators arising from functional
equations may be
interpreted as programs, further encouraging the application of static
analysis beyond traditional programming languages (e.g. in cyberphysical systems or systems biology).
Our own motivation to abstract programs is summarised as
``simplifying problems gives access to powerful techniques to solve them'':
increasingly abstracted programs are amenable to increasingly powerful
analyses, potentially up to complete methods
(e.g. decidable logic fragments, some classes of recurrence equations),
enabling the use of diverse analysis portfolios.

To implement abstract compilers based on these ideas, we build upon the classical abstract
interpretation methodology,
but develop original \emph{higher-order abstract domains}, whose elements are interpreted as
functions, operators and \emph{programs themselves}.
This distinguishes our goals and perspective from established uses of operators and semantic
equations
(e.g.~\cite{Bourdoncle93-long} in abstract interpretation or the immediate consequence operator in logic programming~\cite{EmKo76}),
as well as from previous work on \emph{abstract compilation}, discussed below.

\subsubsection*{Application to Cost Analysis}

A primary motivation for this work stems from challenges in recurrence-based static cost analysis~\cite{Wegbreit75,Le88-shortest,Rosendahl89-shortest,granularity-short,caslog-short,low-bounds-ilps97-short,resource-iclp07,ciaopp-sas03-journal-scp-shortest,AlbertAGP11a-short,plai-resources-iclp14-short,gen-staticprofiling-iclp16-extrashortest,montoya-phdthesis-short,resource-verification-tplp18-shortest,LommenGiesl23-shortest}
which consists in extracting recurrence equations from programs, whose solutions
bound resource usage (e.g. execution time, memory consumption)
as functions of numerical descriptions of inputs, called \emph{(size) metrics}.
More generally, recurrence extraction and resolution is a powerful technique
to infer \emph{non-linear numerical invariants}~\cite{kovacs-phdthesis,kovacs17-short,kincaid2018,LommenGiesl23-shortest,Wang-PRS23-short}.
Historically, these approaches have been constrained by the capabilities of
available solvers, traditionally focused on exact resolution of well-behaved equations. Unfortunately, such equations do not capture the
full range of complex behaviours arising in programs,
which may be multivariate, non-linear recursive, or non-monotone;
contain conditional expressions; require parametric optimisation to model non-determinism, etc.
Recent advances in recurrence solving, particularly
postfixpoint-based~\cite{order-recsolv-sas24,Goharshady-ESOP2025,absfun-STTT26,PLDI2026-supermartingalesuniquefixedpoints},
have considerably broadened the class of analysable equations, but recurrence
extraction has yet to fully exploit these developments.

Rather than incrementally adapting to successive solver improvements, we seek a
solver-independent notion of \emph{optimal} recurrence extraction. We answer this
question by making precise the long-standing intuition that ``recurrences are
abstractions of programs''.
Crucially, our notion of soundness compares operators rather than their fixpoints.
This fixed-point delaying approach yields optimal abstract programs---generalised
recurrence equations---that remain recursive, are often finitely representable,
and are amenable to further analysis by arbitrary techniques.

As a direct benefit of our semantic considerations, the extraction techniques
we propose yield great flexibility in the \emph{type} of the functional unknowns that can be
considered.
Additionally, we identify a class of metrics for which operator abstraction can be practically computed,
which we call \emph{catamorphic (collections of) metrics}, including metrics unsupported by established approaches.

\subsubsection*{Abstract Compilation}

The term \emph{abstract compilation} appears to have been coined
by~\cite{pracai-jlp}, and has reappeared regularly in the literature since,
with slightly different definitions. Our paper shows that these techniques may be viewed as different implementations and
syntactic representations of a \emph{single} underlying ideal semantic
notion: optimal abstraction of operator semantics (for varying function
types). Reviewing them thus situates our contribution while showcasing applications.

Abstract compilation was initially a technique to speed up analysers by
avoiding ``interpretation overhead'': partially evaluating the abstract
interpreter with respect to a (state-level) abstract domain and
program~\cite{DebrayWarren88,warren:TR88,pracai-jlp,BF-abscomp-96} yields an
abstract program in the original language whose \emph{execution} performs the
analysis for any abstract input.
Later work instead extracts abstract programs in order to \emph{analyse} them
by other techniques: \cite{RossignoliS06} extracts recursive Boolean
functions detecting non-cyclical data structures in object-oriented
languages,
and~\cite{Ancona-10-abscomp-ty-clp,Ancona-11-abscomp-ty-clp,Ancona-12-abscomp-ssi}
coinductive CLP(X) for elaborate type analyses.
In~\cite{jacquemin-arxiv19-abscomp-float}, the term denotes an ``intermediate
technique between testing and static analysis'', executing programs over
values representing e.g. ranges of floating-point numbers;
related possible future applications of abstract compilation include
non semantic-preserving compilers, e.g. speeding up numerical analysis
loops by (principled) precision loss.

Our framework captures and generalises these works thanks to the richness of
Galois connections of operators, together with flexible notions of function
spaces and abstract composition: this enables \emph{control-flow}
abstractions rather than merely value abstractions (as in partial evaluation
with respect to a fixed interpreter and value domain), and state
``abstractions'' that cannot arise from Galois connections at the state
level, with correctness-by-construction guarantees.

Abstract compilation is also related to \emph{program transformation by
abstract interpretation}; further comparisons are given in
Section~\ref{sec:rel-work}. For now, observe that, thanks to fixpoint
delaying, the operators we consider often admit \emph{finite representations}
(contrasting with the infinite trace sets and non-computable
syntax/semantics adjunctions of~\cite{Cousot02}), while enjoying more
convenient algebraic properties than more syntactic objects (e.g. syntax
trees themselves, or the free algebra domains of~\cite{lesbre2024compiling},
which are preorders rather than complete lattices).
In that sense, operator semantics may be viewed as an \emph{intermediate
choice between syntax and classical semantics}, identifying equivalent
presentations of programs while preserving recursive structure.

\subsubsection*{Main contributions}

\begin{itemize}[leftmargin=*,topsep=0pt]
  \item \emph{We propose operator semantics as an intermediate semantic representation of programs}
    (Section~\ref{sec:operator-semantics}),
    parameterised by flexible notions of computational effects and abstractions.
Unlike standard denotational semantics, operator semantics explicitly preserves
    recursive structure while abstracting away language-specific syntax.
We exemplify (Sec.~\ref{sec:ir-lang}, Fig.~\ref{fig:opsem-ir-num-monadic}) on an imperative-style
    language with algebraic data types (ADTs),
and show how \emph{order-enriched monads} not only enable invariant-based reasoning
    (Proposition~\ref{prop:clat-image-KT}),
but also capture analysis techniques, such as the \emph{counter} approach to
    non-linear invariant inference, via an order-enriched monad transformer~(Sec.~\ref{subsubsec:counter-monad-transformer}).

  \item \emph{We introduce abstract compilation as an abstraction of operators}
    (Section~\ref{sec:abs-op-for-abs-comp}).
To enable sound, precise and modular abstract compilers, \emph{we develop a systematic categorical framework} to design and study \emph{higher-order abstract
    domains} representing operators as well as abstraction \emph{constructors}.

    In particular, we identify and study \emph{algebraic properties of abstract composition}:
while it is non-associative in general,
    requiring implementer caution, best bracketings can be identified under some assumptions
    (Proposition~\ref{prop:best-bracketing}).
We also define a notion of \emph{abstract monads} (Definition~\ref{def:abstract-monad},
    Theorem~\ref{theorem:optimal-oplax-monad-by-GC}), promoting classical abstract domains to abstract effects
    composable with other effects. Precision advantages of
    abstract monadic composition are identified.

  \item \emph{We instantiate the framework for static cost analysis} of recursive programs over ADTs
    (Section~\ref{sec:application-and-implementation}),
    deriving sound and optimal recurrence extraction techniques that support broad classes of metrics
(applying \emph{size abstractions} derived from
    Theorems~\ref{theorem:kan-hom-domain-codomain-abstraction}~and~\ref{theorem:size-abstraction-kleisli-oplax-endofunctor}).
\end{itemize}

\section{Illustrative Example}
\label{sec:illustrative-example}

We now give the intuition behind \emph{operator semantics} and
\emph{abstract compilation} through a small illustrative example,
before presenting the formal framework.
Consider a simple program
manipulating symbolic and numerical data, such as the program of
Fig.~\ref{subfig:example-prog-concrete-syntax}, which performs (very
naive) symbolic differentiation
and evaluation of expressions.
Suppose we want to analyse it, for example to
infer bounds on its execution cost. To do so, we first need a semantic
description of the program together with a cost model, i.e. a way to
measure execution costs in terms of some resource of interest (e.g. time or a more abstract
resource).  For this example, we choose a cost model that assigns a
unit cost $1$ to function calls,
and logarithmic cost to integer arithmetic.\footnote{For simplicity, we only count $\ell(x)+\ell(y)$ for
additions and $\ell(x)\cdot\ell(y)$ for multiplication, where
$\ell(n)=1+\lfloor \log_2(\max(1,\abs(n))) \rfloor$ is essentially the
bit-length of an integer $n\in\zz$. All other operations have cost
$0$.}

The operator $\Phi$ in Fig.~\ref{subfig:example-prog-concrete-operator}
is the operator semantics for the program in Fig.~\ref{subfig:example-prog-concrete-syntax}.
Given a
  triple $\funvec = (\funvec_{{\tt diff}}, \funvec_{{\tt eval}},
  \funvec_{{\tt diffeval}})$
of arbitrary denotations associated with program functions
  ({\tt diff}, {\tt eval}, and {\tt diffeval}),
$\Phi(\funvec)$
  returns a triple of denotations obtained by performing a single
  semantic unfolding of every function definition
  simultaneously. Classical denotational semantics would interpret the recursive functions
  {\tt diff}, {\tt eval}, and {\tt diffeval} as components of a \emph{fixed point}
  of $\Phi$. We instead preserve $\Phi$ itself, and allow its evaluation outside of fixpoints.
The semantic operators considered in this paper are parameterised by
monads, which can be viewed as
  a modular way of specifying what kind of computations semantic
  functions perform.  Two effects are composed in our example: the powerset monad $\pp(-)$ that
  models nondeterminism by allowing functions to return sets
  of possible results (reminiscent of collecting semantics approaches
  in abstract interpretation), and a cost monad $\costm(-)$ (where
  $\costm(X)=X\times \zz$) that augments every result with an execution
  cost. Their composition therefore represents computations producing
  multiple possible results together with their associated costs.
For example, the semantic representation of the output of a call to {{\tt diff}}
  with an input expression $\Add(e_1,e_2)$ is a set of pairs, where the first element is the result and the second
  element, $2+c_1+c_2$, is its associated cost: the cost of the two
  recursive calls to {{\tt diff}} plus the costs of those calls,
  $c_1+c_2$, where these costs are provided
  by the input denotation $\funvec_{{\tt diff}}$.
  For function {\tt eval}, the cost of performing addition, namely $\ell(r_1)+\ell(r_2)$, is also added to the
  cost. $\Phi$ may be viewed as \emph{a semantic representation of the
  program itself}, control-flow included (for the recursion points
  selected when defining the operator semantics), with only syntactic
  details abstracted away.

  \definecolor{ConcreteBlue}{HTML}{EAF3FF}
\definecolor{AbstractOrange}{HTML}{FFF3E0}

\newcommand{\PipeSep}{2cm}

\newlength{\ExamplePanelHeight}
\setlength{\ExamplePanelHeight}{9.5cm}
\newcommand{\PanelWidth}{.49\textwidth}
\newcommand{\PanelSep}{\hspace{.005\textwidth}}

\tcbset{
  concretepanel/.style={
    enhanced,
    colback=ConcreteBlue,
    colframe=blue!55!black,
    boxrule=0.6pt,
    arc=2pt,
    left=3pt,right=3pt,top=3pt,bottom=3pt,
    width=\linewidth,
    height=\ExamplePanelHeight,
    valign=top
  },
  abstractpanel/.style={
    enhanced,
    colback=AbstractOrange,
    colframe=orange!75!black,
    boxrule=0.6pt,
    arc=2pt,
    left=3pt,right=3pt,top=3pt,bottom=3pt,
    width=\linewidth,
    height=9.2cm,  valign=top
  }
}
\begin{figure*}[t]
\centering
\scriptsize

\begin{tikzpicture}[
  pipebox/.style={
    rounded corners,
    draw,
    thick,
    inner xsep=2pt,
    inner ysep=3pt,
    align=center,
    minimum height=1.05cm
  },
  arr/.style={-{Latex[length=1.8mm]}, very thick},
  lab/.style={
    font=\scriptsize,
    align=center,
    inner sep=0pt,
},
  every node/.style={font=\scriptsize}
]

\node[pipebox, fill=ConcreteBlue, text width=.115\textwidth] (prog)
 {\textbf{(a) Concrete}\\\textbf{program}\\[2pt]
\scriptsize \IR\\[-2pt]syntax};

\node[pipebox, fill=ConcreteBlue, text width=.155\textwidth, right=1.4cm of prog] (op)
     {\textbf{(b) Concrete}\\\textbf{operator $\Phi$}\\[2pt]
\scriptsize symbolic values\\[-2pt]+ costs};

\node[pipebox, fill=AbstractOrange, text width=.155\textwidth, right=1.8cm of op] (absop)
     {\textbf{(c) Abstract}\\\textbf{operator $\Phi^\sharp$}\\[2pt]
\scriptsize numerical sizes\\[-2pt]+ costs};

\node[pipebox, fill=AbstractOrange, text width=.115\textwidth, right=1.4cm of absop] (chc)
 {\textbf{(d) Abstract}\\\textbf{program}\\[2pt]
\scriptsize (e.g. CHC\\[-2pt]encoding)};

\draw[arr] (prog) -- node[above, lab] {Operator \\ semantics \vspace{.4em} } (op);

\draw[arr] (op) -- node[above, lab] {Abstract \\ compilation \vspace{.4em}}
           node[below, font=\scriptsize\itshape, align=center] {size abstraction \\ $(nb,\ell_{\max})$} (absop);

\draw[arr] (absop) -- node[above, lab] {Syntax \\[-.75em] \phantom{.} \vspace{.4em}} (chc);

\end{tikzpicture}

\vspace{.1em}
\noindent
\begin{subfigure}[t]{.42\textwidth} \begin{tcolorbox}[concretepanel]
\textbf{(a) Concrete program: \IR syntax}
\begin{lstlisting}[morekeywords={type,int,def,if,then,else,match,return,true,false},basicstyle=\ttfamily\tiny,columns=fullflexible]
type Expr = {
  Cst(int) | Var | Add(Expr,Expr) | Mul(Expr,Expr)
}

def diff(e: Expr) -> Expr {
  let res: Expr = top;
  match (e) {
    Cst(c)      => { res = Cst(0) }
  | Var         => { res = Cst(1) }
  | Add(e1, e2) => { res = Add(diff(e1), diff(e2)) }
  | Mul(e1, e2) => {
    res = Add(Mul(diff(e1),e2), Mul(e1,diff(e2))) }
  }
  return res }

def eval(e: Expr, x: int) -> int {
  let res: int = top;
  match (e) {
    Cst(c)      => { res = c }
  | Var         => { res = x }
  | Add(e1, e2) => { res = eval(e1,x) + eval(e2,x) }
  | Mul(e1, e2) => { res = eval(e1,x) * eval(e2,x) }
  }
  return res }

def diffeval(e: Expr, x:int) -> int {
  return eval(diff(e), x)
}
\end{lstlisting}
\end{tcolorbox}
\phantomcaption
\label{subfig:example-prog-concrete-syntax}
\end{subfigure}\PanelSep \begin{subfigure}[t]{.57\textwidth} \begin{tcolorbox}[concretepanel]
\textbf{(b) Concrete operator semantics $\Phi$ (with $\pp\costm(-)$ monad)}
$$\Phi = \semop{\Prog}^{\pp\costm} \in \End\Big(\prod_{f\in\Funs}\Mem_{\Vars_{f,in}} \to \pp\costm(\Mem_{\Vars_{f,out}})\Big)$$
\vspace{-.5em}
\begin{align*}
          \big(\Phi(\funvec)\big)_{{\tt diff}} \colon
             \semtypes{{\tt Expr}} & \to \pp\costm(\semtypes{{\tt Expr}})\\
\Cst(c) & \mapsto \{(\Cst(0),\, 0)\}\\
\Var & \mapsto \{(\Cst(1),\, 0)\}\\
\Add(e_1,e_2) & \mapsto
                \big\{
                     \big(\Add(r_1,r_2),\, 2+c_1+c_2\big) \\&\phantom{\mapsto\quad}
                     \,\big|\,
                     (r_1,c_1)\in \funvec_{{\tt diff}}(e_1),\, (r_2,c_2)\in \funvec_{{\tt diff}}(e_2)
                \big\}\\
             \Mul(e_1,e_2) & \mapsto
                \big\{
                     \big(\Add(\Mul(r_1,e_2),\Mul(e_1,r_2)),\, 2+c_1+c_2\big) \\&\phantom{\mapsto\quad}
                     \,\big|\,
                     (r_1,c_1)\in \funvec_{{\tt diff}}(e_1),\, (r_2,c_2)\in \funvec_{{\tt diff}}(e_2)
                \big\}
        \end{align*}
        \vspace{-1.5em}
        \begin{align*}
          \big(\Phi(\funvec)\big)_{{\tt eval}} \colon
             \semtypes{{\tt Expr}}\times\zz & \to \pp\costm(\zz)\\
(\Cst(c),\,x) & \mapsto \{(c,\, 0)\}\\
(\Var,\,x) & \mapsto \{(x,\, 0)\}\\
(\Add(e_1,e_2),\,x) & \mapsto
                \big\{
                     (r_1+r_2,\, 2+c_1+c_2+\ell(r_1)+\ell(r_2)) \\&\phantom{\mapsto\quad}
                     \,\big|\,
                     (r_1,c_1)\in \funvec_{{\tt eval}}(e_1),\, (r_2,c_2)\in \funvec_{{\tt eval}}(e_2)
                \big\}\\
             (\Mul(e_1,e_2),\,x) & \mapsto
                \big\{
                     (r_1 \cdot r_2,\, 2+c_1+c_2+\ell(r_1)\cdot\ell(r_2)) \\&\phantom{\mapsto\quad}
                      \,\big|\,
                     (r_1,c_1)\in \funvec_{{\tt eval}}(e_1),\, (r_2,c_2)\in \funvec_{{\tt eval}}(e_2)
                \big\}
        \end{align*}
        \vspace{-1.5em}
        \begin{align*}
          \big(\Phi(\funvec)\big)_{{\tt diffeval}} \colon
             \semtypes{{\tt Expr}}\times\zz & \to \pp\costm(\zz)\\[-2pt]
(e,\,x) &\mapsto
               \left\{
                    (r,c_1+c_2)
                      \;\middle|\;
                      {\tiny\begin{matrix*}[l]
                        (e',c_1)\in \funvec_{{\tt diff}}(e),\\
                        (r,c_2)\in \funvec_{{\tt eval}}(e',x)
                      \end{matrix*}}
               \right\}
        \end{align*}
\end{tcolorbox}
\phantomcaption
\label{subfig:example-prog-concrete-operator}
\end{subfigure}\vspace{.04em} \noindent
\begin{subfigure}[t]{.66\textwidth} \begin{tcolorbox}[abstractpanel]
\textbf{(c) Abstract operator
(size abstraction)}
\scriptsize
$$\Phi^\sharp = \semop{\Prog}^{\sharp_{sz},\pp\costm} \in \End\Big(\prod_{f\in\Funs}\aMem_{\Vars_{f,in}} \to \pp\costm(\aMem_{\Vars_{f,out}})\Big)$$
$\big(\Phi^\sharp(\funvec^\sharp)\big)_{{\tt diff}} \colon \zz^2 \to \pp\costm(\zz^2)$
        \begin{align*}
           & \begin{pmatrix} n \\ l \end{pmatrix}
\mapsto \ite\left(n=1~\mathrm{and}~l\geq 0,\,
                    \Big\{\Big(\begin{pmatrix} 1 \\ 1 \end{pmatrix},\, 0\Big)\Big\},\,
                    \varnothing
                \right)\\
& \cup\bigcup
\left\{\;
                     \begin{Bmatrix*}[l]
\left(\begin{pmatrix} 1+n'_1+n'_2\\ \max(l'_1,l'_2) \end{pmatrix},\, 2+c_1+c_2\right), \\[8pt]
\left(\begin{pmatrix} 3+n'_1+n'_2+n_1+n_2\\ \max(l'_1,l'_2,l_1,l_2) \end{pmatrix},\, 2+c_1+c_2\right)
                     \end{Bmatrix*}
\;\middle|\;
\begin{matrix*}[l]
                     n_1,n_2\geq 1,\; n = 1+n_1+n_2,\\
                     l_1,l_2\geq 0,\; l = \max(l_1,l_2),\\
                     \big((n'_1,l'_1), c_1\big)\in \funvec^\sharp_{{\tt diff}}(n_1,l_1),\\
                     \big((n'_2,l'_2), c_2\big)\in \funvec^\sharp_{{\tt diff}}(n_2,l_2)
                     \end{matrix*}
                     \,
              \right\}
        \end{align*}

        $\big(\Phi^\sharp(\funvec^\sharp)\big)_{{\tt eval}} \colon \zz^2\times\zz \to \pp\costm(\zz)$
        \begin{align*}
           & \left(\begin{pmatrix} n \\ l \end{pmatrix},\,x\right)
\mapsto \ite\Big(n=1~\mathrm{and}~l\geq 0,\, \ite\Big(l=0,\,
\big\{(x,0)\big\},\,
\big\{(r,0)\,|\,\ell(r)=l\big\}
                 \Big),\,
                \varnothing
                \Big)\\
& \cup\bigcup
\left\{\;
                     \begin{Bmatrix*}[l]
\big(r_1 + r_2,\, 2+c_1+c_2+ \ell(r_1)+\ell(r_2)\big), \\[6pt]
\big(r_1 \cdot r_2,\, 2+c_1+c_2+\ell(r_1)\cdot\ell(r_2)\big) \\
                     \end{Bmatrix*}
\,\middle|\,
{\tiny\begin{matrix*}[l]
                     n_1,n_2\geq 1,\; n = 1+n_1+n_2,\\
                     l_1,l_2\geq 0,\; l = \max(l_1,l_2),\\
                     \big(r_1, c_1\big)\in \funvec^\sharp_{{\tt eval}}(n_1,l_1),\\
                     \big(r_2, c_2\big)\in \funvec^\sharp_{{\tt eval}}(n_2,l_2)
                     \end{matrix*}}
                     \,
              \right\}
        \end{align*}

        $\big(\Phi^\sharp(\funvec^\sharp)\big)_{{\tt diffeval}} \colon \zz^2\times\zz \to \pp\costm(\zz)$
        \begin{align*}
          \big((n,l),\,x\big)
          & \mapsto
               \left\{
                    (r,\,c_1+c_2)
                      \;\middle|\;
                      {\tiny\begin{matrix*}[l]
                        ((n',l'),\,c_1)\in\funvec_{{\tt diff}}\big((n,l)\big),\\
                        (r,c_2)\in\funvec_{{\tt eval}}\big((n',l'),\,x\big)
                      \end{matrix*}}
               \right\}
\end{align*}
\end{tcolorbox}
\phantomcaption
\label{subfig:example-prog-operator-size-abs}
\end{subfigure}\PanelSep \begin{subfigure}[t]{.33\textwidth} \begin{tcolorbox}[abstractpanel]
\textbf{(d) CHC representation of $\Phi^\sharp$}
\begin{lstlisting}[basicstyle=\ttfamily\tiny,columns=fullflexible,mathescape=true]
diff(1, L, 1, 1, 0) :- L >= 0.
diff(N, L, Nd, Ld, C) :-
  N = 1 + N1 + N2, N1 >= 1, N2 >= 1,
  L = max(L1, L2), L1 >= 0, L2 >= 0,
  diff(N1, L1,  Nd1, Ld1, C1),
  diff(N2, L2,  Nd2, Ld2, C2),
  ( (Nd = 1 + Nd1 + Nd2,
     Ld = max(Ld1, Ld2))
  ;
    (Nd = 3 + Nd1 + Nd2 + N1 + N2,
    Ld = max(Ld1, Ld2, L1, L2) ),
  C = 2 + C1 + C2.

eval(1, L, X, R, 0) :- L >= 0, $\ell$(R) = L.
eval(1, 0, X, X, 0).
eval(N, L, X, R, C) :-
  N = 1 + N1 + N2, N1 >= 1, N2 >= 1,
  L = max(L1, L2), L1 >= 0, L2 >= 0,
  eval(N1, L1, R1, C1),
  eval(N2, L2, R2, C2),
  ( (R = R1 + R2,
     C = $\ell$(R1) + $\ell$(R2) + 2 + C1 + C2)
  ;
    (R = R1 * R2,
     C = $\ell$(R1) * $\ell$(R2) + 2 + C1 + C2) ).

diffeval(N, L, X, R, C) :-
  diff(N, L, Nd, Ld, C1),
  eval(Nd, Ld, X, R, C2),
  C = C1 + C2.
\end{lstlisting}
\end{tcolorbox}
\phantomcaption
\label{subfig:example-prog-chc-size-abs}
\end{subfigure}\vspace{-2em}
\caption{
Illustrative example of \emph{abstract compilation} by \emph{abstraction of operator semantics}.\\
The selected operator abstraction, \emph{size abstraction}, is the first step
of our recurrence extraction pipeline.
}
\label{fig:motivating-example}
\end{figure*}
 \afterpage{\FloatBarrier}

\vspace{-.25em}
\subsubsection*{First Operator Abstraction: Size Abstraction}
Fig.~\ref{subfig:example-prog-operator-size-abs} shows the operator
$\Phi^\sharp$ resulting from an \emph{abstract compilation} of $\Phi$
that performs a \emph{size abstraction},
i.e. which transforms a program operating on concrete data
into a program operating over purely numerical quantities
describing this data (called \emph{metrics} on the original data,
either user-provided or inferred).
In our example, each expression $e\in \semtypes{{\tt Expr}}$ is
replaced
by a pair representing its \emph{number of nodes} $nb(e)$
and the \emph{maximum bitlength of its \Cst nodes} $\ell_{\max}(e)$.
The operator $\Phi^\sharp$ may be interpreted as a
(non-deterministic) numerical program, retaining the recursive
structure of the original program; it can be analysed by
techniques whose scope is restricted to numerical programs.
Note that many choices of syntactic representation are valid for such abstract programs:
constrained Horn clauses (CHCs~\cite{anal-peval-horn-verif-2021-tplp-short}), recursive logic formulae in SMT format,
syntax trees of numerical programs,
recurrence equations on set-valued functions, etc.
This is an implementation decision, that can be governed
by convenience, and available tools to perform desired downstream analyses
(e.g. verification of invariants by SMT solvers, inference of them by
abstract interpretation with numerical domains, search of candidate
invariants through executions, etc.).
Fig.~\ref{subfig:example-prog-chc-size-abs} shows one possible
syntactic encoding of $\Phi^\sharp$ using CHCs (a convenient format
for expressing recursive relations, which can be \emph{executed} under some conditions).

\vspace{-.25em}
\subsubsection*{Second Operator Abstraction: Interval Effect}
While numerical programs are amenable to various invariant inference techniques,
it is typically difficult to obtain highly non-linear invariants, and this example
requires it.
To obtain these, our approach supports a further operator abstraction,
where the powerset monad $\pp(-)$ is replaced by an interval \emph{abstract monad}.
In our example, the operator of Fig.~\ref{subfig:example-prog-operator-size-abs}
is further abstracted into an operator in
$\End(\prod_{f\in\Funs}\aMem_{\Vars_{f,in}} \to \ii\costm(\aMem_{\Vars_{f,out}}))$,
which may be interpreted as a \emph{system of recurrence equations of interval-valued functions}
(where interval-valued functions can themselves be split into lower and upper bound functions).
For instance, the system extracted from our example contains the following
equation\footnote{\label{footnote:fasible-vectors}
  As an interesting feature, consider the guards of this equation, and the parity constraints
  that occur: optimal size abstraction ensures that vectors of sizes are ``feasible'', i.e. that
  we do not consider size vectors if they do not actually arise from possible data.
In practice, this feasibility relation is overapproximated by a first analysis pass over
  the metric definitions.
} on the upper bound of the \emph{cost} of {{\tt eval}},
which also depends on bounds on the values produced by {{\tt eval}}.
\vspace{-.3em}{\scriptsize
\begin{equation*}(f_{{\tt eval}})_c^{ub}(n,l,x) = \begin{cases}
    -\infty & \text{if $n < 1$ or $l < 0$ (or $n$ even)}\\
    0 & \text{if $n = 1$ and $l \geq 0$}\\
\raisebox{1.8em}{\(\displaystyle
      \sup_{\substack{1 \leq n_1,n_2 \leq n-1\\ 1+n_1+n_2=n \\0\leq l_1,l_2\leq l\\\max(l_1,l_2)=l}}
    \)}
\begin{pmatrix*}[l]
          2 + (f_{{\tt eval}})_c^{ub}(n_1,l_1,x)+(f_{{\tt eval}})_c^{ub}(n_2,l_2,x) \\
          \phantom{2} {}+{}
           {\displaystyle \max_{\substack{\Diamond\in\{+,\times\},\\b\in\{lb,ub\}}}\Big(
             \ell\big((f_{{\tt eval}})_v^{b}(n_1,l_1,x)\big) ~\Diamond~ \ell\big((f_{{\tt eval}})_v^{b}(n_2,l_2,x)\big)
           \Big)}
         \end{pmatrix*}
& \text{if $n\geq 3$ and $l \geq 0$ (and $n$ odd)}
  \end{cases}
\vspace{-.4em}
\end{equation*}
}
\noindent Of course, such systems of equations display challenging features, such as the parametric
maximisation $\sup_{...}$, and we cannot simply compute their \emph{exact} solution by using
traditional Computer Algebra Systems (CAS). However, these equations can be tackled by using \emph{postfixpoint-based}
techniques (e.g.~\cite{order-recsolv-sas24,absfun-STTT26}) which can be used to infer and
verify bounds on the solution to these equations.
In that case, we can prove the following bounds (valid everywhere, and not merely asymptotically),
{
$(f_{{\tt diff}})^{ub}(n,l,x) \leq ((n^2, \max(1,l)), 2n)$, $(f_{{\tt eval}})^{ub}(n,l,x) \leq  (2^{n\cdot \max(l,\ell(x))},\, n^2 \cdot \max(l, \ell(x))^2)$,
and
$(f_{{\tt diffeval}})^{ub}(n,l) = (2^{n^2\cdot \max(l,\ell(x))},\, 2n + n^4 \cdot \max(l, \ell(x))^2 ))$.
}

Nevertheless, if the obtained equations cannot be solved, they can be further abstracted or
overapproximated, thus obtaining less precise but simpler equations: as in
classical abstract interpretation, the operator-based
viewpoint on abstraction of programs themselves allows analysis designers to explore a whole
landscape of semantic abstractions with different precision/complexity trade-offs,
but the framework itself does not commit to coarse approximation a priori.

\vspace{-1mm}
\section{Preliminaries}\label{sec:preliminaries}
\vspace{-1pt}

For brevity, the fundamental notions of order theory, including Galois
Connections (GCs, for short), are assumed and
recalled in Appendix~\ref{subsec:appendix-order}.
Section~\ref{subsec:preliminaries-order-theory} introduces only
fundamental intuition, notation, and a useful construction.
Section~\ref{subsec:preliminaries-categories} presents fundamental
notions of category theory, used throughout the paper. For clarity the
simplest
notions
are summarised in tables and presented in more detail in
Appendix~\ref{subsec:app-category-theory}.
Section~\ref{subsec:ADT-intepretation} recalls interpretation of ADTs
as initial algebras, and introduces \emph{catamorphic functions}.
Section~\ref{sec:boring-notation} adds standard notation.

\vspace{-2mm}
\subsection{Order Theory and Abstract Interpretation}\label{subsec:preliminaries-order-theory}We denote by $\Pos(X,Y)$ the set of all monotone functions from poset $X$ to poset $Y$,
by $\pleq$ the pointwise order, and by $\End(X)$ the poset $\Pos(X,X)$ ordered pointwise.
Unless said otherwise, products of posets are ordered coordinate-wise (product order),
and (po)sets of functions are ordered pointwise.
We denote GCs by $C\galois{\alpha}{\gamma} A$, with the left adjoint below,
going from left to right.

Broadly speaking, \emph{abstract interpretation}~\cite{Cousot21-book} is a theory for
reasoning about programs by computing sound abstractions of their
semantics (defined in a ``concrete'' lattice) via primitives performing
computations within a (simpler) ``abstract'' lattice. GCs offer a
principled approach for designing these primitives, as they provide a
notion of \emph{best} abstraction, rather than a collection of
\emph{sound} ones.
We recall the following construction of GC in function space, lifted from GCs of values.

\vspace*{-2.5mm}

\begin{example}
   \label{ex:GC-lift-2GCs}
   Let $C_1 \galois{\alpha_1}{\gamma_1} A_1$ and $C_2 \galois{\alpha_2}{\gamma_2} A_2$ be GCs.
   Then, $\Pos(C_1,C_2) \galois{\alpha_2\circ (-) \circ \gamma_1}{\gamma_2\circ (-) \circ \alpha_1} \Pos(A_1,A_2)$ is a GC.
When $C_1=C_2$ and $A_1=A_2$, we call this construction ``\End-lifting''.
\end{example}

\vspace{-3.5mm}
\subsection{Collections of Abstract Functions and Weak Notions of Categories}\label{subsec:preliminaries-categories}
\vspace*{-3.5mm}
\begin{longtable}{p{0.21\textwidth}p{0.74\textwidth}}
\toprule
  \textbf{Notion} & \textbf{Meaning}\\
  \toprule
  \endfirsthead
\toprule
  \textbf{Notion} & \textbf{Meaning} \\

  \toprule
  \endhead
\midrule
  \multicolumn{2}{r}{\textit{Continued on next page}} \\
  \endfoot
\bottomrule
  \endlastfoot
Magmoid
    & A structure \catex with objects $\Ob(\catex)$, homsets $\catex(X,Y)$ (sets of arrows, a.k.a.~morphisms), and composition $\circ_{X,Y,Z}\colon\catex(Y,Z)\times\catex(X,Y)\to\catex(X,Z)$. \\
  Unital magmoid
    & A magmoid with identity morphisms $\id_X \in \catex(X,X)$ for every object~$X$, satisfying $f\circ\id_X=f=\id_Y\circ f$. \\
  \scalebox{0.96}{Associative magmoid}
    & A magmoid satisfying $(h\circ g)\circ f = h\circ(g\circ f)$ for all compatible $f$, $g$, $h$. \\
  Category
    & A unital and associative magmoid. \\
  (Pre)functor
    & A prefunctor $F \colon \catex \to \catex'$ maps objects and arrows from a magmoid to another; a functor additionally preserves composition and identities. We add the prefix \emph{endo-} (e.g. \emph{endofunctor}) whenever $\catex=\catex'$. \\
  Prenatural transfor.
    & Between two prefunctors $F,G \colon \catex \to \catex'$, it is denoted as $\tau \colon F \Rightarrow G$. It is a collection $\tau_X \in \catex'(FX,GX)$ indexed by objects $X \in \Ob(\catex)$. \\
  \scalebox{0.9}{Natural transformation}
    & A prenatural transformation with $\tau_Y\circ Ff=Gf\circ\tau_X$ for every $f\in\catex(X,Y)$. \\
  \bottomrule
\end{longtable}
\vspace*{-1em}
{\tiny\phantom{.}}

\noindent Examples of categories relevant to this paper include:
  the category \Set of sets and (arbitrary) functions;
  the category \Pos of posets and monotone functions;
  the category \CPO of cpos and monotone functions;
  the category \CLat of complete lattices and \emph{monotone} functions;
  the category $\CLat_\sqcup$ of complete lattices and \emph{additive} functions;
  for a given (typed) programming language, its \emph{syntactic category} $\syntcat$, whose objects are the types in $\syntcat$, and whose arrows are (definable) functions between them.
We also later recall the definition of the \emph{Kleisli category} $\Kleisli_T$ for a monad.

\paragraph*{{``Category ${{}\galois{\alpha}{\gamma}{}}$ Magmoid''}}
Both categories and magmoids arise naturally in higher-order abstract interpretation, through the study of compositionality.
In classical denotational semantics, the structure of composition yields a \emph{category}, which enjoys the strong structural properties of associativity and unitality.
However, when abstracting this semantics, its strong structure may disappear entirely, or relations may hold only up to inequality; we obtain a \emph{magmoid}.

To be more precise, consider a ``higher-order concrete domain'' $(C,\leq)$, whose values are thought of as functions, equipped with a composition operation $\circ :C\times C \to C$ monotone in both arguments, associative, and unital.
Then, for a Galois connection $(C,\leq)\galois{\alpha}{\gamma}(A,\sqsubseteq)$, the \emph{best} ``abstract composition'' operation on $A$, given by $g \circ^\sharp f \coloneqq \alpha(\gamma(g) \circ \gamma(f))$, is in general \emph{non}-associative and \emph{non}-unital.
Similarly, it is well-known that the best abstraction of a composition is not the composition of best abstractions: we have only $\alpha(g \circ f)\sqsubseteq\alpha(g) \circ^\sharp \alpha(f)$.

\paragraph*{Monadic composition}

Given a $T:\catex\to\CLat$ representing concrete or abstract computations, one can perform higher-order abstract interpretation in function spaces of the form ${(T(X)\to T(Y), \pleq)}$.
We will work instead with function spaces of the form $(X\to T(Y), \pleq)$, as they offer design and precision advantages, at the cost of requiring a more sophisticated, \emph{monadic} form of composition.

\begin{definition}[Premonad]\label{def:premonad}
  A~\emph{premonad} on a category~$\catex$ is a triple
  $(T(-),\eta_T,\circ_T)$, where
  \begin{enumerate}
    \item $T(-) \colon \Ob(\catex)\to\Ob(\catex)$ is a map (a mapping-on-objects)
    \item $(\eta_{T,X} \colon \catex(X,TX))$
    is a collection of maps indexed by $X\in\Ob(\catex)$ (the
    \emph{unit} of the premonad),
    \item $(\circ_{T,(X,Y,Z)} \colon \catex(Y,TZ)\times\catex(X,TY)\to\catex(X,TZ))$
    is a collection indexed by~${X,\!Y\!,Z\in\Ob(\catex)}$,
  \end{enumerate}
such that
(4) $(g \circ_T f) \circ u = g \circ_T (f \circ u)$ for every $u\in\catex(W,X)$, $f\in\catex(X,TY)$, $g\in\catex(Y,TZ)$,
and
(5) $\eta_{T,Y}\circ_T f = f = f\circ_T \eta_{T,X}$ for all $f\in\catex(X,TY)$
  (i.e.~$\circ_T$ is unital, with units given by $\eta_T$).

\end{definition}

\begin{definition}[Kleisli magmoid]\label{def:kleisli-magmoid}
  For a premonad $(T(-),\eta_T,\circ_T)$ on a category $\catex$, we denote by $\widetilde{\Kleisli_T}$
  the \emph{Kleisli (unital) magmoid} associated with $T$, where $\Ob(\widetilde{\Kleisli_T})=\Ob(\catex)$, $\widetilde{\Kleisli_T}(X,Y)=\catex(X,TY)$, and the composition is given by $\circ_T$.
  We thus refer to $\circ_T$ as the \emph{Kleisli composition} of the premonad.
\end{definition}

\begin{definition}[Monad]\label{def:monad}
  A \emph{monad} $(T(-), \eta_T, \circ_T)$ is a premonad with associative Kleisli composition~$\circ_T$.
  Equivalently, it is a premonad such that the Kleisli \magmoid is a category.
\end{definition}

When $T$ is a monad, we call $\widetilde{\Kleisli_T}$ the \emph{Kleisli category} associated with $T$,
and write $\Kleisli_T$ instead (dropping the tilde).
This ``composition first'' definition of monad is equivalent to the classical definition via natural transformations, see Appendix~\ref{subsec:proof-monad-by-kleisli-equivalent-to-monad-by-mult} for a proof.
Many other equivalent definitions exist.
In our context, it is more convenient to start from
the Kleisli composition $\circ_T$ as the primary operation, as it is a central
primitive when designing higher-order abstract domains.

\subsubsection*{Order enrichment}

In the context of program analysis, it is useful to be able to \emph{compare
morphisms}, associating to $f \leq g$ meanings such as ``$f$ is more
precise than $g$'' ($f$ is lower), ``$g$ is more abstract than $f$'', or
roughly ``$g$ induces more behaviours than $f$''.
In other words, it is useful to add a poset structure on homsets; a process
known as \emph{order-enrichment}.

\begin{definition}[Order-enrichment]\label{def:order-enrichment}
  An \emph{ordered-enriched}...\\
  \begin{tabular}{rcp{0.78\linewidth}}
    magmoid & is
      & a magmoid~$\catex$ together with a poset structure $(\catex(X,Y),\leq_{\catex(X,Y)})$ on each homset, such that each composition operation $\circ$ is monotone in both inputs;\\
    category & is & a category that is order-enriched as a magmoid;\\
    (pre)functor & is & a (pre)functor $F \colon \catex \to \catex'$ between two ordered-enriched magmoids satisfying $f \leq g \implies F(f) \leq F(g)$, i.e. $F(-)$ is monotone on arrows;\\
    (pre)monad & is & a (pre)monad $(T(-),\eta_T,\circ_T)$ with an ordered-enriched Kleisli magmoid~$\widetilde{\Kleisli_T}$. In particular, $\circ_T$ is monotone in both arguments.
  \end{tabular}
\end{definition}

Our notion of order-enriched (pre)monad is a generalisation of the homonymous notion from~\cite{Hasuo15,GoncharovSchroder13}.
As in these works, we do not require the base $T(-)$ to be order-enriched.
In fact, Definition~\ref{def:premonad} just requires $T(-)$ to be a mapping on objects.

\begin{example}[Powerset]\label{example:powerset}
  Consider the Powerset monad $(\pp(-), \eta_\pp, \circ_\pp)$ on~\Set,
  with unit $\eta_\pp$ given by singletons $\eta_{\pp,X}(x) \coloneqq \{x\}$, and
  Kleisli composition $g \circ_\pp f \coloneqq x \mapsto \bigcup_{y\in f(x)} g(y)$
  for every ${f \in \Set(X,\pp(Y))}$ and $g \in \Set(Y, \pp(Z))$.
Its Kleisli category $\Kleisli_\pp$ can be enriched with the pointwise order $\dot\subseteq$
  defined from the complete lattice structure $(\pp(-), \subseteq)$:
  given $f,g \in \Set(X,\pp(Y))$, $f \mathbin{\dot\subseteq} g$
  whenever $f(x) \subseteq g(x)$ for all $x \in X$; which we can see as ``$f$ being more precise than $g$''.
\end{example}

  As observed by~\cite{Jacobs03}, when $T:\Ob(\Set)\to\Ob(\Set)$, a quick way to provide an order
  structure on all homsets $\Set(X,TY)$ consists in selecting a poset
  structure $\leq_{TY}$ on all images $TY$, and to extend this order
  pointwise as $\leq_{\Set(X,TY)}:=\pleq_{TY}$.
In this situation, we will say that $(T(-),\eta_T,\circ_T)$ is a
  premonad that is
  \emph{pointwise order-enriched via a poset structure
  on objects}.
To express the fact that all chosen orders $(TY,\leq_{TY})$ have additional
  properties we will simply say so directly, e.g.
  ``a premonad pointwise order-enriched via a \emph{complete lattice}
  structure on $T$-objects''.
For simplicity, in this paper, we will only consider such kind of
  (pointwise) order-enrichment on Kleisli magmoids.

\subsection{Algebraic Data Types, their Initial Algebra Semantics, and Catamorphisms}
\label{subsec:ADT-intepretation}

Throughout the paper, we write $\Types$ for a set of \emph{type names}, $\BasicTypes \subseteq \Types$ for a set of \emph{basic type names} (e.g., {\textbf{int}})
and $\Conss$ for a set of \emph{constructor names}.
Every basic type $\tau \in \BasicTypes$ comes with a preselected denotation
$\semtypes{\tau}$ (e.g., $\semtypes{\textbf{int}} = \zz$).
For non-basic types from $\Types \setminus \BasicTypes$,
the denotation $\semtypes{-}$ will be given via algebraic data types (ADTs).
We now recall the standard initial-algebra approach to semantics of ADTs,
in the context of the category~\Set for simplicity.

\subsubsection*{Polynomial endofunctors corresponding to ADTs}
A collection of (possibly mutually recursive) ADTs is a tuple $(\tau_1,\dots,\tau_n)$
of type names from $\Types \setminus \BasicTypes$, together with a
function {\tt cnstr} assigning to each $\tau_i$ a finite set of \emph{constructors}
of the form $(C,\tau_{1}',\dots,\tau_{k}')$, with $C \in \Conss$ and each $\tau_j' \in \{\tau_1,\dots,\tau_n\} \cup \BasicTypes$. A constructor name $C$ may appear in at most one of the sets ${\tt cnstr}(\tau_i)$, and there in at most one constructor.
This syntactic definition of ADTs induces an endofunctor $F:\Set^n\to\Set^n$,
whose components have the form
$F_i(X_1,\dots,X_n)=\sum_{c\in{\tt cnstr(t_i)}}\prod_{j=1}^{k_c}A_{c,j}$,
where, for $c = (C,\tau_1',\dots,\tau_{k_c}')$:
if $\tau_{j}' = \tau_\ell$ for some $\ell \in [1..n]$,
then $A_{c,j} = X_\ell$;
otherwise $\tau_{j}'$ is a basic type name and $A_{c,j} = \semtypes{\tau_j'}$.
Here, $\Sigma$ and $\Pi$ respectively corresponds to disjoint unions
and Cartesian products (or more generally coproducts and products in a
category other than \Set). Intuitively, sums models constructor
choice, products model constructor arguments, and occurences of the
$X_\ell$ model recursion.

\begin{example}
  The  type of list of integers given by ${\tt {\bf type}~ListInt~=~\{~\Nil~|~\Cons({\bf int}, ListInt)~\}}$
  induces the functor $F \colon X\to 1 + \zz\times X$.
  The two components of the coproduct correspond to the
  $\Nil$ and $\Cons$ constructors, respectively.
\end{example}

\subsubsection*{Initial algebra semantics of ADTs}
Endofunctors $F:\Set^n\to\Set^n$ induced by syntactic definitions of ADTs
are an instance of the notion of (multivariate) \emph{polynomial
functors} (see, e.g.,~\cite{PierreLMCS2014}).
They allow to define the denotation $\semtypes{-}$ of ADTs as initial algebras.

\begin{definition}[$F$-algebras]
  Given a category $\catex$ and an endofunctor $F:\catex\to\catex$,
  an $F$-algebra is a pair $(A,a)$ where $A$ is an object of $\catex$
  and $a$ is a morphism $F(A) \to A$.
\end{definition}

\begin{example}\label{ex:sum-algebra}
  For \ListInt and its corresponding functor $F \colon X\to 1 + \zz\times X$,
  an example of $F$-algebra is given by $(\zz, {{\tt sum}})$
  where
  ${{\tt sum}} \coloneqq 1 + \zz \times \zz \to \zz$ is defined as
  ${{\tt sum}}({{\tt inl}}()) \coloneqq 0$ and ${{\tt sum}}({{\tt inr}}(x,s)) \coloneqq x+s$.
  (where we write ${{\tt inl}}$ and ${{\tt inr}}$ for the coproduct
  injections).
\end{example}

\begin{wrapfigure}[4]{r}{0.15\textwidth}
  \vspace{-15pt} \begin{tikzcd}
    FA \arrow[r, "Ff"]\arrow[d, "a"] & Fb \arrow[d, "b"]\\
    A \arrow[r, "f"] & B
  \end{tikzcd}
\end{wrapfigure}
The collection of $F$-algebras form a category, where the morphisms
between $(A, a)$ and $(B, b)$ are maps $f \in \catex(A,B)$ such that
the diagram on the right commutes.
An initial algebra is an initial object in this category.
Explicitly, an initial algebra is an $F$-algebra
$(\mu F, \textit{in})$ such that for all $F$-algebras
$(A,a)$, there is a unique map
$h\in\catex(\mu F, A)$ such that
$h \circ \textit{in} = a \circ F(h)$.
The initial algebra of an
endofunctor is unique (up to unique isomorphism) whenever it exists.

\begin{proposition}[Lambek's Theorem~\cite{Lambek1968}]
If an endofunctor $F$ admits an initial algebra $\textit{in} \colon F(\mu F) \to \mu F$, then it is
  an isomorphism $\textit{in} \colon F(\mu F) \cong \mu F$.
\end{proposition}

Initial algebras provide semantics to ADTs: given $(\tau_1,\dots,\tau_n)$
and the corresponding endofunctor $F:\Set^n\to\Set^n$, we take
the initial algebra $(\mu F, in)$ and
define $\semtypes{\tau_i} \coloneqq (\mu F)_i$ for each $i\in[1..n]$,
where $(\mu F)_i$ is the $i$-th component of $\mu F$.
This is justified by the next proposition.

\begin{proposition}[\cite{LehmannSmyth1981algebraicsem}]
  Endofunctors $F$ induced by syntactic definitions of ADTs (and more generally polynomial
  functors $\Set^n\to\Set^n$) admit an initial algebra.
\end{proposition}

\begin{example}\label{ex:initial-algebra-listint}
  Considering again {\tt ListInt} and $F \colon X\to 1 + \zz\times X$,
  the initial $F$-algebra $(\mu F, in)$
  exists, and $\semtypes{\ListInt} \coloneqq \mu F$ corresponds to the expected
  interpretation of lists of integers, where we have
  by $in({{\tt inl}}()) = \Nil$ and $in({{\tt inr}}(x,l)) = \Cons(x,l)$.
\end{example}

\subsubsection*{Catamorphisms}
In addition to giving a way of defining the semantics to ADTs, $F$-algebras
can be interpreted as (generators of) inductive definitions.
We will say that a function is \emph{catamorphic} when it arises from
such an inductive definition.

\begin{definition}[Catamorphism~\cite{Bananas91}]
  Let $(\mu F, \textit{in})$ be the initial $F$-algebra.
For any $F$-algebra $(A,a)$, the unique map $\mu F\to A$ induced
  by the initiality of $\textit{in}:F(\mu F)\to\mu F$ is called the
  \emph{catamorphism} induced by $a$, and is denoted by
  $cata_a:\mu F\to A$.
\end{definition}

\begin{definition}[Catamorphic function]
  \label{defi:catamorphic-function}
  We say that a function $f:\mu F \to A$ is \emph{catamorphic}
  whenever there is an $F$-algebra $(A,a)$ such that
  $cata_a = f$.
In other words, a function $f:\mu F\to A$ is catamorphic whenever
  there exists a function $a:FA \to a$ with
  $f \circ in = a \circ F(f)$.
\end{definition}

\begin{example}
  The function ${{\tt sum}}:F\zz\to \zz$ from Example~\ref{ex:sum-algebra}
  induces a catamorphism $cata_{\tt sum}:\ListInt\to\zz$ that computes the sum of the elements of a list.

\end{example}

\subsection{Notation}\label{sec:boring-notation}

$\bb=\{\bot,\top\}$ denotes the Booleans,
$\zz$ the integers, and $\nn$ the non-negative integers, including $0$.
The function $\ite \colon \bb\times X\times X \to X$ is defined by
$\ite(\top,x,y)=x$ and $\ite(\bot,x,y)=y$.
Given a (pre)monad $(T,\eta_T,\circ_T)$ we define the shorthand
$\big(\letin{T}{x}{e_1}{e_2}\big) := \big((x\mapsto e_2) \circ_T e_1\big)$,
which we use in a right-associative way:
$\big(\letin{T}{x}{e_1}{\letin{T}{x}{e_2}{e_3}}\big) := \big(\letin{T}{x}{e_1}{(\letin{T}{x}{e_2}{e_3})}\big)$.
The subscript $T$ is omitted in the case of the identity monad. We denote the restriction of a function $f:X\to Y$ to a subdomain $X'\subseteq X$ by $\pi_{X'}(f):X'\to Y$.
We denote by $\pi_i(x_1,\ldots,x_n) = x_i$ the $i$-th projection of the tuple $(x_1,\ldots,x_n)$.
The disjoint union of sets is denoted by $\uplus$, or simply by $+$ when the context is clear.
For a map $m$, element $x$ and value $y$, $m[x \leftarrow y]$ denotes the map obtained by updating $m$ at $x$ to $y$, extending the domain and codomain if necessary.

\section{The \IR language}
\label{sec:ir-lang}

\begin{figure}[t]
  \begin{minipage}[][7.5cm]{0.49\linewidth}
    {\tt\footnotesize
    \begin{tabular}{@{}ll@{}}
Prog    $\Coloneqq$ (typedef)*(fundef)* & (programs)\\[5pt]
typedef $\Coloneqq$                              & (type \\
        \quad{\bf type} $\tau$ = \{                    & \phantom{(}definition)\\
        \qquad\phantom{($\mid$ }consdef ($\mid$ consdef)* &\\
        \quad\}
      &\\[5pt]
consdef $\Coloneqq$ $C$( (ty)* ) & (constructor)
      \\[5pt]
ty $\Coloneqq$ {\bf int} $\mid$ $\tau$ & (type lexemes)
      \\[5pt]
fundef $\Coloneqq$ & (function \\
        \quad{\bf def} $f$( ($X_{\text{arg}}$: ty)* ) -> ty \{ & \phantom{(}definition)\\
        \qquad(stmt)*;                      &\\
        \qquad{\bf return} ${X_{\text{res}}}$ &\\
        \quad\}
\end{tabular}
    }

    \vfill

    \hrule

    \vspace{3pt}

    \begin{center}
      \begin{minipage}{0.95\linewidth}
        \footnotesize
        $f$ is a user-defined function name from $\Funs$\\
        $\tau$ is a user-defined type name from $\Types \setminus \{\textbf{int}\}$\\
        $C$ is a user-defined constructor name from $\Conss$\\
        $X$, $X_{\text{arg}}$ and $X_{\text{res}}$ are program variables from $\Vars$\\
        $\Diamond_{\mathrm{a}} \in \{{}+{},{}-{},{}\cdot{}\}$,
        \
        $\Bowtie \in \{<,\leq,=,\geq,>\}$,
        \
        $\Diamond_{\mathrm{b}} \in \{\texttt{and},\texttt{or}\}$
      \end{minipage}
    \end{center}
  \end{minipage}\vrule{}\begin{minipage}[][7.5cm]{0.49\linewidth}{\tt\footnotesize
    \begin{tabular}{ll@{}}
stmt $\Coloneqq$
                  {\bf let} $X$: ty = expr          & (var.\ decl.)\\
      \quad$\mid$ $X$ = expr                        & (assignment)\\
      \quad$\mid$ {\bf if} cond {\bf then} \{ stmt \}
                                                    &(conditional)\\
      \phantom{\quad$\mid$ {\bf if} cond} {\bf else}\hspace{2.5pt} \{ stmt \}
                                                    &\\
      \quad$\mid$ {\bf match} expr \{               & (match)\\
        \qquad\quad{}mtcc ($\mid$ mtcc)*            &\\
        \qquad\}                                    &\\
      \quad$\mid$ stmt; stmt                        & (sequence)\\
      \quad$\mid$ {\bf skip}                        & (no-op)\\
      \color{azure}\quad$\mid$ {\bf assume}(cond) & \color{azure}(assume)
      \\[5pt]
mtcc $\Coloneqq$
        $C$( ($X$)* ) => stmt                  & (match case)\\[5pt]
expr $\Coloneqq$
                    $X$                                 & (variable)\\
      \quad$\mid$   $C$( (expr)* )                      & (construction)\\
      \quad$\mid$   $f$( (expr)* )                      & (function call)\\
      \quad$\mid$   expr $\Diamond_{\mathrm{a}}$ expr   & (arithm. operation)\\
      \quad$\mid$   $c$                                 & (constant $c \in \zz$)\\
\color{azure}\quad$\mid$   ($\top$ : ty)        & \color{azure}(top)
      \\[5pt]
cond $\Coloneqq$
                    expr $\Bowtie$ expr                 & (comparison)\\
      \quad$\mid$   cond $\Diamond_{\mathrm{b}}$ cond   & (Boolean operation)\\
      \quad$\mid$   {\bf true} $\mid$ {\bf false}       & (Boolean constants)
    \end{tabular}
    }
    \vfill
  \end{minipage}
  \caption{Syntax of a our \IR language. Elements in {\color{azure}azure} allow for non-deterministic functions.}
  \label{fig:syntax-ir-num-det}
\vspace{-4mm}
\end{figure}

Fig.~\ref{fig:syntax-ir-num-det}
introduces the syntax of \IR,
a simple imperative-style language with user-defined algebraic data types
and recursive functions as the only form of iteration (for now).
Throughout the paper, this language serves as a running example for our notions of operator semantics and abstract compilation.
Its recursive functions provide a natural setting in which to introduce operator semantics, while its algebraic data types allow us to present the main application of abstract compilation studied in this paper: the compilation of programs manipulating symbolic values into programs manipulating purely numerical quantities describing these values.

We interpret \IR as a typed, lexically scoped, store-manipulating language.
The type system is standard; we omit the inference rules and simply assume that every
function, expression, and variable in a program has an associated type,
and that every program is well-typed.
Our only primitive type is \textbf{int}: there are boolean expressions, but no boolean variables.
Each program is also equipped with a finite set \Locs of program points (or locations), placed immediately before and
after each statement: these are called the input and output locations of that statement, respectively.
Expressions and conditions inherit the input location of their enclosing
statement. Program points $\ell\in\Locs$ are often omitted for readability (this is the case in~Fig.~\ref{fig:syntax-ir-num-det}, see the semantics of Fig.~\ref{fig:opsem-ir-num-monadic} for some occurrences).
Finally, we assume pattern matching (deconstruction) to be exhaustive: every \textbf{match} contains exactly one case for each constructor of the matched type.

Let $\Types$ denote the finite collection of types. Given a finite set of variables $\Vars$ and a
type $t\in \Types$, we denote by $\Vars_t$ the set of variables of type $t$, so that we have the
partition $\Vars = \sum_{t\in\Types} \Vars_t$.
The corresponding set of memories (or stores) is denoted
$\Mem_\Vars := \prod_{t\in\Types}\big(\Vars_t \to \semtypes{t}\big)$, where
$\semtypes{t}$ is the set of values of type $t$ (see Section~\ref{subsec:ADT-intepretation}).
Since the language is lexically scoped, each $\ell\in\Locs$
determines a set of variables in scope, denoted $\Vars_\ell$. Likewise, each function $f$ has
associated sets of input and output variables, denoted $\Vars_{f,in}$ and $\Vars_{f,out}(=\{X_{res}\})$ respectively.

In classical abstract interpretation fashion, \IR will be extended with non-deterministic primitives
(\texttt{top}, \texttt{assume}) for modelling uncertainty (e.g. unknown inputs, underspecified
functions, user-provided assumptions, etc.).

\vspace*{-2mm}
\section{Operator Semantics: Programs $\equiv$ Operators}
\label{sec:operator-semantics}

This paper brings forward the idea of seeing programs as monotone operators on lattices.
To start, consider the simple denotational semantics approach
that interprets each function
${f \in \Funs}$
as a partial function
$\sem{f}:\Mem_{\Vars_{f,in}} \rightharpoonup \Mem_{\Vars_{f,out}}$
or, equivalently, as a total function into
$\big(\Mem_{\Vars_{f,out}}\big)_\bot := \Mem_{\Vars_{f,out}} \cup \{\bot\}$,
where $\bot$ is a value representing ``{\tt undefined}''.
Because functions may be recursive, these denotations cannot be defined independently.
Instead, one first constructs an operator
\vspace*{-1mm}
\begin{equation}\label{eq:operator-semantics}
  \Phi_{{\tt Prog}} :
    \Big(\prod_{f\in \Funs} \Mem_{\Vars_{f, in}} \to \big(\Mem_{\Vars_{f, out}}\big)_\bot\Big)
    \to
    \Big(\prod_{f\in \Funs} \Mem_{\Vars_{f, in}} \to \big(\Mem_{\Vars_{f, out}}\big)_\bot\Big),
\vspace*{-2mm}\end{equation}that performs a single recursion step (a single semantic unfolding of all recursive definitions).
Finally, classical denotational semantics proceeds by taking a suitable fixed point of
$\Phi_{{\tt Prog}}$. For example, under the usual information ordering on partial functions,
$\Phi_{{\tt Prog}}$ is monotone, and its least fixed point gives the meaning of the program.

Let us stop one step earlier: what we call \emph{operator semantics} simply consists in
\emph{reifying} this operator (and generalisations of it) into a semantic object, interesting in its
own right, instead of forgetting about it after taking the fixpoint.
In other words, the operator semantics $\semop{{\tt Prog}} := \Phi_{{\tt Prog}}$ is a
\emph{fixpoint-delaying semantics}: one recovers classical denotational semantics by taking~$\lfp~\semop{{\tt Prog}}$.

\vspace*{-1mm}
\subsection{Operator Semantics for the \IR language}

{
\begin{figure}\footnotesize \begin{align*}
    \semop{{\tt Prog}}\, &
\in \End\Big(\prod_{f\in \Funs} \Mem_{\Vars_{f, in}} \to T\big(\Mem_{\Vars_{f, out}}\big)\Big)
       \\[5pt]
& = \funvec
\mapsto \left(\argvec \mapsto
            \Big(\big(\eta_T \circ \pi_{\{X_{\text{res}}\}}\big)\circ_T \semstmt{f_{\text{body}}}(\funvec)\Big)(\argvec)
         \right)_{f \in \Funs}
\end{align*}
  Above, we assume that {\tt Prog} contains a definition {\tt def $f$(...) -> ty \{ $f_{\text{body}}$; \textbf{return} $X_{\text{res}}$  \} } for every for ${\tt f} \in \Funs$.

  \vspace{5pt}
  \hrule height 0.08em
  \smallskip
  \hrule height 0.08em

  \begin{align*}
  \semstmt{^{\ell_{in}}{\tt stmt}^{\ell_{out}}}\,&
       \in \Big(\prod_{f\in \Funs} \Mem_{\Vars_{f, in}} \to T\big(\Mem_{\Vars_{f, out}}\big)\Big)
       \to \Big(\Mem_{\Vars_{\ell_{in}}} \to T\big(\Mem_{\Vars_{\ell_{out}}}\big)\Big)\\
\semstmt{{\bf skip}}(\funvec)
    & = m \mapsto
      \eta_T(m) \\
  \semstmt{{\tt s1 ; s2}}(\funvec)
    & = m \mapsto
      \big(\semstmt{{\tt s2}}(\funvec)\circ_T \semstmt{{\tt s1}}(\funvec)\big)(m) \\
  \semstmt{ X = {\tt expr}}(\funvec)
    & = m \mapsto
        \letin{T}
              {a}{\semexpr{{\tt expr}}(\funvec)(m)}
              {\eta_T\big(m[X \leftarrow a]\big)} \\
\semstmt{{\tt {\bf if}~cond~{\bf then}~\{ s1 \}~{\bf else}~\{ s2 \}}}(\funvec)
    & = m \mapsto
        \letin{T}
              {b}{\sembexp{{\tt cond}}(\funvec)(m)\\[-3pt]&\phantom{~=m\mapsto}}
              {\ite\big(b,\,\semstmt{{\tt s1}}(\funvec)(m),\,\semstmt{{\tt s2}}(\funvec)(m)\big)}
            \\
\widesemstmt{
    \begin{aligned}
          &\tt {\bf match}~expr~\{ \dots\\
          &\quad{}\mid C_i(X_{i1},\dots,X_{ik}) {\tt~=>~stmt_i}\,\\
          &\ \quad{}\dots\}
    \end{aligned}}\hspace{-5pt}(\funvec) & = m \mapsto \letin{T}{a}{\semexpr{{\tt expr}}}{
    \left(
      \begin{aligned}
        &\text{match $a$ with}
        \dots\\
        &\quad{}C_i(e_1,\dots,e_k)\ {\tt =>}\\
        &\qquad\semstmt{{\tt stmt_i}}(\funvec)(m[X_{ij} \gets e_j : j \in [1..k]])\\
        &\dots
      \end{aligned}
    \right)}
  \end{align*}

  \vspace{5pt}
  \hrule height 0.08em
  \smallskip
  \hrule height 0.08em

  \begin{align*}
  \semexpr{^\ell{\tt expr}}^{\tau} &
       \in \Big(\prod_{f\in \Funs} \Mem_{\Vars_{f, in}} \to T\big(\Mem_{\Vars_{f, out}}\big)\Big)
       \to \Big(\Mem_{\Vars_{\ell}} \to T\big(\semtypes{\tau}\big)\Big)\\[5pt]
\semexpr{{\tt X}}(\funvec)(m)
         & = \eta_T\big(m[{\tt X}]\big)\\
\semexpr{C({\tt expr_1, \dots, expr_k})}(\funvec)(m)
         & = \letin{T}{a_1}{{\semexpr{{\tt expr_1}}}(\funvec)(m)}
               {\dots~
                 \letin{T}{a_k}{{\semexpr{{\tt expr_k}}}(\funvec)(m)
                 \\[-3pt]&\phantom{={}}}
                {C(a_1, \dots, a_k)}}\\
\semexpr{{\tt f(expr_1, \dots, expr_k)}}(\funvec)(m)
            &=  \letin{T}{a_1}{{\semexpr{{\tt expr_1}}}(\funvec)(m)}
               {\dots~
                 \letin{T}{a_k}{{\semexpr{{\tt expr_k}}}(\funvec)(m)
                 \\[-3pt]&\phantom{={}}}
                {\funvec_{\tt f}(a_1, \dots, a_k)}}\\
\semexpr{{\tt expr_1}\, \Diamond_{\mathrm{a}}\, {\tt expr_2}}(\funvec)(m)
          &=  \letin{T}{a_1}{\semexpr{{\tt expr_1}}(\funvec)(m)}
             {\letin{T}{a_2}{\semexpr{{\tt expr_2}}(\funvec)(m)}}
             {\eta_T(a_1\,\Diamond_{\mathrm{a}}\,a_2)}\\
\semexpr{{c}}(\funvec)(m)
         & = \eta_T(c)
  \end{align*}

  \vspace{5pt}
  \hrule height 0.08em
  \smallskip
  \hrule height 0.08em

  \begin{align*}
  \sembexp{^\ell {\tt cond}} &
       \in \Big(\prod_{f\in \Funs} \Mem_{\Vars_{f, in}} \to T\big(\Mem_{\Vars_{f, out}}\big)\Big)
       \to \Big(\Mem_{\Vars_{\ell}} \to T\big(\bb\big)\Big)\\[5pt]
\sembexp{{\tt expr_1}\, \Bowtie\, {\tt expr_2}}(\funvec)(m)
          &=  \letin{T}{a_1}{\semexpr{{\tt expr_1}}(\funvec)(m)}
             {\letin{T}{a_2}{\semexpr{{\tt expr_2}}(\funvec)(m)}}
             {\eta_T(a_1\,\Bowtie\,a_2)}\\
       \sembexp{{\tt cond_1}\, \Diamond_{\mathrm{b}}\, {\tt cond_2}}(\funvec)(m)
          &=  \letin{T}{b_1}{\sembexp{{\tt cond_1}}(\funvec)(m)}
             {\letin{T}{b_2}{\sembexp{{\tt cond_2}}(\funvec)(m)}}
             {\eta_T(b_1\,\Diamond_{\mathrm{b}}\,b_2)}\\
       \sembexp{{b}}(\funvec)(m)
         & = \eta_T(b), \quad\ \text{for } b \in \{\textbf{true},\textbf{false}\}
  \end{align*}

  \vspace*{-2mm}
  \caption{Operator semantics of the \IR language of Fig.~\ref{fig:syntax-ir-num-det}
    for a generic monadic effect $T(-)$. For readability, we omit labels of program points, type arguments~$\tau$ in $\semexpr{\cdot}^\tau$, and projections of variables that go out of scope. Variable declarations are reflected by lexical scope
    and otherwise handled like variable assignment.
  }
  \label{fig:opsem-ir-num-monadic} \vspace*{2mm}
\end{figure}

 \begin{figure}\footnotesize\centering
  \begin{align*}
    \semstmt{{\color{azure}\textbf{assume}(cond)}}(\funvec)(m) & = \ite\big(\textbf{true} \in \sembexp{cond}(\funvec)(m),\ \{m\},\ \emptyset\big)\\
    \semexpr{{\color{azure}(\top : \tau)}}(\funvec)(m) & = \semtypes{\tau}
  \end{align*}
  \vspace*{-5mm}
  \caption{Extension of \IR with {\color{azure}non-deterministic primitives}, instantiating $T(-)$ as $\pp(-)$.}\label{fig:opsem-ir-num-ndet-extension}
  \vspace*{-3mm}
\end{figure}

 \afterpage{\FloatBarrier}
}

Consider the \IR language described in Fig.~\ref{fig:syntax-ir-num-det}, excluding for
the time being the elements in {\color{azure}azure}.
Its operator semantics for well-typed programs is given in Fig.~\ref{fig:opsem-ir-num-monadic}.
To achieve the aforementioned flexibility,
instead of defining this semantics as in~\eqref{eq:operator-semantics} using $\big(\Mem_{\Vars_{f,out}}\big)_\bot$, we parametrise it over
a monad $(T, \eta_T, \circ_T)$ in $\Set$.
The semantics in~\eqref{eq:operator-semantics} is then only
a special case, obtained by selecting the \emph{Maybe} monad $(-)_\bot$.
The interested reader can find this instantiation in Appendix~\ref{sec:operator-semantics-maybe-monad} (see Fig.~\ref{fig:opsem-ir-num-fail-extension}).

\begin{remark}\label{remark:maybe-monad}
  The Maybe monad~$((-)_\bot, \eta_\bot, \circ_\bot)$ is not merely a monad in $\Set$:
  it induces a natural order structure, namely
  the flat order structure on all
  objects $X_\bot$, defined by $X_\bot=\{\bot\}\uplus X$
  and $x \leq_{\text{flat}} y \iff (x = \bot \text{ or } x = y)$.
This order is inherited by the products and function spaces, and preserved by $\circ_{(-)_\bot}$, ensuring that $\semop{\Prog}$ is monotone.
Moreover, all $(X_{\bot},\leq_{\text{flat}})$ are CPOs, and this also lifts
  to products and piecewise orders.
  This guarantees that the classical denotational semantics $\lfp\semop{\Prog}$ can be defined
  (and that, furthermore, it can be computed by transfinite iteration).

  We also note that $(X_\bot, \leq_{\text{flat}})$ can be embedded in $(\pp(X),\subseteq)$
  via $\bot\mapsto \varnothing$, $x\mapsto\{x\}$, which allows moving from the Maybe monad
  to the (more general) Powerset monad;
  this is related to the collecting semantics approach in abstract interpretation~\cite{Cousot77,amato-collencting-2020}.
\end{remark}

Starting from the operator semantics of Fig.~\ref{fig:opsem-ir-num-monadic},
one can progressively strengthen the order-theoretical assumptions on the monad $(T,\eta_T,\circ_T)$,
not only to recover the classical denotational semantics via fixpoint, but also enable
invariant-based proof techniques via Knaster-Tarski.

\begin{proposition}[Order-enriched monads provide monotonicity]
  \label{prop:pos-monads-monot}
Assume that $T$ is an order-enriched monad,
  i.e. a monad such that the Kleisli homsets $\catex(X,TY)$ carry a
  poset structure, and that $\circ_T$ is monotone.
Then, $\semop{\Prog}$, $\semstmt{{\tt stmt}}$, $\semexpr{{\tt expr}}$, $\sembexp{{\tt cond}}$.
  are all monotone in $\funvec$, for the natural coordinate-wise and
  pointwise lifts of the order on Kleisli arrows, respectively for
  any well-typed program, statement and arithmetic/boolean expression.
\end{proposition}

\phantom{.}

\begin{proposition}[$\CPO$ structure provides $\lfp$]\label{prop:cpo-image-lfp}
  Suppose that $T$ is an order-enriched monad,
  and that the orders on Kleisli homsets are all CPOs.
Then, $\lfp\semop{\Prog}$ is well-defined.
Moreover, it can be obtained by transfinite iteration from $\bot$,
  that is, $\lfp\semop{\Prog} = \sup_{\alpha:\Ord}(\semop{\Prog})^\alpha(\bot)$.
\end{proposition}

\begin{proposition}[$\CLat$ structure provides Knaster-Tarski]\label{prop:clat-image-KT}
  Suppose that $T$ is an order-enriched monad
  and that the orders on Kleisli homsets are all complete lattices.
Then, the set~$\Fix(\semop{\Prog})$ of fixed-points of $\semop{\Prog}$ forms a complete lattice (for the inherited order).
Moreover, pre/postfixpoint provide lower/upper bounds on gfp/lfp:
  for every $\funvec$,
    $\funvec \pleq \semop{\Prog}(\funvec) \implies \funvec \pleq \gfp\semop{\Prog}$
  and
    $\semop{\Prog}(\funvec) \pleq \funvec \implies \lfp\semop{\Prog} \pleq \funvec$.
\end{proposition}

\begin{example}[Inductive invariants]
  Consider the Powerset monad in Example~\ref{example:powerset},
  and the program

  {\setlength{\tabcolsep}{3pt}\begin{tabular}{rl}
    $\Prog \coloneqq {}$
    {\tt {\bf def} f(x:{\bf int}) -> {\bf int} \{}& {\tt {\bf let} res:{\bf int} = 0;}\\
    & {\tt {\bf if} x < 0 {\bf then} \{\,res\,=\,x\,\} {\bf else} \{\,res\,=\,$3 \cdot{}$f(x${}-1$)${}+2$\,\};}\\
& {\tt {\bf return} res\,\}}.
  \end{tabular}}
Its operator semantics is
  $\semop{\Prog}(f) = n \mapsto \ite( n < 0, \{n\}, \{3y+2\,|\,y\in f(n-1)\} )$,
  for all $f:\zz\to\pp(\zz)$.
To prove that, for example, {\tt f} preserves parity (that is,
  $f(n)~\mathrm{mod}~2 = n~\mathrm{mod}~2$ for all $n$), one can define
  the candidate function $\candf \colon \zz\to\pp(\zz)$ which sends even
  (resp. odd) numbers to the set of all even (resp. odd) numbers:
  $\candf \coloneqq n \mapsto \ite(n \in 2\zz, 2\zz, 2\zz+1)$.
  We observe that this function a postfixpoint of $\semop{\Prog}$:
$\semop{\Prog}(\candf) = \big(n \mapsto \ite(n < 0, \{n\}, \ite(n \in 2\nn, {6\zz+2}, {6\zz+3}))\big)
   \mathbin{\dot\subseteq} \candf$.

  In other words, postfixpoints may be interpreted as a form of inductive invariants.
  Notice that proofs by postfixpoint techniques do not require an explicit induction scheme:
  this is already provided by Knaster-Tarski and the recursive flow encoded in the operator.
\end{example}

\subsubsection*{Non-determinism}
It is customary in the abstract interpretation literature to extend the languages under study with non-deterministic primitives to create or restrict execution paths.
Path creation is used to model underspecified functions or user-inputs, while path restriction models failing executions such as division by zero, or is used to refine the analysis through user-defined lemmas.

Thanks to support for monadic effects, it is simple to extend the operator semantics of \IR from~Fig.~\ref{fig:opsem-ir-num-monadic} to support these types of primitives.
Just select the Powerset monad $\pp(-)$ in Fig.~\ref{fig:opsem-ir-num-monadic}, and interpret the primitives highlighted in {\color{azure}azure} in Fig.~\ref{fig:syntax-ir-num-det} as in Fig.~\ref{fig:opsem-ir-num-ndet-extension}---the extension is compositional.
Propositions~\ref{prop:pos-monads-monot}~and~\ref{prop:clat-image-KT} still apply, by monotonicity of the interpretations
of the new primitives.

\begin{example}
  Intuitively, with the sequence of statements {\tt x = $(\top : \textbf{int})$; \textbf{assume}($0 \leq {}$x${}\leq 1$)}
  one obtains in the output $\pp(\Mem_{\Vars_{-}})$ two elements, one where {\tt x} is assigned the value $0$,
  and one where it is assigned the value $1$. A function exhibiting such statements is \emph{underspecified}.
  We can also easily model non-total operations such as integer division: an assignment {\tt y = $\frac{1}{\texttt{x}}$}
  can be expressed as {\tt {\bf if} x = $0$ {\bf then} \{ {\bf assume}({\bf false}) \} {\bf else} \{ y = $(\top : \textbf{int})$; {\bf assume }(x${}\cdot{}$y = 1) \}},
  which intuitively sends memories assigning {\tt x} to $0$ to $\emptyset \in \pp(\Mem_{\Vars_{-}})$; signaling a ``panic'' state.
\end{example}

\subsubsection*{Example: probabilistic semantics}
Another instantiation of this framework is given in Appendix~\ref{app:proba-sem},
which illustrates how the \emph{a priori} obvious choice of distribution monad $\dd_{=1}(-)$
must be extended to \emph{subdistributions} $\dd_{\leq 1}(-)$ to interpret recursion, and can be further extended to a monad $\ww_\infty(-)$ pointwise order-enriched by a $\CLat$ structure on objects,
allowing invariant-based reasoning.
An interesting property of that example is that monadic multiplication $\mu_{\ww_\infty}$ no longer coincides with order-theoretical join,
although composition $\circ_{\ww_\infty}$ is still additive in both inputs.

\subsection{Example: Cost Analysis via Monad Transformers}
\label{subsec:eff-cost-monad-transformers}

A limitation of the monadic approach to semantics and analysis of programs
is that monads do not always compose well:
given two functors $T:\catex\to\catex$ and $S:\catex\to\catex$, each equipped with monadic structure
as $(T,\eta_T,\circ_T)$, $(S,\eta_S,\circ_S)$, there may not always exist a monad structure of the
composed functor $TS:\catex\to\catex$.\footnote{In fact, no-go theorems exist in important cases;
see e.g.~\cite{VaraccaWinskelMSCD06} for barriers in combining non-determinism and probabilities.}
There are, however, monads that compose extremely well with any monad, and give rise to
what are called \emph{monad transformers}. In turns out that
several of these monads additionally interact nicely with various forms of
\emph{order-enrichments}. We describe here a few examples that are useful in the context
of cost analysis,
and allow us to integrate analysis techniques within semantics, simply as modular monadic effects.

\subsubsection{The cost monad transformer}

Cost models (the operation of assigning resource usage to statements and expressions) are central to cost analysis~\cite{resource-iclp07}.
Semantically, they can be easily integrated into the monadic framework by using the \emph{cost monad transformers}.

\begin{definition}[Cost Monad Transformer]
  Given a (pre)monad $(M(-),\eta_M,\circ_M)$ on~\Set,
  the \emph{cost (pre)monad transformer} $(M(\costm(-)),\eta_{M\costm},\circ_{M\costm})$
  is defined as follows:
  \begin{itemize}
    \item $M(\costm(-)) \colon \Ob(\Set) \to \Ob(\Set)$ is the mapping-on-objects $X \mapsto M(X \times \zz)$;
    \item the unit is given by $\eta_{M\costm,X}(x) \coloneqq \eta_{M,X}((x,0))$ for every $x \in X$;
    \item the composition is given by $g \circ_{M\costm} f \coloneqq g^+ \circ_M f$, for $f \in \Set(X, M(\costm Y))$ and $g \in \Set(Y, M(\costm Z))$,
    where $g^+ \in \Set(\costm Y, M(\costm Z))$ stands for
    $g^+(y,c) \coloneqq \big(\big(\eta_M\circ((z,c')\mapsto(z,c+c'))\big)\circ_M g\big)(y)$.
  \end{itemize}
\end{definition}

It is easy to see that $(M(\costm(-)),\eta_{M\costm},\circ_{M\costm})$ is indeed a (pre)monad.
We speak of the \emph{cost monad} $(\costm(-),\eta_\costm,\mu_\costm)$ as a special case of the cost monad transformer, where $M$ is the identity monad.

\begin{remark}[Cost equations in recurrence-based cost analysis]
  Consider the cost monad $\costm$, and denote
  by $(-)_s : (X \to \costm X) \to (X \to X)$ and
  $(-)_c : (X \to \costm X) \to (X \to \zz)$ the (post)
  projections on the first and second components respectively,
  which we call the \emph{state} and \emph{cost} projections.
  We~get
  \begin{equation*}
    (g \circ_\costm f)_s = (g_s \circ f_s)
    \quad \text{ and } \quad
    (g \circ_\costm f)_c = (g_c \circ f_s) + f_c,
  \end{equation*}
  where the sum is taken pointwise.
  This indicates that the cost of $g \circ_\costm f$
  corresponds to the cost of $f$ plus the cost of $g$ applied to the result of $f$.
  When $X = \zz$, as it is the case after performing the \emph{size} abstraction
  of a program to track the evolution of
  data structures with respect to some metric,
  the two projections $(-)_s$ and $(-)_c$
  capture the phenomenon described e.g. in~\cite{caslog-short,brockschmidt2014alternating}:
  recurrence-based cost analysis must produce
  \emph{size} equations in addition to \emph{cost} equations to be
  compositional.
\end{remark}

\begin{proposition}[Order preservation]
  Let $(M(-),\eta_M,\circ_m)$ be a (pre)monad, pointwise order-enriched by a \Pos- (resp. \CPO-, resp.~\CLat-) structure on objects.
Then, $(M(\costm(-)),\eta_{M\costm},\circ_{M\costm})$ can also be pointwise order-enriched by a \Pos- (resp. \CPO-, resp.~\CLat-) structure on objects.
\end{proposition}
\begin{proof}
  If all $MX$ are posets (resp. cpos, resp. complete lattices), this is in particular the case of all $M(X\times \zz)$.
Moreover, the operation $(-)^+$ above is monotone in $g$ by
  reasoning pointwise on each $(y,c)$, and by monotonicity of $\circ_M$,
  hence $\circ_{M\costm}$ is monotone in both arguments.
\end{proof}

\begin{example}
  Consider our $\IR$ language with respect to any monad $(T(-),\eta_T,\circ_T)$ on~\Set,
  and apply the cost monad transformer to obtain a semantics of $\IR$ with respect to the monad $(T(\costm(-)),\eta_{T\costm},\circ_{T\costm})$.
  We can extend the language with {\tt tick(expr)} statements ({\tt expr} of type {\bf int})
  interpreted as follows:
    $\semstmt{{\tt tick(expr)}}(\funvec)(m) := \letin{T\costm}{a}{\semexpr{{\tt expr}}(\funvec)(m)}{(m,\pi_1(a))}$.
  Informally, this states that the cost of executing {\tt tick(expr)} is the value of {\tt expr}.
  Adding now {\tt tick(1)} before every call to a function counts the number of times that function
  is executed.
\end{example}

To consider several types of resources it suffices to compose the cost monad transformer with itself several times,
obtaining a monad with mapping-on-objects $M(\costm^k(-)) \colon X \mapsto M(X \times \zz^k)$.
Instead of adding {\tt tick} statements, one can also make the evaluation of expressions themselves contribute to the cost
(e.g., when the cost model must track the cost of arithmetic operations)
by tweaking the semantics locally.

One can also relate the cost monad transformer with the notion of ``ghost variables'',
which in this case consists in adding a variable ${\tt \_cost}$ initialised at $0$ and
incremented when cost is produced. This variable does not modify the behaviour of the
program, but it can be helpful during program analysis to discover additional variable
relations~\cite{Ghosts-CAV14,Ghosts-VMCAI20}.

\subsubsection{The counter monad transformer}
\label{subsubsec:counter-monad-transformer}

Another technique that is commonly used in cost analysis,
is to
introduce ``counter'' variables, and then formulate recurrences using these counters.
For instance, in recurrence-based loop analysis~\cite{kovacs-phdthesis,kincaid-closed-forms-popl19,regular-path-clauses-hcvs21-short-no-url}, the number of iterations~$k$ \emph{still to be performed} is made the input of an
unknown function~$f$. With a recurrence equation of the form $f(k)(\vec{x}) = e(...f(k-1)(\vec{x})...)$, with $e$ being an expression,
one then describes variable values at loop exit as a function of
their value at loop entry and of the total number of performed iterations.
We can give a novel viewpoint on this approach, as an instance of a \emph{counter monad transformer}.

\begin{definition}[Counter monad transformer]
  Let $(M,\eta_M,\circ_M)$ be a (pre)monad on \Set, assumed to be pointwise order-enriched via a $\CLat$ order structure on objects.
  The \emph{counter (pre)monad transformer} $(\counterm M(-),\eta_{\counterm M},\circ_{\counterm M})$ is defined as follows:
  \begin{itemize}
    \item $\counterm M(-) \colon \Ob(\Set) \to \Ob(\Set)$ is the mapping-on-objects $X \mapsto (\nn \to M(X))$;
    \item the unit is given by $\eta_{\counterm M,X}(x) \coloneqq (n\mapsto \ite(n=0,x,\bot_{MX}))$ for every $x \in X$;
    \item the composition is given by $g \circ_{\counterm M} f = x \mapsto \left(k \mapsto \sqcup_{i+j=k} \big((y \mapsto g(y)(j)) \circ_M f(x)(i)\big)\right)$,
    for $f \in \Set(X, \counterm M Y)$ and $g \in \Set(Y, \counterm M Z)$.
\end{itemize}
\end{definition}
One can see that $(\counterm M(-),\eta_{\counterm M},\circ_{\counterm M})$ is indeed a (pre)monad,
that is also pointwise order-enriched via a \CLat structure on objects.

For some intuition, the counter monad can be seen as a ``dual'' of the cost monad: instead of accumulating costs along the execution by summing costs via compositions, it shares a \emph{budget} in compositions, and somehow accumulates ``promises'': behaviours that can be realised for a given budget. To represent counters counting down to $0$, a possible program instrumentation would be to initialise a variable as $(\top : \textbf{int})$, decrement it when budget is consumed, and \emph{assert} it $0$ at the end of the program.
The program can then be analysed by abstract interpretation using \emph{backward analysis} passes.
The interested reader may have a look at~\cite{ProphecySepLog-POPL19,brahmakshatriya-arxiv26-bwd-prophecy-buildit} for use cases of such
  ``(ghost) prophecy variables'' in other contexts.

Starting from a \CLat order structure, we can combine the counter monad transformer with other effects, such as cost, e.g. obtaining the order-enriched monad $\counterm \pp \costm (-)$,
or by combining several $\counterm(-)$.
Variations are also possible, e.g. replacing the structure $(\nn,+)$ of counters by a different monoid structure. At the language level, one can consider operations on counters, such as resets.

\subsection{Outlook: Applicability of Operator Semantics to Diverse Languages}

Note that the concept of operator semantics is much more general than the case of \IR,
which is merely an example (used in our implementation as an intermediate language for
analysis).

In particular, this (global) fixed-point delaying approach does not
require recursive function calls to be the only form of recursivity in our language.
For instance, supposed we wished to add a $\textbf{while}(\texttt{cond})\{\texttt{stmt}\}$
primitive to our language.
We can interpret it by adding, to the vector of denotations on which $\semop{-}$ operates,
one function variable $\funvec_n$ for each loop in the program (indexed by a
\emph{node} symbol $n\in\wnodes$), and by setting e.g.\footnote{
  While nodes could be interpreted either left or right recursively, leading to functional
  equations with different properties. Similarly, more nodes (e.g. for each basic block)
  could be included, though it does not seem to help in postfixpoint proofs.
},
$\semstmt{{{\tt while}}^n{{\tt(cond)\{stmt\}}}}(\funvec)(m) = \funvec_{n}(m)$,
{{\footnotesize \[
  \semop{{\tt Prog}} \in \End\big(\prod_{f\in \Funs\cup\wnodes} \Mem_{\Vars_{f, in}} \to T\big(\Mem_{\Vars_{f, out}}\big)\big)
,\quad
\funvec \mapsto
  \begin{pmatrix*}[l]
    \left(\argvec \mapsto
    \Big(\big(\eta_T \circ \pi_{\{X_{\text{res}}\}}\big)\circ_T \semstmt{f_{\text{body}}}(\funvec)\Big)(\argvec)
    \right)_{f \in \Funs}
\\
\left(m \mapsto \begin{pmatrix*}[l]\letin{T}
      {b}{\sembexp{{n_{\texttt{cond}}}}(\funvec)(m)\\}
      {\ite\big(b,\, \semstmt{n_{\texttt{body}}}(\funvec)(m),\, \eta_T(m)\big)
    }\end{pmatrix*}
    \right)_{n \in \wnodes}
  \end{pmatrix*},
\]
}}
where $n_{\texttt{cond}}$ and $n_{\texttt{body}}$ refer to the guard and body of the loop indexed by $n\in\wnodes$.

This approach is also not limited to first-order languages: it is straightforward to express
an operator semantics for an higher-order functional language like
\pcf~\cite{PLOTKIN1977-PCF,curien-pcf-chapter-98}, where \emph{\texttt{fix}} commands lead to functional
unknowns (i.e. ``recursion points'') tracked by the operator.

Similarly, we can very naturally support logic programming paradigms, e.g. the predicates
of \textsc{Prolog} providing recursion points. Both monotonic and non-monotonic features
can be modelled (by basing value semantics either on sets of groundings
or on arbitrary sets of terms).

In fact, this idea can be applied to semantics of systems well beyond traditional
programming languages, e.g. to continuous systems and differential equations, using
non-standard analysis approaches~\cite{Hasuo12-nsa-CAV-with-doi}, transforming e.g. the
example of Appendix~\ref{app:proba-sem} into a study of Brownian motion by postfixpoint
techniques, using infinitesimal steps.

\section{Abstraction of Operators for Abstract Compilation}
\label{sec:abs-op-for-abs-comp}

The ``Programs $\equiv$ Operators'' viewpoint enables programs to be compared,
abstracted and combined using purely algebraic operations, while
preserving control flow structure and machine representability.
Since programs become elements of complete lattices, it is natural to consider abstractions of these lattices, leading to \emph{abstract programs}. This opens the possibility of implementing \emph{abstract compilers} as \emph{abstract interpreters} employing higher-order abstract domains built upon \emph{Galois connections in operator space}.
That is,
instead of just abstracting states from~$\pp(\Mem_{\Vars})$,
we abstract operators from $\End\big(\prod_{f\in\mathcal{F}}\Mem_{\Vars_{f,in}} \to T(\Mem_{\Vars_{f,out}})\big)$.

In this Section, we introduce a theoretical toolkit to build and reason about GCs in operator space
and abstract compilers, applicable in general settings.
To present the results, we adopt a
categorical lens, which we believe is well-suited to higher-order abstract interpretation
and to study abstract notions of compositions.
We introduce key notions of this framework in
Section~\ref{subsec:cat-gc}.
Sections~\ref{subsec:state-space-absop} and~\ref{subsec:abstract-monads}
exemplify the approach
by discussing properties of the two GCs below:
{\scriptsize
 \begin{equation*}
\End\Bigg(\prod_{f\in\mathcal{F}}\Mem_{\Vars_{f,in}} \to T(\Mem_{\Vars_{f,out}})\Bigg)
   \underset{\ref{subsec:state-space-absop}}{\galois{}{}}
   \End\Bigg(\prod_{f\in\mathcal{F}}\Mem^\sharp_{\Vars_{f,in}} \to T(\Mem^\sharp_{\Vars_{f,out}})\Bigg)
   \underset{\ref{subsec:abstract-monads}}{\galois{}{}}
   \End\Bigg(\prod_{f\in\mathcal{F}}\Mem^\sharp_{\Vars_{f,in}} \to T^\sharp(\Mem^\sharp_{\Vars_{f,out}})\Bigg)
 \end{equation*}}In the leftmost GC (Section~\ref{subsec:state-space-absop}) the abstraction of operators effectively acts on memories,
while the rightmost GC (Section~\ref{subsec:abstract-monads}) acts on the monad $T$.
In the context of cost analysis, the first step corresponds to what we call a \emph{size abstraction},
and the second enables \emph{recurrence extraction}.

\subsection{Parametric Abstract Domains and Galois Connections, categorically}
\label{subsec:cat-gc}

This section provides categorical tools used to build and study the GCs in
operator space announced above. Section~\ref{subsubsec:soundness-oplaxity}
recalls how oplax functors and transformations model sound computations and their
abstractions; the mechanism giving rise to them, namely transport of
a structure through object-wise GCs, is adapted to construct monad abstractions
in Section~\ref{subsec:abstract-monads}.
Section~\ref{subsubsec:param-gcs-are-functors} then develops \emph{parametric}
GCs as functors into $\CLat_\sqcup$, culminating in the Kan domain/codomain
abstraction (Theorem~\ref{theorem:kan-hom-domain-codomain-abstraction}), the
key ingredient of the state-space abstraction of
Section~\ref{subsec:state-space-absop}.

\subsubsection{``Soundness $\rightsquigarrow$ Oplaxity'': Transfer of Structure by object-wise GCs.}
\label{subsubsec:soundness-oplaxity}

We build upon~\cite{CatAbs-MFPS23} and recall that sound, state-level
abstract domains that work across states of varying signatures---that is, \emph{parametric abstract
domains}---may be viewed as \emph{oplax} functors.
In (order-)enriched category theory, (op)lax notions are weakened versions of
algebraic structures in which some axioms hold only up to inequality. As
observed by~\cite{CatAbs-MFPS23,Steffen92}, this is precisely the slack needed
for \emph{sound} program analysis: oplax notions model \emph{simulations},
sound \emph{interpretations} (abstract semantics), and sound/best
\emph{abstractions} between semantics. \emph{Lax} notions instead correspond to
\emph{completeness}.

\begin{definition}[(Op)lax functor]
  In the context of order-enrichments\footnote{
    Note that we do \emph{not} require (op)lax functors to be
    order-enriched themselves a priori.
    While this is a common assumption in general enriched category theory,
    we follow~\cite{Hasuo15, GoncharovSchroder13},
    and find it to be limiting as it complicates the study of
    important examples such as $\pp(-)$, which is
    not monotone on arrows if viewed as a \Pos-endofunctor.
  }, we define an \emph{oplax} functor
  $F:\catex\to\catex'$ between unital magmoids, where $\catex'$ is order-enriched, as
  a prefunctor such that $F(f \circ g) \leq F(f)\circ F(g)$ and $F(\id_X) \leq \id_{FX}$.
  (Similarly, it is said to be \emph{lax} when the inequalities are in the other direction.)
  An (op)lax functor is \emph{normal} whenever $F(\id_X)=\id_{FX}$ for all $X$.
\end{definition}

In~\cite{CatAbs-MFPS23}, oplax functors arise by \emph{transport} of the structure
of a functor $\sem{-}_c \colon \syntcat \to \CLat_\sqcup$ through an object-wise collection of GCs
$\big(\pp(\sem{X})\galois{\alpha_X}{\gamma_X}A(X)\big)_{X \in \Ob(\syntcat)}$.
More specifically, for the syntactic category $\syntcat$ of some programming language
without recursion, the authors define the ``usual'' set-theoretical semantics as a functor
$\sem{-} \colon \syntcat \to\Set$ that maps, e.g., variables $\Vars$ to memories $\Mem_\Vars$,
and statements to their denotations
(e.g.,~$\sem{\texttt{let x = 3}} = m \mapsto m[x \gets 3]$),
and recover the \emph{collecting semantics} as
$\sem{-}_c \coloneqq \pp(-)\circ\sem{-} \colon \syntcat \to \CLat$.
For a concrete example, consider a purely numerical language with an object-wise collection of GCs between sets of values and hyperboxes. The authors build an \emph{interval} abstract interpretation
functor $A_\ii \colon \syntcat \to\CLat$ mapping, on objects, variables $\Vars$
to interval hyperboxes $\ii(\zz^\Vars)$, and on arrows,
$A_\ii(f) \coloneqq \alpha \circ \sem{f}_c \circ \gamma$,
which yields the \emph{best abstract transformer}. This functor is only \emph{oplax}: composition is not preserved exactly. Rather:
\[
  A_\ii(f \circ g)
 = \alpha \circ \sem{f \circ g}_c \circ \gamma
 = \alpha \circ \sem{f}_c \circ \sem{g}_c \circ \gamma
 \pleq \alpha \circ \sem{f}_c \circ \gamma \circ \alpha \circ \sem{g}_c \circ \gamma
 = A_\ii(f) \circ A_\ii(g),
\]
i.e. the composition of two best abstractions is in general sound but not optimal.
Finally, the soundness of the \emph{abstract transfer functions} $A_{\ii}(f)$
relative to the \emph{concrete transfer functions} $\sem{f}_c$
is captured by an \emph{oplax} natural transformation $\sem{-}_c\Rightarrow A_{\ii}(-)$,
defined next.

\begin{definition}[Oplax natural transformation]
  An \emph{oplax} natural transformation $\alpha:F\Rightarrow G$ between two
  prefunctors $F,G:\catex\to\catex'$, where $\catex'$ is order-enriched, is a prenatural transformation (i.e. a collection of maps $(\alpha_X)_{X\in\Ob(\catex)}$)
  such that $\alpha_Y \circ Ff \leq Gf \circ \alpha_X$ for all $f\in\catex(X,Y)$. \end{definition}

In other words,~\cite{CatAbs-MFPS23} show that transporting a \emph{functor} through GCs
yields an \emph{oplax functor}, made precise in
Proposition~\ref{prop:optimal-oplax-functor-by-GC} below. Compared
to~\cite{CatAbs-MFPS23}, we slightly update the terminology, introducing
\emph{additive} prenatural transformations, so as to conclude that
\emph{additive oplax natural transformations between oplax functors model
sound abstractions between sound domains}: given an object-wise collection of
GCs presented as a pair of prenatural transformations $\alpha:T\Rightarrow
T^\sharp$, $\gamma:T^\sharp\Rightarrow T$, additivity of $\alpha$ is equivalent
to the existence of right adjoints $\gamma_X$ to $\alpha_X$ for each object~$X$.

\begin{definition}[Additive (pre)natural transformation]
  A prenatural transformation $\alpha:F\Rightarrow G$ between two
  prefunctors $F,G:\catex\to\Pos$ (where the targets are posets)
  is \emph{additive} whenever each $\alpha_{X}$ is additive,
  i.e. $\alpha_X(\sqcup_i x_i) = \sqcup_i \alpha_X(x_i)$ whenever these joins exist,
  for all $X \in \Ob(\catex)$.
\end{definition}

\begin{proposition}[\cite{CatAbs-MFPS23}]
  \label{prop:optimal-oplax-functor-by-GC}
  Let $T:\catex\to\CLat$ be an oplax functor,
  and $\big(TX \galois{\alpha_X}{\gamma_X} A_X\big)_{X\in\Ob(\catex)}$ be an object-wise
  collection of Galois connections.
  Then, the following $T^\sharp$ is an oplax functor,
  and $\alpha:T\Rightarrow T^\sharp$ is an additive oplax natural transformation.
This oplax functor is \emph{initial} among all sound abstractions of $T$,
  that is, among all oplax functors $T':\catex\Rightarrow\CLat$ such that
  $\alpha:T\Rightarrow T'$ is oplax natural.
  \begin{align*}
    T^\sharp:
           \catex \to \CLat    &&
               X  \mapsto A_X &&
       (f:X\to Y) \mapsto \alpha_Y \circ Tf \circ \gamma_X.
  \end{align*}
\end{proposition}
As discussed in Section~\ref{subsec:preliminaries-categories}, transporting
categorical structure through GCs yields \emph{magmoid} structures.
Following the same pattern, Section~\ref{subsec:abstract-monads} will develop
\emph{abstract monads} by transporting order-enriched monads. In contrast
with~\cite{CatAbs-MFPS23}, we apply Proposition~\ref{prop:optimal-oplax-functor-by-GC}
(and variants of it)
starting directly from a \emph{semantic} category rather than the syntactic
category~$\syntcat$. This change is motivated by monadic interpretations of abstract
domains in the context of recursive programs.

\begin{example}[Interval abstraction from a semantic category]
  Let $\pp(-):\Set\to\CLat$ be the powerset functor.
  We can define an oplax functor $\ii_\iota(-):\Set\to\CLat$ directly on $\Set$:
  we impose $\ii_\iota(\Mem_\Vars) = \ii(\zz^\Vars)$ on the sets of interest,
  choosing arbitrary lattices elsewhere while ensuring an object-wise collection of GCs $\pp(-)\galois{}{}\ii_\iota(-)$
  (e.g., the trivial 1-point lattice where we do not wish to use intervals).
  Arrows are given by lifting: $\ii_\iota(f):=\gamma \circ \pp(f) \circ \alpha$.
  This flexibility is permitted because oplaxity does not require preserving isomorphisms.
  We thus obtain a pair of functors, the ``concrete'' $\pp(-)$ and an ``abstract'' oplax
  $\ii_\iota(-):\Set\to\CLat$, together with an additive oplax natural transformation
  $\alpha:\pp(-)\Rightarrow\ii_\iota(-)$.
  By postcomposing with the forgetful functor $U \colon \CLat\to\Set$
  sending lattices to their underlying set,
  these become $\Set$ endofunctors together with a choice of order
  structure on objects. This is the viewpoint later used to define abstract monads.
\end{example}

\begin{remark}[Unnaturality as Choice of Structure]
  \label{remark:unnatural-constructions}
  Oplax functors are \emph{unnatural}: unlike functors, they need not preserve isomorphisms.
  For example, the oplax functor $\ii_\iota(-):\Set\to\CLat$ above maps the isomorphic sets $\Mem_{\{x\}}$ and $\Mem_{\{x,y\}}$ to the non-isomorphic complete lattices $\ii(\zz)$ and $\ii(\zz^2)$. It is
  helpful to view $\ii_\iota$ as the composition of two ingredients:
  first, an analysis-dependent choice $(\iota,i)$ of a complete lattice $\iota(X)$
  and an embedding $i_X:X\to U(\iota(X))$ for each set $X$,
  without naturality restrictions\footnote{$\iota:\Set\to\CLat$ is a prefunctor,
    $U:\CLat\to\Set$ is the forgetful functor,
    $i:\Id_\Set \Rightarrow U\circ\iota$ is a prenatural transformation.};
  second, a canonical construction of ``intervals in a complete lattice'',
  described by a functor $\ii:\CLat\to\CLat$ of the form\footnote{
    Such arbitrary pairs of bounds $[x_{\text{lb}},x_{\text{ub}}]$, without the requirement
    $x_{\text{lb}}\leq x_{\text{ub}}$,
    are known as \emph{Kaucher} intervals.
    Collections of intervals in complete lattices, whether Kaucher or usual,
    form complete lattices.
  } $\ii(L)=(L,\geq) \times (L,\leq)$.
  From this perspective, the symmetry-breaking of $\ii_\iota$ is entirely concentrated
  in the choice of order structure $\iota$, much as $\nn$ and $\qq$
  are isomorphic as sets but not as ordered sets.
  We interpret this unnaturality as a witness
  that analyses may \emph{choose} algebraic structures relevant to the \emph{analysis} itself,
  without having to integrate this structure in concrete semantics, where it
  is irrelevant and may be considered artificial\footnote{
    Nevertheless, such symmetry-breaking could be eliminated  by choosing a semantic category recording more ``low-level'' details actually used by
    implementations to select abstract domains,
    e.g. here remembering in objects both the set of stores and the underlying set of
    variables.
  }.

The freedom offered by $(i,\iota)$ allows considering quite general forms of interval analyses,
exploiting other orders than the usual on $\zz$, e.g. divisibility order,
and diverse orders on non-numerical types. Other abstract domains can be seen to be parameterised by a choice of structure.
\end{remark}

\subsubsection{Parametric GCs are Functors $\catex \to \CLat_\sqcup$}
\label{subsubsec:param-gcs-are-functors}

Parametric GCs (Galois \emph{connectors} in~\cite{cousotpopl14galoiscalculus-short-with-doi})
can be used to modularly construct GCs from given parameters,
making them useful to build GCs in operator space; we study their properties.
While several common parametric GCs take other GCs as parameters
(i.e. additive maps $\alpha$, admitting right adjoints $\gamma$),
the state-space (size) abstractions of
Section~\ref{subsec:state-space-absop} are instead induced by
\emph{mere functions} $m : X \to \mathcal{M}X$, which admit no such
right adjoint. (In cost analysis, these functions are the \emph{metrics}
used to perform the size abstraction.)
This section presents the tool handling this situation,
and moreover incorporates it into a functor that encompasses
multiple previous constructions of GCs in function space.

First observe that ordinary GCs between complete lattices are just \emph{arrows} in
$\CLat_\sqcup$: an additive $\alpha\colon L\to L^\sharp$ in \CLat induces a unique
coadditive $\gamma\colon L^\sharp\to L$; the category $\CLat_\sqcup$ of complete lattices and
additive morphisms is \emph{isomorphic} to the category $\CLat_{\mathrm{GC}}$ of
complete lattices and GCs.
A \emph{parametric GC} is then a prefunctor $\catex\to\CLat_\sqcup$,
mapping \emph{arrows} of a parameter category $\catex$ to arrows of $\CLat_\sqcup$
(and implicitly objects to objects); for $\catex=\CLat_\sqcup$, the parameters are again GCs.

It is natural to ask whether parametric GCs are \emph{functorial}, i.e. how they behave
with respect to composition in the parameter category.
Classical examples (codomain abstraction
${\dot\alpha \colon f\mapsto \alpha \circ f}$; Example~\ref{ex:GC-lift-2GCs}, which yields a functor\footnote{
  The double covariance should be surprising to the reader thinking of the classical
  $\Hom(-,-):\Set^{op}\times\Set\to\Set$ functor, where the arrows $u$ and $v$ in $\Hom(u,v)(f) = v
  \circ f \circ u$ are not ``parallel'' but go in opposite directions.
  For complete lattices, schematically, this would give $(L_1^\sharp \to L_2)\galois{}{}(L_1\to L_2^\sharp)$,
  which is not as useful, as it mixes concrete and abstract worlds.
  The redeeming feature of $\CLat_\sqcup$, enabling the double covariance, is its strong structure:
  any additive map $\alpha$ comes attached with a coadditive map in the reverse direction.
}
${\CLat_\sqcup\times\CLat_\sqcup\to\CLat_\sqcup}$;
pre/post abstraction~\cite{Cousot21-book} as a functor
$\Set\to\CLat_\sqcup$, etc.) confirm that this is the case.
Rather than detailing these examples, we directly present
the \emph{(Kan) domain/codomain abstraction}, which simultaneously generalises these
function-space examples
(including the recently discovered \emph{domain abstraction}~\cite{absfun-STTT26}, which
generalises pre/post abstraction to posets beyond sets and to complete lattices $L$ of
truth values beyond $\bb$), and is the key ingredient of the forthcoming size
abstraction of Theorem~\ref{theorem:size-abstraction-kleisli-oplax-endofunctor}.

\begin{definition}[(Left) Kan extension along a monotone function]
  \label{def:left-kan-ext-order}
  Let $(X,\leq_X)$, $(A, \leq_A)$ be posets, $(L,\leq_L)$ be a complete lattice,
  and $m:X\to A$, $f:X\to L$ be monotone maps.
  The \emph{left Kan extension of $f$ along $m$} is the monotone map
  $\exists_m(f):A \to L$ defined as the left adjoint to precomposition $\gamma_m(-)=(-)\circ m$,
  computed as $\exists_m(f) \coloneqq a \mapsto \sqcup_L \{ f(x)\,|\, m(x) \leq_A a \}$.
\end{definition}

We have $\exists_{m \circ m'}(-) = (\exists_m\circ\exists_{m'})(-)$,
so the parametric GC $\Pos\to\CLat_\sqcup$ transforming $\Pos$ arrows into Galois connections is in
fact a \emph{functor} (Appendix, Proposition~\ref{prop:Kan-is-functorial}).
Moreover, if $X \galois{\alpha}{\gamma} A$ is a GC between complete lattices and
$f\in\Pos(X,L)$, then $\exists_\alpha(f) = f \circ \gamma$:
$\exists_m(-)$ generalises the ``precomposition by $\gamma$'' operation
underlying Example~\ref{ex:GC-lift-2GCs}
to situations in which no right adjoint to $m$ is available. This observation leads to the following theorem.

\begin{restatable}[(Kan) domain/codomain abstraction]{theorem}{TheoremKanHomDomainCodomainAbstraction}
  \label{theorem:kan-hom-domain-codomain-abstraction}
  The following is a functor
\begin{align*}
    K(-,-) : \Pos \times \CLat_\sqcup \to \CLat_\sqcup &&
    (X, L)      \mapsto \Pos(X, L)                  &&
    (m, \alpha) \mapsto \big(f \mapsto \alpha \circ \exists_m(f)\big).
  \end{align*}
\end{restatable}

\noindent As a direct consequence of functoriality, $K(\id_Y,\alpha)\circ K(m,\id_L) = K(m,\alpha) = K(m,\id_{L^\sharp})\circ K(\id_X,\alpha)$;
we obtain a commuting square of GCs: \emph{domain abstraction commutes with codomain abstraction}.

The following derived constructions will lead directly to the operator abstraction of
Section~\ref{subsec:state-space-absop}.

\begin{corollary}\label{corr:size-abstraction-bifunctor}
  If $T(-):\Pos\to\CLat_\sqcup$ is a functor, then so is $K(-,T-):\Pos\times\Pos\to\CLat_\sqcup$.
\end{corollary}

\begin{example}
  Corollary~\ref{corr:size-abstraction-bifunctor} applies to covariant predicate functors,
  which generalise $\pp(-)$ and $\ww_\infty(-)$. These are functors $T \colon \Pos\to\CLat_\sqcup$ given by
  $X \mapsto \Pos(X,\Omega)$, $f\mapsto (p \mapsto y \mapsto \sqcup_\Omega \{p(x)\,|\,f(x)\leq y\})$,
where $\Omega \in \CLat$ serves as a truth object.
\end{example}

\begin{corollary}\label{cor:size-abstraction-bifunctor}
  For any functor $T(-):\Set\to\CLat_\sqcup$
  we have a functor $K(-,T-):\Set\times\Set\to\CLat_\sqcup$.
  In particular, any pair $m_1:X_1\to A_1$, $m_2:X_2\to A_2$ of functions between
  sets yields a Galois connection
  $\big(X_1\to TX_1, \pleq_{TX_1}\big)\galois{}{}\big(A_1 \to TA_2, \pleq_{TA_2}\big).$
\end{corollary}

For the last corollary, we include $\Set$ as a subcategory of $\Pos$ by viewing
sets $X$ as discrete posets $(X,=)$.
The functor $T(-)$ can be instantiated to~$\pp(-)$, $\ww_\infty(-)$, $\pp\costm(-)$,
$\counterm\pp\costm(-)$, etc.\ of Section~\ref{sec:operator-semantics},
each of which sends \Set maps to \emph{additive} maps for the chosen order structure.

\subsection{State-space Abstraction of Operators, by Mere Functions}
\label{subsec:state-space-absop}

This section constructs the \emph{leftmost} GC of the pipeline announced at the
start of Section~\ref{sec:abs-op-for-abs-comp}:
we abstract the state spaces $\Mem_\Vars$, leaving the monad $T$ untouched.
The distinctive feature of this step is that it is induced by \emph{mere
functions} $m_X : X \to \mathcal{M}X$ with no right adjoint required: a
metric need not admit a meaningful concretisation back to concrete data (there
is no canonical program term of a prescribed size).
Concretely, Theorem~\ref{theorem:size-abstraction-kleisli-oplax-endofunctor}
below applies the GCs of Corollary~\ref{cor:size-abstraction-bifunctor} to each
Kleisli homset, and shows that, taken together, these GCs are
\emph{compatible with Kleisli composition}: they assemble into an oplax
endofunctor of $\Kleisli_T$.

\begin{theorem}
  \label{theorem:size-abstraction-kleisli-oplax-endofunctor}
  Consider an arbitrary mapping $\mathcal{M}:\Ob(\Set)\to\Ob(\Set)$,
  and a collection of functions $m = \big(m_X : X \to \mathcal{M}X \big)_{X: \Set}$.
Suppose that $(T(-), \eta_T, \circ_T)$ is a $\Set$-monad,
  order-enriched pointwise via a choice of $\CLat$ structure on $T$-objects,
  and that additionally all arrows $Tf$ are \emph{additive} for this \CLat structure
  (in other words, $T$ factorises as $T=U \circ T_s$ via a functor
  $T_s:\Set\to\CLat_\sqcup$).

  Then, $\mathcal{M}$ induces an \emph{oplax} endofunctor in the corresponding Kleisli category, given by
  \begin{align*}
    \mathcal{M} :
      \Kleisli_T  \to \Kleisli_T &&
              X   \mapsto \mathcal{M}X &&
      (f:X\to TY) \mapsto K(m_X, T m_Y)(f),
  \end{align*}
whose action on arrows is additive.
  Moreover, $m:\Id_{\Kleisli_T} \Rightarrow \mathcal{M}$ is an oplax natural transformation.
\end{theorem}

See the oplax naturality of $m$ as soundness: every concrete computation
$f : X \to TY$ is simulated, up to~$\pleq$, by its abstraction
$\mathcal{M}(f)$ on abstract states.

\begin{corollary}
  \label{corollary:size-abstraction-gives-nice-GCs}
  $\mathcal{M}$ provides a collection of GCs
  ${(\Mem \!\shortrightarrow\! T(\Mem_{\Vars}))\leftrightarrows(\Mem^\sharp\!\shortrightarrow\! T(\Mem^\sharp))}$,
where $\Mem_\Vars^\sharp:=\mathcal{M}(\Mem_\Vars)$,
  abstracting all the domains used for function denotations in Fig.~\ref{fig:opsem-ir-num-monadic}.
This further lifts
  to GCs abstracting all the objects $\semop{\Prog}$, $\semstmt{{\tt stmt}}$, $\semaexp{{\tt aexp}}$ and $\sembexp{{\tt bexp}}$.
\end{corollary}
For example, for $\semop{\Prog}$ we get a GC
\begin{equation*}
    \textstyle\End\big(\prod_{f\in\mathcal{F}}\Mem_{\Vars_{f,in}} \to T(\Mem_{\Vars_{f,out}})\big)
    \galois{\alpha_\Phi}{\gamma_\Phi}
    \End\big(\prod_{f\in\mathcal{F}}\Mem^\sharp_{\Vars_{f,in}} \to T(\Mem^\sharp_{\Vars_{f,out}})\big).
  \end{equation*}

This is an abstraction between \emph{programs} operating on concrete stores
$\Mem$ and ``abstract programs'' operating on abstract stores $\Mem^\sharp$,
so that $\alpha_\Phi(\semop{\Prog})$ produces the \emph{optimal abstract program}
corresponding to $\Prog$.
Notice that, unlike classical abstract interpretation, we do \emph{not} require a Galois
connection between $\Mem$ and $\Mem^\sharp$, which are allowed to be mere sets.

Oplaxity of $\mathcal{M}$ gives, in particular, $\mathcal{M}(f \circ_T g) \pleq \mathcal{M}(f) \circ_T \mathcal{M}(g)$:
from the practical viewpoint of abstract compiler design, the $\circ_T$ operation
on $\Mem$-functions can be \emph{soundly overapproximated} by the $\circ_T$ operation on
$\Mem^\sharp$-functions.
In other words, there is a clear \emph{compositional} path to produce sound abstractions
$\semop{\Prog}^\sharp$ of the \emph{optimal abstract compilation} $\alpha_\Phi(\semop{\Prog})$ of
$\Prog$, where, e.g., abstract recursive calls can be simply handled by $\circ_T$.
Using another vocabulary, (effectful) computations on $\Mem$ can be \emph{simulated} by computations on $\Mem^\sharp$.

\subsection{Abstract Composition with Abstract Monads}
\label{subsec:abstract-monads}

This section constructs the \emph{rightmost} GC of the pipeline of
Section~\ref{sec:abs-op-for-abs-comp}: having abstracted state spaces
(Section~\ref{subsec:state-space-absop}), we now abstract the monad $T$
itself, e.g. moving from powersets of sizes to \emph{intervals} of sizes.
Beyond the existence of the resulting GCs, the two main takeaways are
algebraic: abstract Kleisli composition is \emph{non-associative}, but admits
a \emph{best bracketing} (Proposition~\ref{prop:best-bracketing}), and it is
\emph{strictly more precise} than the associative composition of abstract
transfer functions
(Remark~\ref{rem:example-monadic-more-precise}).

Similarly to the situation of Proposition~\ref{prop:optimal-oplax-functor-by-GC},
the class of (order-enriched) monads is not closed under transport through object-wise GCs.
Searching for a monadic analog of the class of \emph{oplax functors} leads us to the
following definition, building upon Definition~\ref{def:premonad} of \emph{premonads}.
For simplicity, we restrict ourselves to GCs between complete lattices
that are precise enough to represent ``singletons'' (more precisely, units $\eta_{T^\natural}$),
and that are \emph{Galois insertions} (i.e. $\alpha\circ\gamma=\id$)\footnote{The Galois
insertion hypothesis could be avoided, at the cost of oplax weakenings of the unitality of
$\widetilde{\Kleisli}_{T^\sharp}$.}.

\begin{definition}[Abstract monads]
  \label{def:abstract-monad}
  An \emph{abstract monad} is a \Set-premonad\footnote{
    Recall Definition~\ref{def:premonad}: it is a $T\colon\Ob(\Set)\to\Ob(\Set)$ and a structure of
    unital Kleisli magmoid with
    $(g \circ_T f)\circ u = g\circ_T (f \circ u)$.
  } $(T(-),\eta_T,\circ_T)$ which is order-enriched pointwise via a \CLat structure on
  $T$-objects.
\end{definition}

In other words, beside order-enrichment, an abstract monad is a monad in
which associativity of $\circ_T$ has been dropped; a monadic analogue, in
spirit, of passing from functors to oplax functors. The next theorem justifies
this weakening: it is what survives transport through GCs.

\begin{theorem}[Optimal abstract monad by transfer of object-wise GC]
  \label{theorem:optimal-oplax-monad-by-GC}
  Let $(T(-),\eta_T,\circ_T)$ be an abstract monad,
  and $\big(TX \galois{\alpha_X}{\gamma_X} A_X\big)_{X\in\Ob(\Set)}$ be an object-wise
  collection of GCs between complete lattices.
Suppose, furthermore, that these GCs are precise enough to represent units,
  in the sense $\gamma_X\circ\alpha_X\circ\eta_{T,X} = \eta_{T,X}$ for all $X$,
  and that these are Galois insertions, i.e. $\alpha_X\circ\gamma_X=\id$.

  Then, we obtain a structure of abstract monad
  $(T^\sharp(-),\eta_{T^\sharp},\circ_{T^\sharp})$
  by setting
  $\eta_{T^\sharp,X}~:=~\alpha_X \circ \eta_{T,X}$ and
  $g \circ_{T^\sharp} f~:=~\alpha_Z \circ \big((\gamma_Z \circ g) \circ_T (\gamma_Y \circ f)\big)$
  for all $f\in\Pos(X,T^\sharp Y)$, $g\in\Pos(Y,T^\sharp Z)$.
Moreover, $\alpha$ induces the following (additive, normal) oplax functor between the corresponding
  Kleisli (unital) magmoids.
\begin{align*}
    \mathcal{A}:
       \widetilde{\Kleisli_T} \Rightarrow \widetilde{\Kleisli_{T^\sharp}} &&
                           X \mapsto X                                 &&
                 (f:X\to TY) \mapsto \alpha\circ f.
  \end{align*}
\end{theorem}
\noindent As in Corollary~\ref{corollary:size-abstraction-gives-nice-GCs}, by composing
$\mathcal{A}:\Kleisli_{T}\Rightarrow\widetilde{\Kleisli_{T^\sharp}}$ with the additive
oplax functor $\mathcal{M}$ of
Theorem~\ref{theorem:size-abstraction-kleisli-oplax-endofunctor},
we can further abstract
programs (e.g. moving from powersets of sizes to intervals of sizes).
As is common in abstract interpretation, we then obtain overapproximating abstract
operator semantics almost mechanically, by simply replacing $\circ_T$ by $\circ_{T^\sharp}$ throughout
figures, yielding sound abstract compilers\footnote{At least as a fallback: optimal abstract compilation
must avoid a too compositional approach, which can lose precision.}.

In some sense, the GCs
$(\Mem^\sharp_{\Vars_{f,in}} \to T(\Mem^\sharp_{\Vars_{f,out}}))\leftrightarrows(\Mem^\sharp_{\Vars_{f,in}} \to T^\sharp(\Mem^\sharp_{\Vars_{f,out}}))$
provided by Theorem~\ref{theorem:optimal-oplax-monad-by-GC},
then lifted to abstractions of $\semop{\Prog}$, $\semstmt{{\tt stmt}}$, etc., are very simple and
extremely natural: just GCs on objects, lifted pointwise, to products, and $\End$-lifted.
However, the value of the abstract monadic viewpoint does not lie merely in the \emph{existence}
of these GCs; it is threefold. First, it provides insights on the \emph{algebraic properties of abstract composition}
$\circ_{T^\sharp}$: because of non-associativity, the implementer must \emph{choose} between
possible bracketings when adapting Fig.~\ref{fig:opsem-ir-num-monadic}. And, as we show below,
there are cases where a \emph{best bracketing} can always be identified
(Proposition~\ref{prop:best-bracketing}), providing guidance to the abstract compiler implementer.
Second, this non-associative composition in monadic function spaces
$(X\to T^\sharp Y)$ is \emph{strictly more precise} than associative composition in function spaces
$\Pos(T^\sharp X, T^\sharp Y)$ (that in abstract interpretation
reflect usual abstract transfer functions, at least for cases like
$T^\sharp=\ii_\iota$ that avoid additional effects like costs and counters); see
Remark~\ref{rem:example-monadic-more-precise}.
Third, more conceptually, it promotes classical state abstractions into \emph{abstract effects},
with modularity advantages: thanks to the
(pre)monad transformers of Section~\ref{subsec:eff-cost-monad-transformers},
our results apply to abstract monads such as $\counterm\ii_\iota\costm(-)$,
enabling abstract compilation to generalised recurrence equations exploiting cost and counters.

To illustrate, let us provide details on the simple case of \emph{intervals} as an abstract effect.

\begin{example}[Details of the $\ii_\iota(-)$ construction and GC with $\pp(-)$]
  Recall that we define $\ii_\iota(-):\Ob(\Set)\to\Ob(\Set)$, equipped with a \CLat order structure
  on objects, by composition of a selection of orders $(\iota,i)$, of the ``intervals in complete
  lattice'' functor $\ii \colon \CLat\to\CLat$, and of the forgetful functor $U:\CLat\to\Set$.
For each complete lattice $L$, there is a GC $\pp(U(L)) \leftrightarrows \ii(L)$,
  given by $\alpha^\pp_L(E) = [\sqcap E,\sqcup E]$ and
  $\gamma^\pp_L([x_{lb},x_{ub}])=\{x\in L\,|\, x_{lb} \leq_L x \leq_L x_{ub}\}$. Moreover, since $\pp(-)$ is additive on arrows, for each set $X$, there is a GC
  $\pp(X)\leftrightarrows\pp((U\circ\iota)(X))$,
  which takes a concrete set
  $E\subseteq X$ to its direct image $\pp(i_X)(E)$ by the ``embedding'' $i_X$;
  the right adjoint to $\pp(i_X)$ is the preimage operation $(i_X)^{-1}$. By composition, this leads to an object-wise collection of GCs $(\pp(X) \galois{\alpha_X}{\gamma_X} \ii_\iota(X))_{X\in\Ob(\Set)}$.
To ensure that these are Galois insertions, one can follow the construction of ``ordinary''
  intervals, and consider only intervals in the image of the lower closure
  $\alpha\circ\gamma$: this collapses to a single $\bot_{\ii(X)}$ all ``improper''
  intervals (those where $\gamma(I)=\varnothing$).

\end{example}

\begin{example}[Interval Abstract Monad]
  Consider the collection of GCs
  $(\pp(-)\galois{\alpha}{\gamma} \ii_\iota(-))$ described above,
  now assumed to be Galois insertions,
  and consider the (concrete) order-enriched monad structure $(\pp(-),\eta_\pp,\circ_\pp)$.
Recall that $\eta_\pp$ produces singletons, so it suffices to assume that all ``embeddings'' $i_X:X\to(U\circ\iota)(X)$ are injective to ensure that all singletons $\{x\}$ can be represented
  by $[i_X(x),i_X(x)]$, i.e. that $\gamma \circ \alpha \circ \eta_\pp = \eta_\pp$.
  Then, Theorem~\ref{theorem:optimal-oplax-monad-by-GC} yields
  the \emph{interval abstract monad} $(\ii_\iota(-),\eta_\ii,\circ_\ii)$,
  where $\eta_{\ii,X}(x) = [i_X(x),i_X(x)]$,
  and where, for $f: X\to\ii_\iota(Y)$ and $g:Y\to\ii_\iota(Z)$, abstract composition is given by
  \[
    (g\circ_\ii f)(x) = \textstyle{\bigsqcup}_{\iota(Z)} \big\{ g(y) \,\big|\, y\in\iota(Y),\, f_{lb}(x) \leq_{\iota(Y)} y \leq_{\iota(Y)} f_{ub}(x) \big\}.
  \]
\end{example}

Note that applying Proposition~\ref{prop:optimal-oplax-functor-by-GC} would instead produce the
  ``intervalisation'' operation that sends a function $f:X\to Y$ in \Set to
  \[
    \ii_\iota(f)
    = \alpha\circ\pp(f)\circ\gamma
    = \Big([x_{lb},\,x_{ub}]
           \mapsto
           \textstyle{\bigsqcup}_{\iota(Y)} \big\{ f(x) \,\big|\, x\in\iota(X),\, x_{lb}\leq_{\iota(X)} x \leq_{\iota(X)} x_{ub} \big\}
      \Big).
  \]
  The oplax functor and abstract monad structures of $\ii_\iota$ are compatible,
  see Prop.~\ref{prop:oplax-functor-and-abstract-monad-compat} in Appendix.

We now substantiate the two algebraic claims made above. The following
counter-example shows that $\circ_\ii$ is indeed non-associative, so an
abstract compiler \emph{must} choose a bracketing.
Proposition~\ref{prop:best-bracketing} shows that this choice is not arbitrary:
one bracketing is always at least as precise as the other.

\begin{example}[Non-associativity of $\circ_\ii$]
  \label{ex:non-assoc}
  Let $\iota$ take the usual order on $\zz$,
  let ${f,g,h \colon \zz\to\ii(\zz)}$ be the functions $f \coloneqq x\mapsto [-1,1]$,
  $g \coloneqq y\mapsto \ite(y<0, [-1,-1], [1,1])$,
  and $h \coloneqq z \mapsto [z^2,z^2]$.
Then, $(g \circ_\ii f)(x) = \sqcup_{-1\leq y \leq 1}\, g(y) = [-1,1]$,
  so $(h \circ_\ii (g\circ_\ii f))(x) = \sqcup_{-1 \leq z \leq 1}[z^2,z^2] = [0,1]$.
However, $(h \circ_\ii g)(y) = \ite({y<0},[1,1],[1,1]) = [1,1]$,
  hence $((h \circ_\ii g)\circ_\ii f)(x) = [1,1]$ is more precise.
\end{example}

\begin{proposition}\label{prop:best-bracketing}
   For all compatible maps $f,g,h$ in $\widetilde{\Kleisli}_\ii$,
   $(h \circ_\ii g) \circ_\ii f {}\mathop{\dot\sqsubseteq}{} h \circ_\ii (g \circ_\ii f)$.\footnote{ In fact, this proposition generalises to all abstract monads derived from $\pp(-)$,
  using the fact that $\alpha$ commutes with $\mu_\pp$, which is just the order-theoretical join.
  However, more exotic multiplications, like $\mu_{\ww_\infty}$ of distributions, also seem to
  display this property.
  At this stage, we have not identified the exact hypotheses enabling this
  ``best bracketing'': we conjecture a relation to \emph{left-additivity}
  properties of $\circ_T$, related to the semi-quantales studied in~\cite{GiacobazziTCS99}.
}
\end{proposition}

In other words, in Example~\ref{ex:non-assoc} the right-associated bracketing
$[0,1]$ is not an accident: when composing chains of abstract Kleisli arrows
(e.g. sequences of statements, or nested recursive calls), associating
\emph{to the right} is always at least as precise.

\begin{remark}[Monadic composition $\circ_\ii$ is more precise than associative composition]
  \label{rem:example-monadic-more-precise}
  Instead of abstracting monadic function spaces of the form $\zz^r\to\pp(\zz^s)$
  by $\zz^r\to\ii(\zz^s)$, another natural approach would be to embed the former into
  $\Pos(\pp(\zz^r),\pp(\zz^s))$ by Kleisli extension, and then abstract to
  $\Pos(\ii(\zz^r),\ii(\zz^s))$ by \End-lifting. Then, composition is simple and associative,
  although the data structures representing these objects are more complex than for
  $\zz^r\to\ii(\zz^s)$.

  At first sight, it may appear that $\zz^r\to\ii(\zz^s)$ simply \emph{embeds} into
  $\Pos(\ii(\zz^r),\ii(\zz^s))$, by a Galois co-insertion (i.e. $\gamma\circ\alpha = \id$,
  $\alpha\circ\gamma\pleq\id$), defined by
  $\alpha(f)([x_{lb},x_{ub}]) = \sqcup_{x_{lb}\leq x \leq x_{ub}} f(x)$ and
  $\gamma(f^\sharp)(x)=f^\sharp([x,x])$.
This is indeed a co-insertion, so operations in $\zz^r\to\ii(\zz^s)$
  can be \emph{exactly} represented by operations in $\Pos(\ii(\zz^r),\ii(\zz^s))$.
  The catch is that the associative composition $\circ$ in $\Pos(\ii(-),\ii(-))$ is \emph{not} the
  best abstraction
  $g^\sharp \circ^\sharp f^\sharp := \alpha(\gamma(g^\sharp) \circ_\ii \gamma(f^\sharp))$
  of monadic composition.

  As an illustration, consider again Example~\ref{ex:non-assoc}, and
  $h^\sharp,g^\sharp,f^\sharp:\ii(\zz)\to\ii(\zz)$ abstracting $h,g,f$.
  Explicitly, ignoring $\bot$ inputs, we have $f^\sharp([x_{lb},x_{ub}])=[-1,1]$,\\
  {\small\centering
  \begin{tabular}{c c c}
    $g^\sharp([y_{lb},y_{ub}]) = \begin{cases}
       [-1,-1] & \text{if $y_{ub} < 0$}\\
       [-1, 1] & \text{if $y_{lb} \leq 0 \leq y_{ub}$}\\
       [1, 1] & \text{if $0 < y_{lb}$}\\
    \end{cases}$
    &
    and
    &
    $h^\sharp([z_{lb},z_{ub}]) = \begin{cases}
       [z_{ub}^2,z_{lb}^2] & \text{if $z_{ub} < 0$}\\
       [0,\max(z_{ub}^2,z_{lb}^2)] & \text{if $z_{lb} \leq 0 \leq z_{ub}$}\\ [z_{lb}^2,z_{ub}^2] & \text{if $0 < z_{lb}$}
    \end{cases}$
  \end{tabular}}
  In particular, using the usual associative $\circ$,
  $(h^\sharp\circ g^\sharp \circ f^\sharp)([x_{lb},x_{ub}]) = h^\sharp([-1,1]) = [0,1]$.
  However, ${\gamma(h^\sharp \circ^\sharp g^\sharp)(y) = (h \circ_\ii g)(y) = \ite(y<0,[1,1],[1,1])}$,
  hence $((h^\sharp \circ^\sharp g^\sharp) \circ^\sharp f^\sharp)([x_{lb},x_{ub}]) = [1,1]$.
\end{remark}

\section{Application and Implementation}\label{sec:application-and-implementation}

We have applied our abstract compilation framework to the cost analysis
problems of Section~\ref{sec:illustrative-example}, in a prototype abstract
interpreter containing higher-order abstract domains.
Given an \IR program and a collection of
\emph{catamorphic} metrics (see below), this abstract compiler produces purely numerical
non-deterministic programs (Section~\ref{subsec:size-abs-impl-comments}), and
further abstractions of them as functional equations on interval-valued
functions (Section~\ref{subsec:fun-itv-abs-impl-comments}), in several output
formats suited to different backend analyses. Given space constraints, we
briefly highlight the most interesting points of the construction, and defer an
extensive experimental evaluation to future work.

\subsection{Abstract Compilation to Numerical Programs, for Catamorphic Metrics}
\label{subsec:size-abs-impl-comments}

The first step of the pipeline described in Section~\ref{sec:illustrative-example} is \emph{size abstraction},
based on a collection of integer-valued metrics $m:T\to\zz$ on program datatypes. Given such metrics, the
construction of Section~\ref{subsec:state-space-absop} directly yields abstract operator semantics
(Fig.~\ref{fig:opsem-ir-abs-size-version}, Appendix). We briefly sketch how these abstractions can be
implemented, focusing on a few operations of particular interest.

\subsubsection{Catamorphic Metrics}
\begin{wrapfigure}[6]{r}{0.32\linewidth}\vspace{-10pt}
  {\tt\footnotesize \begin{tabular}{@{}ll@{}}
        {\bf metric} $m$( $X_{\text{arg}}$: ty ) -> {\bf int} \{\\
        \qquad{\bf let} $X_{\text{res}}$: {\bf int} = {\bf $(\top : \textbf{int})$}\\
        \qquad$
          \begin{aligned}
                  &\tt {\bf match}~X_{\text{arg}}~\{ \dots\\
                  &\quad{}\mid C_i(X_{i1},\dots,X_{ik}) {\tt~=>~stmt_i}\,\\
                  &\ \quad{}\dots\}
            \end{aligned}
            $\\
        \qquad{\bf return} ${X_{\text{res}}}$ \}&
  \end{tabular}
  }
\end{wrapfigure}
The key abstract transfer functions are those for ADT construction and
deconstruction, denoted $\asem{{\tt cnstr}_i}$ and $\asem{{\tt dnstr}_i}$
respectively: the former maps the size vectors of the arguments of a
constructor $C_i$ to the possible size vectors of the constructed term, the
latter is the converse relation, used to interpret {\tt match}
(see Fig.~\ref{fig:opsem-ir-abs-size-version}, Appendix).
Their discovery is fully automated when the collection of metrics is \emph{catamorphic},
in the sense of Definition~\ref{defi:catamorphic-function}: the vector of sizes of a
term depends only on the sizes of its immediate subterms.
In practice, users provide metrics via a restricted function syntax
akin to the one depicted on the right:
the input $X_{\text{arg}}$ is immediately
deconstructed by a \texttt{match}, and each branch {\tt stmt$_i$} features
only conditionals and integer expressions over an
\emph{extended arithmetic} (e.g., with $\log$ and $\max$). Apart from locally
defined variables and $X_{\text{res}}$, {\tt stmt$_i$} may refer only to the variables
$X_{i1},\dots,X_{ik}$ bound by its match case, and the only permitted
function calls are \emph{metric} calls on those same variables. So, each
match case computes the size of a term from the sizes of its immediate subterms,
as catamorphicity demands.

Simple examples of catamorphic metrics include list length and tree height.
The examples of Section~\ref{sec:illustrative-example} are also catamorphic. Conditionals are allowed, so one can define, e.g.,
metrics counting the number of elements in a collection satisfying some Boolean properties.
Moreover, catamorphic metrics may call other metrics, defined on other types,
allowing complex combinations (e.g., sums of sizes of subterms of distinct types).
A non-example is the number $\cardm$ of \emph{distinct} elements of a list of
integers: $\cardm(\Cons(\mathtt{hd},\mathtt{tl}))$ is not determined by
$\cardm(\mathtt{tl})$ and the value of $\mathtt{hd}$. For such
non-catamorphic metrics, users can either provide the transfer functions
themselves, or supply a
\emph{catamorphic overapproximation} (in the sense of pointwise inclusion) of
the metric.

\subsubsection{Optimal \texttt{skip}: feasible sizes}
\label{subsubsec:size-transfer-skip}

Interestingly, the remaining difficulties are entirely captured by the best
abstraction $\asem{{\tt skip}}$ of ${\tt skip}$. While
$(m^\sharp \mapsto \{m^\sharp\})$ is sound, it is not optimal: the oplax unit
inequality $\mathcal{M}(\eta_{T,X}) \pleq \eta_{T,\mathcal{M}X}$ of
Theorem~\ref{theorem:size-abstraction-kleisli-oplax-endofunctor} is
\emph{strict}. The best abstraction must discard \emph{infeasible} size
vectors, i.e.~vectors realised by no concrete value. (Recall footnote~\ref{footnote:fasible-vectors} of
Section~\ref{sec:illustrative-example}: the parity constraints in the
equations for \texttt{eval} are feasibility information at work.)
This precision matters in practice: spurious size vectors can make inductive
invariant proofs fail even when the candidate invariant holds on all concrete
data. Since the exact feasibility predicate is not computable in general, our
implementation overapproximates it by a preliminary analysis of the metric
definitions.

\subsubsection{FOL abstract domain}

We implement the abstract semantics using a higher-order domain representing nondeterministic functions
$\zz^r\to \pp(\zz^s)$ by first-order logic (FOL) formulae over an extended arithmetic,
interpreted as transition relations $\phi(\vec{x},\vec{x}')$.
This is done by encoding set operations in FOL (unions as disjunctions, etc.),
and recursive calls as uninterpreted functions.
Monadic composition $\circ_\pp$ is replaced by the operation $\exists\vec{y}, \phi(\vec{x},\vec{y})\wedge\psi(\vec{y},\vec{x}')$,
and units are replaced by conjunctions of equalities.
Constructors $\asem{{\tt cnstr}}$
are obtained by a forward analysis pass on metric code. For deconstructors~$\asem{{\tt dnstr}}$, we inverse the relation by swapping inputs and outputs.

\subsubsection{\texttt{Size} Lifter}

Size abstraction is in fact not implemented directly in the FOL domain, but
as a \emph{functor} that lifts any abstract domain for \emph{numerical}
operators to a domain applicable to programs over symbolic data, by supplying
the appropriate interface. Before interpreting the program itself, the
\texttt{Size} lifter \emph{precomputes} $\asem{{\tt cnstr}}$ and
$\asem{{\tt dnstr}}$ by a first analysis pass over the code of the metrics,
using the operations of the underlying numerical domain (including backward
statement semantics).
This architecture has a pleasant consequence, in the same spirit
as~\cite{lesbre2024compiling}: instantiated with domains implementing
essentially identity transformations (reconstruction of syntax trees, FOL or
CHC encodings of concrete semantics), the lifter performs
\emph{transformation by abstract interpretation}. Instantiated with coarser
domains, it produces even simpler abstract programs, giving access to more
powerful backend analyses. The interval-based higher-order domains presented
next provide one such instantiation, and yield systems of generalised
recurrence equations.

\subsection{Interval Effect: Extracting Generalised Recurrence Equations}
\label{subsec:fun-itv-abs-impl-comments}

\subsubsection{Parametric Intervals: Intervals whose Bounds are Functions}

We have implemented domains representing operators in $\End(\prod_{f\in\Funs}\aMem\to\ii(\aMem))$,
whose values can be interpreted (and solved) as recurrence equations.
Interval-valued functions $f:\zz^r\to\ii(\zz^s)$ (or functions in $\Pos(\ii(\zz^r),\ii(\zz^s))$, which is an abstraction thereof),
can be expressed as pairs of functions $[f_{lb},f_{ub}]$.
Abstract values can also be seen as \emph{parametric intervals}: classical intervals
whose bounds are parameterised by some inputs (values and functions),
represented by symbolic expressions. Fig~\ref{fig:machine-rep-param-itv-domain} shows the data structure we used
in the $\End(\prod_{f\in\Funs}\ii(\aMem)\to\ii(\aMem))$ case.

Except for composition, implementations are obtained by lifting standard algorithms for intervals to symbolic expressions rather than concrete integers.
Though relatively simple, these domains are particularly expressive: while output-output relations
cannot be expressed, they can in principle represent arbitrary input-output relations
(non-linear, non-convex, etc.):
this is manageable in our case thanks to our fixed-point delaying operator semantics.
To deal with the problem of expressing for which input values the hyperbox $[f_{lb}(\vec{x}),f_{ub}(\vec{x})]$,
is empty, we use a global boolean guard by default. We also rely heavily on strategies (rewritings, SMT, etc.)
to simplify symbolic expressions.

\begin{figure}\begin{equation*}
    {\scriptsize
      \begin{array}{@{}rl@{}}
d_{eq}
         = \vec{d}
       & \in D^{(\Phi)}_{\Funs}
         \eqdef \prod_{f\in\Funs} D^{(\lambda)}_{\Funs,\Vars_{f,in},\Vars_{f,out}}\\
d
         = {\color{dgreen}{(}}h,~{\color{dgreen}{\varphi)}}
       & \in D^{(\lambda)}_{\Funs,\Vars_{in},\Vars_{out}}
         \eqdef \big(\texttt{NumExpr}_{\Funs,\Vars_{in}} \times \texttt{NumExpr}_{\Funs,\Vars_{in}}\big)^{\Vars_{out}} \color{dgreen}{~\times~\texttt{BoolExpr}_{\Funs,\Vars_{in}}}
       \\
e & \in \texttt{NumExpr}_{\Funs,\Vars_{in}}\\
& = v_{lb}
             \mid v_{ub}
             \quad \text{(variable bound, for $v\in\Vars_{in}$)}
         \\
& \mid \pi_i \circ f_{lb}(\vec{e})
             \mid \pi_i \circ f_{ub}(\vec{e})
             \quad \text{(projected function call, for $f/(r,s)\in\Funs$ and $i\in[1,s]$)}
         \\
& \mid c
             \quad \text{(constant, $c\in\zzi = \{-\infty\}\uplus\zz\uplus\{+\infty\}$)}
         \\
& \mid \Diamond_{\mathrm{a}}(\vec{e})
             \quad \text{(builtin arithmetic operations, $\Diamond_{\mathrm{a}}\in\{\texttt{+, -, *, /, max, min, pow, log2pos}\}$)}
         \\
& \mid \ite(b,e,e)
             \quad \text{(piecewise expressions)}
         \\
& \color{dgreen}{\mid \letin{}{x}{[e,e]}{e'}}
             \quad \text{(let expressions, with $x$ a fresh variable and $e'\in\texttt{NumExpr}_{\Funs,\Vars_{in}\uplus\{x\}}$)}
         \\
b & \in \texttt{BoolExpr}_{\Funs,\Vars_{in}}\\
         & = e~\Bowtie~e
           \mid b~\Diamond_{\mathrm{b}}~b
           \mid {\bf true} \mid {\bf false}
           \quad \text{(basic boolean operations, $\Bowtie \in \{<,\leq,=,\geq,>\}$, $\Diamond_{\mathrm{b}} \in \{\texttt{and},\texttt{or}\}$)}
           \\
         & \color{dgreen}{\mid \ite(b,b,b)}
           \color{dgreen}{\mid \letin{}{x}{[e,e]}{b'}}
           \\
      \end{array}
    }
    \end{equation*}

    \vspace{-1em}

  \caption{Possible machine representation of domain of (operator on) interval-valued functions.
Elements in {\color{dgreen}green} are concision optimisations (global guards, DAG-style let expressions, etc.).
    \vspace{-1em}
  }
  \label{fig:machine-rep-param-itv-domain}
\end{figure}

\subsubsection{Tradeoffs in Abstract Composition}

As expected, \emph{composition} is the most subtle operation.
The framework admits several implementations with different precision/efficiency tradeoffs.
By Section~\ref{subsec:abstract-monads}, a straightforward, sound option consists in replacing all $-_\pp$
indices in Fig.~\ref{fig:opsem-ir-abs-size-version} by $-_\ii$.
Practical implementation of monadic composition $\circ_\ii$ is non-trivial, but approaches are available: parametric optimisation engines, overapproximating rewriting techniques,
symbolic interval arithmetic, etc.

This compositional approach loses precision compared to the opposite extreme: the best abstraction
$\alpha^{\ii}_\Phi\circ\Phi^{\sharp_\pp}\circ\gamma^{\ii}_\Phi$
of only the \emph{topmost} object, i.e. the abstract operator $\Phi^{\sharp_\pp}$ itself
obtained e.g.~in the FOL format from~Sec.~\ref{subsec:size-abs-impl-comments}.
Preliminary experiments via CAS symbolic optimisation engines suggest that
computing this object directly is feasible for small-sized programs,
although this does not scale well. An intermediate approach consists in
computing best abstractions of \emph{medium-sized} syntactic components of the program,
similarly to~\cite{ZixinHuang-AURA-SAS25}.

Finally, one can avoid the monadic composition $\circ_\ii$ altogether,
using a coarser domain operating on functions $\ii(\zz^r)\to\ii(\zz^s)$.
As exemplified by Rem.~\ref{rem:example-monadic-more-precise}, this leads
to more complicated expressions, but composition then reduces to symbolic substitutions,
which, for $\circ_\ii$, is possible only under additional monotonicity assumptions.

\section{Related Work}
\label{sec:rel-work}

Works on \emph{abstract compilation} were already discussed in
Section~\ref{sec:introduction}.
Among these, we highlight~\citep{RossignoliS06}, where, we believe, an implicit form of
operator semantics first appeared (their \emph{non-cyclicity
interpretation}), although not yet reified and generalised, and limited to
\emph{Boolean} values
(which makes the implementation of many operations easier, by finiteness of the value domain).

In \emph{recurrence-based static cost analysis}, we highlight the seminal work
of Wegbreit~\cite{Wegbreit75}, and the approaches of~\cite{caslog-short,plai-resources-iclp14-short}
underlying the cost analyses of \textsc{CiaoPP}~\cite{ciaopp-sas03-journal-scp-shortest}.
Our approach supports all
metrics of these works, and provides automated support for the \texttt{diff}-based metrics
of~\cite{caslog-short} conjointly with the $\min/\max$ ``deep'' combinations
of~\cite{plai-resources-iclp14-short}.
Our catamorphic metrics also support additional features, including
further metric combinations and conditional expressions within
metrics. The recurrences we extract are more precise,
supporting piecewise expressions (whose importance is highlighted
by~\cite{montoya-phdthesis-short,mlrec-tplp2024-nourl}) and $\sup$ expressions
that preserve relations between sizes of subterms,
(e.g. yielding a linear
bound on tree identity where~\cite{plai-resources-iclp14-short} obtains an
exponential one),
with optimality results absent
from~\cite{caslog-short,plai-resources-iclp14-short,montoya-phdthesis-short}.
Moreover, their soundness requires none of the monotonicity hypotheses (on
solutions, program expressions, recursive calls, metrics, etc.) assumed
in~\cite{caslog-short,plai-resources-iclp14-short}: our only monotonicity,
pointwise set inclusion, is derived from the semantics itself
(Proposition~\ref{prop:pos-monads-monot}).

Abstract compilation is also related to \emph{program transformation by
abstract interpretation}. A trace-based formalisation was proposed
in~\cite{Cousot02}, providing correctness proofs and design guidance for
transformations viewed as overapproximations of ideal semantic abstractions,
though its syntax/semantics adjunctions remained more conceptual than
implementable.
Related ideas appeared in~\cite{spec-inv-pepm03,ai-with-specs-sas06},
under the names of \emph{abstract specialisation} and \emph{partial deduction}, building
on combinations of abstract interpretation and partial
evaluation~\cite{gallagher:pepm93,anal-peval-horn-verif-2021-tplp-short}.
More recently,~\cite{lemerre2023ssa} proposed a transformation
approach based purely on abstract domains, applicable to decompilation
and extraction of SSA graphs, which~\cite{lesbre2024compiling}
generalised through the notion of (preordered) \emph{free
functor domains} and simulation theorems for (\textsc{OCaml}) functors
interpreted as compilation passes.
As explained in Section~\ref{sec:introduction}, our approach is intermediate between
\cite{Cousot02} and \cite{lesbre2024compiling}, combining the
algebraic advantages of a semantic approach with the finite
representability of syntactic objects.
In contrast to these works, our approach also directly supports recursive programs.
More generally, we remark that while the representation of sets of
traces in $\pp(\Sigma^*)$ through transition systems in
$\pp(\Sigma\times\Sigma)$ may be seen as an instance of operator semantics, the latter
is more general, as it easily supports (non-linear) recursion, monadic
effects, and modular control of ``observation points'' granularity.

Operator semantics is also inspired by \emph{coalgebraic
approaches to automata theory}~\cite{JacobsBook2016,HasuoTraces07,SokolovaCoalgSurvey11,CoalgNetKat15,EfficientCoalgRefin17,Sanada26CoalgDijkstra}, which highlight the
advantages of working with semantic objects that \emph{generate} trace
semantics rather than syntax or fixpoint semantics only, both in terms
of reasoning (information that can be extracted) and generalisability
of analysis algorithms to more exotic situations (probabilities,
etc.).

Our use of catamorphisms is inspired by the line of work of~\cite{catamorphism-2022,catamorphisms-TPLP2024}.

Our core motivation---simplifying programs to unlock increasingly powerful
analyses---leads at its extreme to \emph{applying complete techniques within
incomplete analyses for undecidable problems}.
This approach is well-illustrated
by~\cite{cra15,kincaid2017-short,kincaid2018,kincaid-closed-forms-popl19,Breck-CHORA,kincaid2023,Kincaid24-ideals,Kincaid24-VASR}
and by algebraic loop
summarisation~\cite{kovacs-phdthesis,kovacs08,kovacs18,kovacsSAS22}. Broadly,
this includes all applications of decidable logic fragments to program analysis.

Lastly, our techniques relate to
categorical approaches to abstract interpretation~\cite{Steffen92,CatAbs-MFPS23,kincaid26-cat-rob}---categorical frameworks going beyond the reading of
order theory as Bool-enriched category theory.
Non-categorical approaches to the
study of composition-like operations in
abstract interpretation are also proposed in~\cite{GiacobazziTCS99},
suggesting possible future work combining
these frameworks.

\section{Future Work}
\label{sec:ccl}

We plan to study further operator abstractions, especially
\emph{flow} abstractions (discarding recursion points through composition,
related to order-theoretical tensor products~\cite{TensorNielson85} for
bi-additive compositions) and their relation to
interpreter and evaluation-strategy choices in partial-evaluation-based
abstract compilation. More practically, as this paper focused on
theoretical foundations, we plan to present our implementation in detail,
together with an experimental evaluation covering cost analysis applications,
extracted equations, and combinations
with backend solvers yielding precise non-linear bounds. Finally, we plan to
instantiate the framework more widely: on concrete semantics \emph{not} based
on the covariant $\pp(-)$ monad (probabilities, other covariant predicate
functors), on numerical domains intermediate between FOL and intervals, on
combinations of size abstractions with other information on symbolic data, on
other monad transformers, such as the state monad applied to (abstractions
of) hardware-sensitive stateful cost models, and more generally
explore instantiations to other applications of abstract compilation.

\bibliographystyle{ACM-Reference-Format}
\bibliography{local}

\newpage \appendix
\section{Additional Preliminary Material}
\label{appendix:additional-preliminaries}

This appendix contains additional preliminary material that was omitted (or only informally explained) from the main body of the paper for space reasons.

\subsection{Order Theory}
\label{subsec:appendix-order}

\vspace{-8pt}
\begin{longtable}{p{0.199\textwidth}p{0.753\textwidth}}
\toprule
  {\small \textbf{Notion}} & {\small \textbf{Meaning}} \\
  \toprule
  \endfirsthead
\toprule
  \textbf{Notion} & \textbf{Meaning} \\
  \toprule
  \endhead
\midrule
  \multicolumn{2}{r}{\textit{Continued on next page}} \\
  \endfoot
\bottomrule
  \endlastfoot
Poset
    & A set $X$ equipped with a partial order. \\
  Monotone function
    & An $f \colon X \to Y$ between posets satisfying $x \leq_X y \implies f(x) \leq_Y f(y)$. \\
  $\Pos(X,Y)$
    & The set of all monotone functions from poset $X$ to poset $Y$. \\
  $\pleq$
    & The pointwise order ${f\pleq g\iff\forall x\in X,\,f(x)\leq_Y g(x)}$. \\
  $\End(X)$
    & The poset $\Pos(X,X)$ ordered pointwise, for $X$ a poset. \\
  Chain
    & A subset of a poset with pairwise comparable elements. \\
  Join
    & Least upper bound $\bigsqcup S$ of a subset $S$ of a poset. \\
  Meet
    & Greatest lower bound $\bigsqcap S$ of a subset $S$ of a poset. \\
  CPO
    & Complete partial order: A poset where every chain admits a join. \\
  Complete lattice
    & A poset where every subset admits both a join and a meet. \\
  Additive function
    & A join-preserving function: $f(\bigsqcup S) = \bigsqcup f(S)$ (when $\bigsqcup S$ exists). \\ Coadditive function
    & A meet-preserving function: $f(\bigsqcap S) = \bigsqcap f(S)$ (when $\bigsqcap S$ exists). \\ Galois Connection
    & Two (monotone) functions $\alpha,\gamma$ between posets with $\alpha(x)\leq y \Leftrightarrow x\leq\gamma(y)$. \\
\end{longtable}

\begin{definition}[Partially ordered sets]\label{def:poset}
    A \emph{partially ordered set} (or \emph{poset} for short) $(X, \leq)$ is a
    set $X$ together with a binary operation $\leq$ that is reflexive, transitive
    and antisymmetric, called its order.
\end{definition}

Two elements $x$ and $y$ of a poset are said to be \emph{comparable} whenever
$x\leq y$ or $y\leq x$.

The opposite $\leq^{op}$ of an order $\leq$ is the order defined by $x
\leq^{op} y \iff y \leq x$. In the body of the paper and in the appendices,
we often denoted $\leq^{op}$ by simply mirroring the symbol,
for instance writing $\geq$ instead of $\leq^{op}$, $\sqsupseteq$ instead of $\sqsubseteq^{op}$, etc.

\begin{definition}[Monotone functions]\label{def:monotone-function}
    Given two posets $(X,\leq_X)$ and $(Y,\leq_Y)$, a
    function $f:X\to Y$ is
    \emph{monotone} (or \emph{order-preserving}) whenever
    $\forall x,x'\in X \colon x \leq_X x' \implies f(x) \leq_Y f(x')$.
\end{definition}

\begin{definition}[$\Pos(X,Y)$ and its pointwise order]\label{def:pointwise-order}
  Given two posets $(X,\leq_X)$ and $(Y,\leq_Y)$, we denote by the set $\Pos(X,Y)$ of all
  monotone functions from $X$ to $Y$.
  This set is equipped with the \emph{pointwise order} $\pleq$ defined by $f \pleq g \iff \forall x\in X,\,f(x)\leq_Y g(x)$,
  for every $f,g \in \Pos(X,Y)$.
\end{definition}

\begin{definition}[Chain]\label{def:chain}
    A \emph{chain} in a poset $(X,\leq)$ is a subset $A\subseteq X$ such that all its elements are pairwise comparable.
\end{definition}

\begin{definition}[Join and Meet]\label{def:join-and-meet}
    Given a poset $(X,\leq)$ and a subset $A\subseteq X$, a \emph{join} of $A$ (also called least upper bound, or lub), is an element $u\in X$ such that for every $a \in A$, $a\leq u$ and for every $v\in X$, $(\forall x\in A,\,x\leq v)\implies u\leq v$.
    Such an element $u$ is unique if it exists, and we denote it by $\bigsqcup A$. When $A = \{x,y\}$, we also write $x \sqcup y$ for binary joins.

    The \emph{meet} $\bigsqcap A$ of a subset $A$ (also called greatest lower
    bound, or glb) is defined dually, replacing the order $\leq_X$ by the reverse
    order $\geq_X$.

    When they exist, the least an greatest elements of the poset~$(X,\leq)$ are respectively
    called its \emph{bottom} and \emph{top}, and are denoted by
    $\bot_X=\sqcup_X \varnothing = \sqcap_X X$ and  $\top_X=\sqcap_X \varnothing = \sqcup_X X$.
\end{definition}

\begin{definition}[CPO]\label{def:cpo}
  A \emph{complete partial order}\footnote{Note that this term is ambiguous in the literature -- we choose
  this definition as it guarantees existence of denotional semantics for
  recursive constructs under weak assumptions.} (or \emph{cpo} for short) is a poset where all
  \emph{chains} (including the empty one and not just finite ones) admit a
  join.
\end{definition}

\begin{definition}[Complete lattice]\label{def:complete-lattice}
  A \emph{complete lattice} is a poset $(X,\leq)$ where all subsets admit both a join and a
  meet, i.e., for all $A\subseteq X$, $\bigsqcup A$ and $\bigsqcap A$ exist.
\end{definition}

\begin{definition}[(Co)additive functions]\label{def:additive-function}
  We say that a function $X\to Y$ between the underlying sets of two complete
  lattices is \emph{additive} (or \emph{join-preserving}) whenever
  $\forall x,x'\in X, f(x \sqcup_X x') = f(x) \sqcup_Y f(y)$.
  Any additive function is monotone, but the converse does not hold in general.
\emph{Co-additive} (or \emph{meet-preserving}) functions are defined dually.
\end{definition}

\begin{remark}[Pointwise and product orders]\label{remark:pointwise-product-orders}
  Whenever $(Y,\leq_Y)$ is a CPO (resp. complete lattice),
  all pointwise orders $\pleq_Y$ are also CPOs (resp. complete lattices).
  Similarly, the \emph{product} $\prod_i X_i$ of CPOs
  (resp. complete lattices) $(X_i,\leq_i)$ is a CPO (resp. complete lattice).
\end{remark}

Because Galois Connections (GCs, for short) underpin much of this paper, let us recall some~of~their properties
(we refer the reader to~\cite{Cousot21-book} for proofs and references).
As suggested above, given posets $(C,\leq_C)$ and $(A,\leq_A)$,
a GC $(C,\leq_C)\galois{\alpha}{\gamma}(A, \leq_A)$ is a pair of functions $\alpha:C \to A$ and $\gamma:A \to C$
satisfying the property that for every $c \in C$ and $a \in A$, $\alpha(c)\leq_A a \iff c \leq_C \gamma(a)$.
\begin{fact}\label{fact:gc-inequalities}
   A pair of functions $(\alpha,\gamma)$ constitutes a GC if and only
   if $\alpha$ and $\gamma$ are both \emph{monotone} and satisfy the inequalities
   $\id_C \pleq \gamma\circ \alpha$ and $\alpha\circ\gamma \pleq \id_A$ (where $\id_Z$ denotes the
   identity function). \end{fact}

The notation $C\galois{\alpha}{\gamma}A$ (we omit the orders when they are clear from the context) serves as a reminder of the inequalities in Fact~\ref{fact:gc-inequalities}: following $\alpha$ then $\gamma$ moves an element up the order, while following $\gamma$ then $\alpha$ moves it down.
The maps $\alpha$ and $\gamma$ of a Galois connection are called \emph{adjoints}, and we say that $\alpha$ is \emph{left adjoint} to $\gamma$, and that $\gamma $ is \emph{right adjoint} to $\alpha$.
Left and right adjoints of a given map are unique, when they exist.
In a GC, $\alpha$ and $\gamma$ are respectively called \emph{abstraction} and \emph{concretisation} functions,
while $C$ and $A$ and called the \emph{concrete} and \emph{abstract domains}.

\begin{fact}In a GC, $\alpha$ is additive and $\gamma$ is coadditive.
\end{fact}

A canonical way of constructing a GC consists in considering an additive map~$\alpha$ between complete lattices $L$ and $L^\sharp$. The lattice structure ensures the existence of the right adjoint to~$\alpha$.

\begin{fact}\label{prop:additive-by-meet}
  Any additive map $\alpha:L\to L^\sharp$ admits a (unique, coadditive) right adjoint $\gamma:L^\sharp\to L$, and it can be computed by $\gamma = a \mapsto \sqcup \{ x \,|\,\alpha(x) \sqsubseteq a\}$.
  Dually, any coadditive map $\gamma:L^\sharp\to L$ admits a (unique, additive) left adjoint $\alpha:L\to L^\sharp$, and it can be computed by $\alpha = x \mapsto \sqcap \{ a \,|\, x \leq \gamma(a)\}$.
\end{fact}

Broadly speaking, \emph{abstract interpretation} is a theory for reasoning about programs by computing sound abstractions of their semantics (defined in a concrete lattice) via primitives performing computations within a (simpler) abstract domain. GCs offer a principled approach for designing these primitives, as they provide a notion of \emph{best} abstraction, rather than a collection of \emph{sound}~ones.

\begin{definition}[Sound abstraction]
  An element $a\in A$ is said to be a \emph{sound} abstraction of $c\in C$ whenever $\alpha(c) \leq a$ (equivalently, $c \leq \gamma(a)$).
Similarly,
  $f^\sharp:A \to A$ is said to be a \emph{sound} abstraction
  of a monotone function $f:C \to C$
whenever
  $\alpha \circ f \pleq f^\sharp \circ \alpha$
  (equivalently, $f \circ \gamma \pleq \gamma \circ f^\sharp$).
\end{definition}

\begin{fact}[Best abstraction]
  The set of sound abstractions of an element $c\in C$ admits a minimum: it is  $\alpha(c)$.
  The set of sound abstractions of some $f\in\Pos(C,C)$ admits a minimum:
  $\alpha \circ f \circ \gamma$.
\end{fact}

The latter fact can in fact be presented as a GC in \emph{function space}
that lifts GCs in value space: see Example~\ref{ex:GC-lift-2GCs}.

\subsection{Category Theory}
\label{subsec:app-category-theory}

\begin{definition}[Magmoid]\label{def:magmoid}
  A \emph{magmoid}\footnote{We use the term \emph{magmoid} to reflect
  the pattern where \emph{groupoids} are categories with inversible
  morphisms, and categories ``are just \emph{monoidoids}'', i.e.
  monoids with several objects.
  We avoid \emph{quivers}, which are often assumed small, and do not
  always imply a composition operation of any kind.}
  (or ``category without axioms'', or sometimes
  \emph{quiver}),
  is a structure \catex consisting of:
  \begin{enumerate}
    \item a \emph{class} $\Ob(\catex)$, called its class of \emph{objects},
    \item for each pair $X,Y\in\Ob(\catex)$ of objects, a
      \emph{set}\footnote{In this paper, we need not go beyond such
      ``locally small'' category-like structures, although it is
      possible.} $\catex(X,Y)$, called a set of \emph{arrows} (or
      \emph{morphisms}, or \emph{maps}), or just an \emph{homset},
    \item for each triple $X,Y,Z\in\Ob(\catex)$, a binary operation
      $\circ_{X,Y,Z}:\catex(Y,Z)\times \catex(X,Y)\to\catex(X,Z)$ called
      \emph{composition}.
When the context is clear, we may drop the index and simply write
      $\circ$.
  \end{enumerate}
\end{definition}

\begin{definition}[Unital magmoids]\label{def:unital-magmoid}
  A magmoid \catex is said to be \emph{unital}
  whenever there exists an object-indexed collection
  $(\id_X)_{X\in\Ob(\catex)}$, where $\id_X \in \catex(X,X)$ is called the
  \emph{identity at $X$},
  such that for all $f\in\catex(X,Y)$, we have $f\circ\id_X=f=\id_Y\circ f$.
\end{definition}

\begin{definition}[Associative magmoids]\label{def:associative-magmoid}
  A magmoid \catex is said to be \emph{associative}
  whenever for all $f\in\catex(W,X)$, $g\in\catex(X,Y)$, $h\in\catex(Y,Z)$,
  we have $(h \circ g) \circ f = h \circ (g \circ f)$.
\end{definition}

\begin{definition}[Categories]\label{def:categories}
  A \emph{category} is a unital and associative magmoid.
\end{definition}

\begin{definition}[(Pre)functor]\label{def:prefunctor}
  A \emph{prefunctor} $F:\catex\to\catex'$ between two \magmoids
  consists in
(1) a function $F:\Ob(\catex)\to\Ob(\catex')$ on objects,
(2) for each $X,Y\in\Ob(\catex)$,
  a function $F:\catex(X,Y)\to\catex(FX,FY)$ on homsets.

  A \emph{functor} $F:\catex\to\catex'$ between \emph{unital} magmoids
  is a prefunctor such that
(3) $F(f\circ g)=F(f)\circ F(g)$ for all compatible $f,g$, and
(4) $F(\id_X) = \id_{FX}$ for all $X$.

  We add the prefix \emph{endo-} (e.g. \emph{endofunctor}) whenever
  $\catex=\catex'$.
\end{definition}

\begin{definition}[(Pre)natural transformation]\label{def:prenatural-transformation}
  A \emph{prenatural transformation} $\tau:F\Rightarrow G$ between
  two (pre)functors $F,G:\catex\to\catex'$ is simply (1) a collection
  $(\tau_{X}\in\catex'(FX,GX))_{X\in\Ob(\catex)}$ of maps in $\catex'$
  indexed by objects in $\catex$.
It is said to be a \emph{natural} transformation whenever
  (2) $\tau_Y \circ Ff = Gf \circ \tau_X$ for all $f\in\catex(X,Y)$.
\end{definition}

\subsection{Equivalence between Definition~\ref{def:monad} and the classical definition of monad}
\label{subsec:proof-monad-by-kleisli-equivalent-to-monad-by-mult}

We remind the reader of the classical definition of monad.
Below, $\Id_\catex$ denotes the identity functor on a category $\catex$,
and $TT = T \circ T \colon \catex \to \catex$ denotes the composition of the endofunctor $T$ with itself.

\begin{definition}[Monad]\label{def:classic-monad}
  A \emph{monad} on a category $\catex$ is a triple $(T, \eta, \mu)$ where:
  \begin{itemize}
    \item $T: \catex \to \catex$ is an endofunctor,
    \item $\eta: \Id_\catex \Rightarrow T$ is a natural transformation called the \emph{unit},
    \item $\mu: TT \Rightarrow T$ is a natural transformation called the \emph{multiplication},
  \end{itemize}
  satisfying the following coherence conditions:
  \begin{align*}
    \mu \circ T\mu &= \mu \circ \mu T, \\
    \mu \circ T\eta &= \mu \circ \eta T = \id_T.
  \end{align*}
\end{definition}

The definition of monad given in Definition~\ref{def:monad} is equivalent to the definition in Definition~\ref{def:classic-monad}.
We only prove the forward direction (Proposition~\ref{prop:monad-by-kleisli-equivalent-to-monad-by-mult} below): a monad $(T(-),\eta_T,\circ_T)$ as in Definition~\ref{def:monad} induces a classical monad $(T,\eta_T,\mu_T)$.
The backward direction is standard\footnote{E.g., one can follow the exposition by Moggi~\cite{moggi-monads-91}, obtain Kleisli triples $(T(-),\eta_T,(-)^{*_T})$ via a result attributed to~\cite{Manes76-book}, then apply $g\circ_T f := g^{*_T} \circ f$.}.
In fact, we believe that the forward direction is also well known,
but we include a full proof for completeness, since we could not find
a reference containing a complete proof under the assumptions we stated.

\begin{proposition}\label{prop:monad-by-kleisli-equivalent-to-monad-by-mult}
  Let $(T(-),\eta_T,\circ_T)$ be a monad on a category $\catex$
  (Definition~\ref{def:monad}), that is:
  \begin{enumerate}[label=(\arabic*)]
    \item $T(-):\Ob(\catex)\to\Ob(\catex)$ is a mapping-on-objects,
    \item $(\eta_{T,X} \in \catex(X,TX))$ is a collection of arrows indexed by $X\in\Ob(\catex)$,
    \item $(\circ_{T,(X,Y,Z)} \colon \catex(Y,TZ)\times\catex(X,TY)\to\catex(X,TZ))$ is a collection of maps indexed by ${X,Y,Z\in\Ob(\catex)}$,
    \item $(g \circ_T f) \circ u = g \circ_T (f \circ u)$ for every $u\in\catex(W,X)$, $f\in\catex(X,TY)$, $g\in\catex(Y,TZ)$,
    \item the $\eta_{T,X}$ are two-sided units for $\circ_T$, i.e.\ $\eta_{T,Y}\circ_T f = f = f\circ_T \eta_{T,X}$ for all $f\in\catex(X,TY)$,
    \item $\circ_T$ is associative, i.e.\ the Kleisli \magmoid of Definition~\ref{def:kleisli-magmoid} is a category~$\Kleisli_T$.
  \end{enumerate}
Define, for every arrow $f \in \catex(X,Y)$ and every object $X\in\Ob(\catex)$,
  \begin{equation*}
    Tf := (\eta_{T,Y} \circ f) \circ_T \id_{TX}
    \qquad\text{and}\qquad
    \mu_{T,X} := \id_{TX}\circ_T\id_{TTX}.
  \end{equation*}
Then, $(T,\eta_T,\mu_T)$ is a monad as in Definition~\ref{def:classic-monad},
  whose induced Kleisli composition ${(g,f) \mapsto \mu_{T,Z} \circ Tg \circ f}$
  coincides with~$\circ_T$.
\end{proposition}

\begin{proof}
  We prove the claim in six points as follows:
  \begin{enumerate}[label=(\roman*)]
    \item \emph{The maps $Tf$ upgrade $T(-)$ into an endofunctor $\catex\to\catex$.}

      \begin{itemize}
        \item For identities, let $X\in\Ob(\catex)$ be any object.
We have
          $T(\id_X)= (\eta_{T,X} \circ \id_X) \circ_T \id_{TX} = \id_{TX}$,
          by the $\circ$-unit and $\circ_T$-unit laws.

        \item For composition, let $f\in\catex(X,Y)$ and $g\in\catex(Y,Z)$.
          We have
\begin{align*}
           Tg \circ Tf
           &= ((\eta_{T,Z} \circ g) \circ_T \id_{TY}) \circ ((\eta_{T,Y} \circ f) \circ_T \id_{TX})
              & \text{(by definition)} \\
           &= (\eta_{T,Z} \circ g) \circ_T (\id_{TY} \circ ((\eta_{T,Y} \circ f) \circ_T \id_{TX}))
              & \text{(compatibility)} \\
           &= (\eta_{T,Z} \circ g) \circ_T ((\eta_{T,Y} \circ f) \circ_T \id_{TX})
              & \text{($\circ$-unit)} \\
           &= ((\eta_{T,Z} \circ g) \circ_T (\eta_{T,Y} \circ f)) \circ_T \id_{TX}
              & \text{(associativity of $\circ_T$)} \\
           &= (((\eta_{T,Z} \circ g)\circ_T \eta_{T,Y}) \circ f) \circ_T \id_{TX}
              & \text{(compatibility)} \\
           &= ((\eta_{T,Z} \circ g) \circ f) \circ_T \id_{TX}
              & \text{($\circ_T$-unit)} \\
           &= (\eta_{T,Z} \circ (g \circ f)) \circ_T \id_{TX}
              & \text{(associativity of $\circ$)} \\
           &= T(g\circ f)
              & \text{(by definition)}.
         \end{align*}

      \end{itemize}

    \item \emph{$\eta_T:\Id_\catex\Rightarrow T$ is a natural transformation},
      i.e.\ $Tf \circ \eta_{T,X} = \eta_{T,Y} \circ f$ for all $f\in\catex(X,Y)$.
\begin{align*}
        Tf \circ \eta_{T,X}
        & = \big((\eta_{T,Y} \circ f)\circ_T \id_{TX}\big) \circ \eta_{T,X}
          & \text{(by definition)} \\
        & = (\eta_{T,Y} \circ f)\circ_T (\id_{TX} \circ \eta_{T,X})
          & \text{(compatibility)} \\
        & = (\eta_{T,Y} \circ f)\circ_T \eta_{T,X}
          & \text{($\circ$-unit)} \\
        & = \eta_{T,Y} \circ f
          & \text{($\circ_T$-unit)}.
      \end{align*}

    \item \emph{$\mu_T:TT\Rightarrow T$ is a natural transformation},
      i.e.\ $Tf \circ \mu_{T,X} = \mu_{T,Y} \circ TTf$ for all $f\in\catex(X,Y)$.
\begin{align*}
        \mu_{T,Y} \circ TTf
        & = \big(\id_{TY}\circ_T\id_{TTY}\big)
            \circ
            \Big(\big(\eta_{T,TY} \circ
                   \big((\eta_{T,Y} \circ f) \circ_T \id_{TX}\big)\big)
                 \circ_T \id_{TTX}\Big)
          & \text{(by definition)} \\
        &= \id_{TY} \circ_T
\Big(\big(\eta_{T,TY} \circ
                   \big((\eta_{T,Y} \circ f) \circ_T \id_{TX}\big)\big)
                 \circ_T \id_{TTX}\Big)
          & \text{(compatibility and $\circ$-unit)} \\
        &= \Big(\id_{TY} \circ_T
\big(\eta_{T,TY} \circ
                   \big((\eta_{T,Y} \circ f) \circ_T \id_{TX}\big)\big)
                 \Big) \circ_T \id_{TTX}
          & \text{(associativity of $\circ_T$)} \\
        &= \Big(\big(\id_{TY} \circ_T \eta_{T,TY}\big) \circ
\big((\eta_{T,Y} \circ f) \circ_T \id_{TX}\big)
                 \Big) \circ_T \id_{TTX}
          & \text{(compatibility)} \\
        &= \big((\eta_{T,Y} \circ f) \circ_T \id_{TX}\big)
              \circ_T \id_{TTX}
          & \text{($\circ_T$-unit and $\circ$-unit)} \\
        &= (\eta_{T,Y} \circ f) \circ_T (\id_{TX} \circ_T \id_{TTX})
          & \text{(associativity of $\circ_T$)} \\
        &= (\eta_{T,Y} \circ f) \circ_T \mu_{T,X}
          & \text{(by definition)} \\
        &= (\eta_{T,Y} \circ f) \circ_T (\id_{TX} \circ \mu_{T,X})
          & \text{($\circ$-unit)} \\
        &= \big((\eta_{T,Y} \circ f) \circ_T \id_{TX}\big) \circ \mu_{T,X}
          & \text{(compatibility)} \\
        &= Tf \circ \mu_{T,X}
          & \text{(by definition)}.
      \end{align*}

    \item \emph{First coherence condition: $\mu_{T,X} \circ T(\mu_{T,X}) = \mu_{T,X} \circ \mu_{T,TX}$
      for all $X\in\Ob(\catex)$.}
      We have
\begin{align*}
        \mu_{T,X} \circ T(\mu_{T,X})
        & = (\id_{TX} \circ_T \id_{TTX})
            \circ
            \big(\big(\eta_{T,TX} \circ (\id_{TX} \circ_T \id_{TTX})\big) \circ_T \id_{TTTX}\big)
          & \text{(by definition)}\\
        & = \id_{TX} \circ_T
              \big(\big(\eta_{T,TX} \circ (\id_{TX} \circ_T \id_{TTX})\big) \circ_T \id_{TTTX}\big)
          & \text{(compatibility and $\circ$-unit)}\\
        & = \big(\id_{TX} \circ_T
              \big(\eta_{T,TX} \circ (\id_{TX} \circ_T \id_{TTX})\big)\big)
            \circ_T \id_{TTTX}
          & \text{(associativity of $\circ_T$)}\\
        & = (\id_{TX} \circ_T \id_{TTX}) \circ_T \id_{TTTX}
          & \text{\!\!\!\!\!\!\!\!\!(compatibility, $\circ_T$-unit, $\circ$-unit)}\\
        & = \id_{TX} \circ_T (\id_{TTX} \circ_T \id_{TTTX})
          & \text{(associativity of $\circ_T$)}\\
        & = \id_{TX} \circ_T (\id_{TTX} \circ \mu_{T,TX})
          & \text{(by definition and $\circ$-unit)}\\
        & = (\id_{TX} \circ_T \id_{TTX}) \circ \mu_{T,TX}
          & \text{(compatibility)}\\
        & = \mu_{T,X} \circ \mu_{T,TX}
          & \text{(by definition)}.
      \end{align*}

    \item \emph{Second coherence conditions:
      $\mu_{T,X} \circ T(\eta_{T,X}) = \mu_{T,X} \circ \eta_{T,TX} = \id_{TX}$
      for all $X\in\Ob(\catex)$.}
      We have
\begin{align*}
        \mu_{T,X} \circ T(\eta_{T,X})
        &= (\id_{TX}\circ_T\id_{TTX}) \circ \big((\eta_{T,TX}\circ\eta_{T,X})\circ_T \id_{TX}\big)
          & \text{(by definition)}\\
        &= \id_{TX} \circ_T \big((\eta_{T,TX}\circ\eta_{T,X})\circ_T \id_{TX}\big)
          & \text{(compatibility and $\circ$-unit)}\\
        &= \big(\id_{TX} \circ_T (\eta_{T,TX}\circ\eta_{T,X})\big)\circ_T \id_{TX}
          & \text{(associativity of $\circ_T$)}\\
        &= \big((\id_{TX} \circ_T \eta_{T,TX})\circ\eta_{T,X}\big)\circ_T \id_{TX}
          & \text{(compatibility)}\\
        &= \id_{TX}
          & \text{($\circ_T$-unit, $\circ$-unit, $\circ_T$-unit)}.
      \end{align*}
Similarly,
\begin{align*}
        \mu_{T,X} \circ \eta_{T,TX}
        &= (\id_{TX}\circ_T\id_{TTX}) \circ \eta_{T,TX}
          & \text{(by definition)}\\
        &= \id_{TX}\circ_T (\id_{TTX} \circ \eta_{T,TX})
          & \text{(compatibility)}\\
        &= \id_{TX}
          & \text{($\circ$-unit and $\circ_T$-unit)}.
      \end{align*}

    \item \emph{$g \circ_T f = \mu_{T,Z} \circ Tg \circ f$, for all $f\in\catex(X,TY)$ and $g\in\catex(Y,TZ)$.}
      We have
\begin{align*}
        \mu_{T,Z} \circ T(g) \circ f
        & = \big(\id_{TZ}\circ_T\id_{TTZ}\big)
            \circ
            \big((\eta_{T,TZ} \circ g) \circ_T \id_{TY}\big)
            \circ f
          & \text{(by definition)} \\
        & = \id_{TZ} \circ_T
            \big(((\eta_{T,TZ} \circ g) \circ_T \id_{TY}) \circ f\big)
          & \text{(compatibility and $\circ$-unit)}\\
        & = \id_{TZ} \circ_T
            \big((\eta_{T,TZ} \circ g) \circ_T f\big)
          & \text{(compatibility and $\circ$-unit)}\\
        & = \big(\id_{TZ} \circ_T (\eta_{T,TZ} \circ g)\big) \circ_T f
          & \text{(associativity of $\circ_T$)}\\
        & = \big((\id_{TZ} \circ_T \eta_{T,TZ}) \circ g\big) \circ_T f
          & \text{(compatibility)}\\
        & = g \circ_T f
          & \text{($\circ_T$-unit and $\circ$-unit)}.
      \end{align*}
  \end{enumerate}
Points (i)--(v) show that $(T,\eta_T,\mu_T)$ satisfies all
  the conditions of Definition~\ref{def:classic-monad}, and (vi) shows that
  the Kleisli composition it induces coincides with~$\circ_T$.
\end{proof}

\subsection{Auxiliary Monad Identities}
\label{subsec:proof-auxiliary monad identities}

As mentioned above, we do not give the full backward direction to
Proposition~\ref{prop:monad-by-kleisli-equivalent-to-monad-by-mult}.
Nevertheless, we prove in next proposition that $(g \circ_T f) \circ u = g \circ_T (f \circ u)$ holds
for the definition of monads via multiplications $\mu$, and use this as a opportunity
to remind that while we also have the related $Tv \circ (g \circ_T f) = ((Tv \circ g) \circ_T f$,
we do \emph{not} have an equality between $u \circ (g \circ_T f)$ and $(u \circ g) \circ_T f$ in general.

\begin{proposition}[$\circ_T/\circ$ compatibility (1)]
   Let $(T, \eta_T, \mu_T)$ be a monad in a category $\catex$,
   in the sense of a functor $T:\catex\to\catex$
   and of natural transformations $\eta_T:\Id_\catex\Rightarrow T$ and $\mu_T:TT\Rightarrow T$
   satisfying
   $\mu_{TX} \circ T(\mu_{X}) = \mu_X \circ \mu_{TX}$
   and
   $\mu_X \circ T(\eta_X) = \mu_X \circ \eta_{TX} = \id_{TX}$
   for all $X\in\Ob(\catex)$.

   Let $\circ_T$ be defined by $g \circ_T f = \mu_{T,Z} \circ T(g) \circ f$
   for $f\in\catex(X, TY)$ and $g\in\catex(Y, TZ)$.

   Then, for all $u\in\catex(W,X)$, $f\in\catex(X, TY)$ and $g\in\catex(Y, TZ)$,
   we have $(g \circ_T f) \circ u = g \circ_T (f \circ u)$.
\end{proposition}
\begin{proof}
   Let $u\in\catex(W,X)$, $f\in\catex(X, TY)$ and $g\in\catex(Y, TZ)$ be three morphisms in $\catex$.

   By definition and associativity of $\circ$,
   $(g \circ_T f) \circ u = (\mu_{T,Z} \circ Tg \circ f) \circ u = \mu_{T,Z} \circ Tg \circ (f \circ u) = g \circ_T (f \circ u)$.
\end{proof}

The analogous result for \emph{postcomposition} via $\circ$ is the following.

\begin{lemma}[$\circ/\circ_T$ compatibility (2)]
   \label{lemma:compatibility-T-circ-circT}
   Let $(T, \eta_T, \mu_T)$ be a monad in $\catex$ and $\circ_T$ defined as above.

   Then, for all $f\in\catex(A, TB)$, $g\in\catex(B, TC)$, and $v\in\catex(C,D)$,
   we have $T(v) \circ (g \circ_T f) = (T(v) \circ g) \circ_T f$.
\end{lemma}
\begin{proof}
  Let $f\in\catex(A, TB)$, $g\in\catex(B, TC)$, and $v\in\catex(C,D)$ be three morphisms in $\catex$,
  so that $Tv\in\catex(TC,TD)$.
\begin{align*}
          Tv \circ (g \circ_T f)
= {}& Tv \circ \mu_{T,C} \circ Tg \circ f
          & \mathcmt{by definition}\\
= {}& \mu_{T,D} \circ TTv \circ Tg \circ f
          & \mathcmt{by naturality of $\mu$}\\
= {}& \mu_{T,D} \circ T(Tv \circ g) \circ f
          & \mathcmt{by functorality of $T$}\\
= {}& (Tv \circ g) \circ_T f
          & \mathcmt{by definition}
  \end{align*}
\end{proof}

\begin{remark}
  However, in general, the rebracketing from  $u \circ (g \circ_T f)$ to $(u \circ g) \circ_T f$ is \emph{not}
  valid in general for an arbitrary $u\in\catex(TC,TD)$ that is not necessarily of the form $u=Tv$ for $v\in\catex(C,D)$.
\end{remark}
\begin{proof}
  Consider the following counter-example. Let $T(-)$ be the \texttt{Maybe} monad in \Set,
  where we will write $TX = \{\None_X\} \uplus \{\Some(x)\,|\,x\in X\}$ for readability.
Recall that for $f:X\to Y$, $Tf:TX\to TY$ is defined by $(Tf)(\None_X)=\None_Y$ and
  $(Tf)(\Some(x))=\Some(f(x))$,
  and that for the multiplication, $\mu_{T,X}(\Some(\Some(x))=\Some(x)$,
  but $\mu_{T,X}(\Some(\None_X)) = \mu_{T,X}(\None_{TX}) = \None_X$.

  Let $\unit=\{*\}$ be a one point set,
  and define $f:\unit\to T\unit$ by $f(*) = \None_\unit$,
  $g:\unit\to T\unit$ by $g(*)=\Some(*)$,
  as well as $u:T\unit \to T\unit$ by $u(\Some(*)) = \None_\unit$ and $u(\None_\unit) = \Some(*)$.
We compare the two bracketings as functions $\unit\to T\unit$.

  In the first case, we have
  \begin{align*}
    (u \circ (g \circ_T f))(*)
    & = (u \circ \mu_{T,C}\circ Tg \circ f)(*)
      = (u \circ \mu_{T,C}\circ Tg)(\None_\unit) \\
    & = (u \circ \mu_{T,C})(\None_{T\unit})
      = u(\None_\unit)
      = \Some(*).
  \end{align*}
  In the second case,
  \begin{align*}
  ((u \circ g) \circ_T f)(*)
   & = (\mu_{T,D} \circ Tu \circ Tg \circ f)(*)
     = (\mu_{T,D} \circ Tu \circ Tg)(\None_\unit)\\
   & = (\mu_{T,D} \circ Tu)(\None_{T\unit})
     = \mu_{T,D}(\None_{T\unit})
     = \None_{\unit}.
  \end{align*}
In particular, $(u \circ (g \circ_T f))\neq((u \circ g) \circ_T f)$.
\end{proof}

\begin{remark}
  More generally, various counterexample can be built for any $u\in\catex(TX, TY)$
  such that $u \circ \mu_{T,X} \neq \mu_{T,Y}\circ Tu$.
  In contrast, if $u$ preserve multiplication in this sense, then
  $u \circ (g \circ_T f) = (u \circ g) \circ_T f$.
\end{remark}
 \vfill

\newpage

\begin{figure}[H]
  \section{Operator Semantics of \IR instantiated to the Maybe monad}
  \label{sec:operator-semantics-maybe-monad}
  \footnotesize
\begin{align*}
    \semop{{\tt Prog}}\, &
      \in \End\Big(\prod_{f\in \Funs} \Mem_{\Vars_{f, in}} \to \big(\Mem_{\Vars_{f, out}}\big)_\bot\Big)
       \\[5pt]
& = \funvec
         \mapsto \left(\argvec \mapsto
           \begin{cases}
             \pi_{\{X_{\text{res}}\}}\Big(\semstmt{f_{\text{body}}}(\funvec)(\argvec)\Big)
                  & \text{if } \semstmt{f_{\text{body}}}(\funvec)(\argvec) \neq \bot\\
             \bot & \text{if } \semstmt{f_{\text{body}}}(\funvec)(\argvec) = \bot
           \end{cases}
         \right)_{f \in \Funs}
  \end{align*}

  \vspace{5pt}
  \hrule height 0.08em
  \smallskip
  \hrule height 0.08em

  \begin{align*}
  \semstmt{^{\ell_{in}}{\tt stmt}^{\ell_{out}}}\,&
       \in \Big(\prod_{f\in \Funs} \Mem_{\Vars_{f, in}} \to \big(\Mem_{\Vars_{f, out}}\big)_\bot\Big)
       \to \Big(\Mem_{\Vars_{\ell_{in}}} \to \big(\Mem_{\Vars_{\ell_{out}}}\big)_\bot\Big)\\
\semstmt{{\tt skip}}(\funvec)
    & = m \mapsto
      m \\
  \semstmt{{\tt s1 ; s2}}(\funvec)
    & = m \mapsto
      \begin{cases}
        \semstmt{{\tt s2}}(\funvec)\big(\semstmt{{\tt s1}}(\funvec)(m)\big)
             & \text{if }\semstmt{{\tt s1}}(\funvec)(m) \neq \bot\\
        \bot & \text{if }\semstmt{{\tt s1}}(\funvec)(m) = \bot
      \end{cases}\\
  \semstmt{ X = {\tt expr}}(\funvec)
    & = m \mapsto
      \begin{cases}
        m\big[X \leftarrow \semexpr{{\tt expr}}(\funvec)(m)\big]
             & \text{if }\semexpr{{\tt expr}}(\funvec)(m) \neq \bot\\
        \bot & \text{if }\semexpr{{\tt expr}}(\funvec)(m) = \bot
      \end{cases}\\
  \semstmt{{\tt {\bf if}~cond~{\bf then}~\{ s1 \}~{\bf else}~\{ s2 \}}}(\funvec)
    & = m \mapsto
      \begin{cases}
        \semstmt{{\tt s1}}(\funvec)(m) & \text{if }\sembexp{{\tt cond}}(\funvec)(m) = {\tt true}\\
        \semstmt{{\tt s2}}(\funvec)(m) & \text{if }\sembexp{{\tt cond}}(\funvec)(m) = {\tt false}\\
        \bot                           & \text{if }\sembexp{{\tt cond}}(\funvec)(m) = \bot
      \end{cases}\\
\widesemstmt{
    \begin{aligned}
          &\tt {\bf match}~expr~\{ \dots\\
          &\quad{}\mid C_i(X_{i1},\dots,X_{ik}) {\tt~=>~stmt_i}\,\\
          &\ \quad{}\dots\}
    \end{aligned}}\hspace{-5pt}(\funvec) & = m \mapsto
      \begin{cases}
        \left(
        \begin{aligned}
          &\text{match $\semexpr{{\tt expr}}(\funvec)(m)$ with}
          \dots\\
          &\quad{}C_i(e_1,\dots,e_k)\ {\tt =>}\\
          &\qquad\semstmt{{\tt stmt_i}}(\funvec)(m[X_{ij} \gets e_j : j \in [1..k]])\\
          &\dots
        \end{aligned}
        \right)
             & \text{if } \semexpr{{\tt expr}}(\funvec)(m) \neq \bot\\
        \bot & \text{if } \semexpr{{\tt expr}}(\funvec)(m) = \bot
      \end{cases}
  \end{align*}

  \vspace{5pt}
  \hrule height 0.08em
  \smallskip
  \hrule height 0.08em

  \begin{align*}
  \semexpr{^\ell{\tt expr}}^{\tau} &
       \in \Big(\prod_{f\in \Funs} \Mem_{\Vars_{f, in}} \to \big(\Mem_{\Vars_{f, out}}\big)_\bot\Big)
       \to \Big(\Mem_{\Vars_{\ell}} \to \big(\semtypes{\tau}\big)_\bot\Big)\\[5pt]
\semexpr{{\tt X}}(\funvec)(m)
         & = m[{\tt X}]\\
\semexpr{C({\tt expr_1, \dots, expr_k})}(\funvec)(m)
         & = \begin{cases}
               C\Big(\semexpr{{\tt expr_1}}(\funvec)(m), \dots, \semexpr{{\tt expr_k}}(\funvec)(m)\Big)
                    & \text{if }\forall i,\, \semexpr{{\tt expr_i}}(\funvec)(m) \neq \bot\\
               \bot & \text{if }\exists i,\, \semexpr{{\tt expr_i}}(\funvec)(m) = \bot
             \end{cases}\\
\semexpr{{\tt f(expr_1, \dots, expr_k)}}(\funvec)(m)
         & = \begin{cases}
               \funvec_{\tt f}\Big(\semexpr{{\tt expr_1}}(\funvec)(m), \dots, \semexpr{{\tt expr_k}}(\funvec)(m)\Big)
                    & \text{if }\forall i,\, \semexpr{{\tt expr_i}}(\funvec)(m) \neq \bot\\
               \bot & \text{if }\exists i,\, \semexpr{{\tt expr_i}}(\funvec)(m) = \bot
             \end{cases}\\
\semexpr{{\tt expr_1}\, \Diamond_{\mathrm{a}}\, {\tt expr_2}}(\funvec)(m)
         & = \begin{cases}
               \semexpr{{\tt expr_1}}(\funvec)(m) \,\Diamond_{\mathrm{a}}\, \semexpr{{\tt expr_2}}(\funvec)(m)
                    & \text{if }\forall i,\, \semexpr{{\tt expr_i}}(\funvec)(m) \neq \bot\\
               \bot & \text{if }\exists i,\, \semexpr{{\tt expr_i}}(\funvec)(m) = \bot
             \end{cases}\\
\semexpr{{c}}(\funvec)(m)
         & = c
  \end{align*}

  \vspace{5pt}
  \hrule height 0.08em
  \smallskip
  \hrule height 0.08em

  \begin{align*}
  \sembexp{^\ell {\tt cond}} &
       \in \Big(\prod_{f\in \Funs} \Mem_{\Vars_{f, in}} \to \big(\Mem_{\Vars_{f, out}}\big)_\bot\Big)
       \to \Big(\Mem_{\Vars_{\ell}} \to \bb_\bot\Big)\\[5pt]
\sembexp{{\tt expr_1}\, \Bowtie\, {\tt expr_2}}(\funvec)(m)
         & = \begin{cases}
               \semexpr{{\tt expr_1}}(\funvec)(m) \,\Bowtie\, \semexpr{{\tt expr_2}}(\funvec)(m)
                    & \text{if }\forall i,\, \semexpr{{\tt expr_i}}(\funvec)(m) \neq \bot\\
               \bot & \text{if }\exists i,\, \semexpr{{\tt expr_i}}(\funvec)(m) = \bot
             \end{cases}\\
       \sembexp{{\tt cond_1}\, \Diamond_{\mathrm{b}}\, {\tt cond_2}}(\funvec)(m)
         & = \begin{cases}
               \sembexp{{\tt cond_1}}(\funvec)(m) \,\Diamond_{\mathrm{b}}\, \sembexp{{\tt cond_2}}(\funvec)(m)
                    & \text{if }\forall i,\, \sembexp{{\tt cond_i}}(\funvec)(m) \neq \bot\\
               \bot & \text{if }\exists i,\, \sembexp{{\tt cond_i}}(\funvec)(m) = \bot
             \end{cases}\\
       \sembexp{{b}}(\funvec)(m)
         & = b, \quad\ \text{for } b \in \{\textbf{true},\textbf{false}\}
  \end{align*}

  \caption{Operator semantics of the \IR language of Fig.~\ref{fig:syntax-ir-num-det}
    for the Maybe monad $(-)_\bot$.}
  \label{fig:opsem-ir-num-fail-extension}
\end{figure}

\section{Example: Probabilistic Semantics}
\label{app:proba-sem}

Powersets and derived constructions will be our main instantiation of the monadic approach to semantics (as is standard in abstract interpretation). However, the advantage of this approach is that it generalizes to more exotic situations. We present an example for programs with probabilities.

Following e.g.~\cite{Hasuo15}, consider the monad $(\dd_{=1}(-),\eta_{\dd_{=1}},\circ_{\dd_{=1}})$ on \Set,
given by
\begin{itemize}
  \setlength{\itemsep}{3pt}
  \item $\dd_{=1}(X) \coloneqq \{\mu : X \to [0,1]\,|\,\sum_{x \in X} \mu(x) = 1\}$, i.e., the distributions with countable support\footnote{\label{footnote:infinite-sums-weight}Here, $\sum_{x\in X}\mu(x)$ is defined as $\sup\{ \sum_{x\in F}\mu(x)\,|\, F\subseteq X \text{ finite}\}$.
  Since $\sum_{x\in X}\mu(x) \leq 1$, $\mu$ has countable support.} \item $\eta_{\dd_{=1}}$ as the dirac distributions defined as $\eta_{\dd_{=1},X}(x) \coloneqq r \mapsto \text{ite}(x = r, 1, 0)$
  \item $(g \circ_{\dd_{=1}} f)(x)(z) \coloneqq \sum_{y \in Y}f(x)(y)\cdot g(y)(z)$, i.e., the weighted average.
\end{itemize}

As it stands, $\dd_{=1}(X)$ lacks a convenient order structure that would give us
Propositions~\ref{prop:cpo-image-lfp} and~\ref{prop:clat-image-KT}.
A common approach in denotational semantics of recursive probabilistic systems is to consider
\emph{subdistributions} instead: $\dd_{\leq 1}(X):= \{\mu : X \to [0,1]\,|\,\sum_{x \in X} \mu(x) \leq 1\}$. These form a CPO under the pointwise order inherited from the usual order on $([0,1],\leq)$.
Proposition~\ref{prop:cpo-image-lfp} tells us that $\lfp\semop{\Prog}^{\dd_{\leq 1}}$ is well-defined.
This lfp may have total weight strictly less than one: this
mirrors the behaviour of partial functions in the non-probabilistic setting (using the Maybe monad),
where non-terminating traces simply do not contribute to the final weight.

Like $X_\bot$, subdistributions $\dd_{\leq 1}(X)$ do not generally form a complete
lattice (e.g. the join of two dirac distributions typically has weight $2>1$).
To enable postfixpoint reasoning (Proposition~\ref{prop:clat-image-KT}),
we can generalise $\dd_{\leq 1}(X)$ further,
considering instead the monad $(\ww_\infty(-),\eta_{\ww_\infty},\circ_{\ww_\infty})$
where $\ww_\infty(X) \coloneqq X\to[0,\infty]$ (with no restrictions)
ands $\eta_{\ww_\infty}$ and $\circ_{\ww_\infty}$ are defined as before,
using the conventions $0\cdot\infty  = 0$ and $x + \infty = \infty$.
Under the pointwise order, $\ww_\infty(X)$ is a complete lattice into which $\dd_{\leq 1}(X)$ naturally embeds.
Leveraging this structure alongside postfixpoint proof techniques, we can establish bounds on the
outcome probabilities of various processes.

For instance, consider an extension of~\IR (without {\color{azure}non-determinism}) by a probabilistic choice primitive
${{{\bf flip}(p)}\in{\tt BExp}}$, defined for each rational $p \in [0,1]$.
Selecting~$\ww_\infty(-)$ as the monad of choice in
Fig.~\ref{fig:opsem-ir-num-monadic}, we define $\sembexp{{\bf flip}(p)}(\funvec)(m)$
as the distribution ${[{\bf true}\mapsto p,\ {\bf false}\mapsto 1-p]}$.
Consider the program $\Prog$ walking randomly in $\zz$, with a 1/10 probability of exiting at each step:

\vspace{\baselineskip}
{\setlength{\tabcolsep}{3pt}\begin{tabular}{rl}
    \tt {\bf def} f(x:{\bf int}) -> {\bf int} \{
      & \tt {\bf if} {\bf flip}(1/10) {\bf then } \{\,{\bf skip}\,\}\\
      & \tt {\bf else} \{\,{\bf if} {\bf flip}(1/2) {\bf then} \{\,x = f(x-1)\,\} {\bf else} \{\,x = f(x+1)\,\}\,\} \\
      & \tt {\bf return} x\,\}.
  \end{tabular}}
\vspace{\baselineskip}

  We can compute its operator semantic, and obtain the operator
  \begin{align*}
    \semop{{\tt Prog}} \colon \big(\zz^2\to \ww_\infty(\zz)\big) &\to \zz^2 \to \ww_\infty(\zz)\\[-3pt]
    f &\mapsto x \mapsto \textstyle \frac{1}{10} \cdot \eta_{\dd_{=1},\zz}(x) + \frac{9}{10}\Big(\frac{1}{2} \cdot f(x-1) + \frac{1}{2} \cdot f(x+1)\Big),
  \end{align*}
  As $\ww_\infty(-)$ admits significantly more functions than $\dd_{=1}(-)$, many postfixpoints of $\semop{{\tt Prog}}$
  yield no meaningful information about the outcome probabilities of the program.
  But there are also many that do.
  An example is~$\candf(x)(r) \coloneqq \frac{1}{4}\Big(\frac{2}{3}\Big)^{\mathrm{|x-r|}}$:
its total weights are above $1$ (it does not yield distributions), but it still provides an upper
  bound strictly less than $1$ on the probability that the
  program exits with output $r$ on input $x$.
One can check (manually, or via a computer algebra system) that $\candf$ is a postfixpoint:
  $\semop{{\tt Prog}}(\candf)(x)(r) = \ite\Big(x=r,\, \frac{1}{4}, \,\frac{39}{40}\candf(x)(r)\Big) \leq \candf(x)(r)$
  for~all~${x,r\in\zz}$.

\begin{figure}[H]
  \section{Abstract Operator Semantics of \IR for Size Abstraction}
  \label{sec:operator-semantics-abs-sizes}
  \footnotesize
  \begin{align*}
    \asemop{{\tt Prog}}\, &
      \in \End\Big(\prod_{f\in \Funs} \aMem_{\Vars_{f, in}} \to \pp\big(\aMem_{\Vars_{f, out}}\big)\Big)
       \\[5pt]
& = \funvec^\sharp
         \mapsto \left(\argvec \mapsto
            \Big(\big(\eta_\pp \circ \pi_{\{X_{\text{res}}\}}\big)\circ_\pp \asemstmt{f_{\text{body}}}(\funvec^\sharp)\Big)(\argvec)
         \right)_{f \in \Funs}
  \end{align*}

  \vspace{3pt}
  \hrule height 0.08em
  \smallskip
  \hrule height 0.08em

  \begin{align*}
  \asemstmt{^{\ell_{in}}{\tt stmt}^{\ell_{out}}}\,&
       \in \Big(\prod_{f\in \Funs} \aMem_{\Vars_{f, in}} \to \pp\big(\aMem_{\Vars_{f, out}}\big)\Big)
       \to \Big(\aMem_{\Vars_{\ell_{in}}} \to \pp\big(\aMem_{\Vars_{\ell_{out}}}\big)\Big)\\
\asemstmt{{\bf skip}}(\funvec^\sharp)
    & = \eta_\pp \circ \alpha_{\Mem_{\Vars_\ell}} \circ \gamma_{\Mem_{\Vars_\ell}}
        \quad{\left(\parbox[l]{5.5cm}{closure, removing infeasible abstract memories,\\ applicable everywhere, see Section~\ref{subsubsec:size-transfer-skip}.}\right)}
        \\
    & = m^\sharp \mapsto
      \begin{cases}
        \eta_\pp(m^\sharp)  & \text{if } \exists m^\natural,\ \mathcal{M}(m^\natural)=m^\sharp\\
        \bot_\pp  & \text{otherwise}
      \end{cases}\\
  \asemstmt{{\tt s1 ; s2}}(\funvec^\sharp)
    & = m^\sharp \mapsto
      \big(\asemstmt{{\tt s2}}(\funvec^\sharp)\circ_\pp \asemstmt{{\tt s1}}(\funvec^\sharp)\big)(m^\sharp) \\
  \asemstmt{ X = {\tt expr}}(\funvec^\sharp)
    & = m^\sharp \mapsto
        \letin{\pp}
              {a}{\asemexpr{{\tt expr}}(\funvec^\sharp)(m^\sharp)}
              {\eta_\pp\big(m^\sharp[X \leftarrow a]\big)} \\
  \asemstmt{{\tt {\bf if}~cond~{\bf then}~\{ s1 \}~{\bf else}~\{ s2 \}}}(\funvec^\sharp)
    & = m^\sharp \mapsto
        \letin{\pp}
              {b}{\asembexp{{\tt cond}}(\funvec^\sharp)(m^\sharp)\\[-3pt]&\phantom{~=m\mapsto}}
              {\ite\big(b,\,\asemstmt{{\tt s1}}(\funvec^\sharp)(m^\sharp),\,\asemstmt{{\tt s2}}(\funvec^\sharp)(m^\sharp)\big)}
            \\
\widesemstmt{
    \begin{aligned}
          &\tt {\bf match}~expr~\{ \dots\\
          &\quad{}\mid C_i(X_{i1},\dots,X_{ik}) {\tt~=>~stmt_i}\,\\
          &\ \quad{}\dots\}
    \end{aligned}}^{\sharp}\hspace{-5pt}(\funvec^\sharp) & = m^\sharp \mapsto \letin{\pp}{\vec{v}}{\asemexpr{{\tt expr}}(\funvec^\sharp)(m^\sharp)}{
    \left(
      \begin{aligned}
        &{\textstyle\bigcup_i}\ \letin{\pp}{(\vec{v}_1,\dots,\vec{v}_k)}{\asem{{\tt dnstr_i}}(\vec{v})}\\
        &\hphantom{{\textstyle\bigcup_i}\ }\Big((\eta_\pp \circ \pi_{\Vars_{\ell,out}}) \circ_\pp \asemstmt{{\tt stmt_i}}(\funvec^\sharp)\Big)\\
        &\hphantom{{\textstyle\bigcup_i}\ \Big(}\big(m^\sharp[X_{ij} \gets \vec{v}_j : j \in [1..k]]\big)
      \end{aligned}
    \right)}
  \end{align*}

  \vspace{3pt}
  \hrule height 0.08em
  \smallskip
  \hrule height 0.08em

  \begin{align*}
  {\asemexpr{^\ell{\tt expr}}}^{\tau} &
       \in \Big(\prod_{f\in \Funs} \aMem_{\Vars_{f, in}} \to \pp\big(\aMem_{\Vars_{f, out}}\big)\Big)
       \to \Big(\aMem_{\Vars_{\ell}} \to \pp\big(\zz^{k_\tau}\big)\Big)\\[5pt]
\asemexpr{{\tt X}}(\funvec^\sharp)(m^\sharp)
         & = \eta_\pp\big(m^\sharp[{\tt X}]\big)\\
\asemexpr{C({\tt expr_1, \dots, expr_k})}(\funvec^\sharp)(m^\sharp)
         & = \letin{\pp}{a_1}{{\asemexpr{{\tt expr_1}}}(\funvec^\sharp)(m^\sharp)}
               {\dots~
                 \letin{\pp}{a_k}{{\asemexpr{{\tt expr_k}}}(\funvec^\sharp)(m^\sharp)
                 \\[-3pt]&\phantom{={}}}
                {\asem{{\tt cnstr}}(a_1, \dots, a_k)}}\\
\asemexpr{{\tt f(expr_1, \dots, expr_k)}}(\funvec^\sharp)(m^\sharp)
            &=  \letin{\pp}{a_1}{{\asemexpr{{\tt expr_1}}}(\funvec^\sharp)(m^\sharp)}
               {\dots~
                 \letin{\pp}{a_k}{{\asemexpr{{\tt expr_k}}}(\funvec^\sharp)(m^\sharp)
                 \\[-3pt]&\phantom{={}}}
                {\funvec^\sharp_{\tt f}(a_1, \dots, a_k)}}\\
\asemexpr{{\tt expr_1}\, \Diamond_{\mathrm{a}}\, {\tt expr_2}}(\funvec^\sharp)(m^\sharp)
          &=  \letin{\pp}{a_1}{\asemexpr{{\tt expr_1}}(\funvec^\sharp)(m^\sharp)}
             {\letin{\pp}{a_2}{\asemexpr{{\tt expr_2}}(\funvec^\sharp)(m^\sharp)}}
             {\eta_\pp(a_1\,\Diamond_{\mathrm{a}}\,a_2)}\\
\asemexpr{{c}}(\funvec^\sharp)(m^\sharp)
         & = \eta_\pp(c)
  \end{align*}

  \vspace{3pt}
  \hrule height 0.08em
  \smallskip
  \hrule height 0.08em

  \begin{align*}
  \asembexp{^\ell {\tt cond}} &
       \in \Big(\prod_{f\in \Funs} \aMem_{\Vars_{f, in}} \to \pp\big(\aMem_{\Vars_{f, out}}\big)\Big)
       \to \Big(\aMem_{\Vars_{\ell}} \to \pp\big(\bb\big)\Big)\\[5pt]
\asembexp{{\tt expr_1}\, \Bowtie\, {\tt expr_2}}(\funvec^\sharp)(m^\sharp)
          &=  \letin{\pp}{a_1}{\asemexpr{{\tt expr_1}}(\funvec^\sharp)(m^\sharp)}
             {\letin{\pp}{a_2}{\asemexpr{{\tt expr_2}}(\funvec^\sharp)(m^\sharp)}}
             {\eta_\pp(a_1\,\Bowtie\,a_2)}\\
       \asembexp{{\tt cond_1}\, \Diamond_{\mathrm{b}}\, {\tt cond_2}}(\funvec^\sharp)(m^\sharp)
          &=  \letin{\pp}{b_1}{\asembexp{{\tt cond_1}}(\funvec^\sharp)(m^\sharp)}
             {\letin{\pp}{b_2}{\asembexp{{\tt cond_2}}(\funvec^\sharp)(m^\sharp)}}
             {\eta_\pp(b_1\,\Diamond_{\mathrm{b}}\,b_2)}\\
       \asembexp{{b}}(\funvec^\sharp)(m^\sharp)
         & = \eta_\pp(b), \quad\ \text{for } b \in \{\textbf{true},\textbf{false}\}
  \end{align*}

  \vspace{3pt}
  \hrule height 0.08em
  \smallskip
  \hrule height 0.08em

  \begin{align*}
    \asemstmt{{\textbf{assume}(cond)}}(\funvec^\sharp)(m^\sharp) & = \ite\big(\textbf{true} \in \asembexp{cond}(\funvec^\sharp)(m^\sharp),\ \{m^\sharp\},\ \emptyset\big)\\
    \asemexpr{{(\top : \tau)}}(\funvec^\sharp)(m^\sharp) & = \top_\pp\ (= \zz^{k_\tau})
  \end{align*}

  \vspace{-1em}

  \caption{\footnotesize
    Abstract operator semantics of the \IR language of Fig.~\ref{fig:syntax-ir-num-det} for size abstraction (Section~\ref{subsec:size-abs-impl-comments}).
The number of metrics defined on type $\tau$ is denoted by $k_\tau\in\nn$.
    The only metric on {{\tt\bf int}} is the identity (integer value).
For $\Mem_\Vars = \prod_{t\in\Types}\big(\Vars_t \to \semtypes{t}\big)$,
        $\aMem_\Vars := \prod_{t\in\Types}\big(\Vars_t \to \zz^{k_t}\big)$.
    We denote by $\mathcal{M}:\Mem_\Vars\to \aMem_\Vars$ the operation which applies
    all metrics, replacing values by size vectors.
Feasibility constraints are left implicit outside of ${{\tt\bf skip}}$.
  }
  \label{fig:opsem-ir-abs-size-version}
\end{figure}

\section{Proofs}

\subsection{Proofs of Section~\ref{subsec:cat-gc}}
\label{subsec:proof-subsubsec-functors-C-CLat-sqcup}

\begin{proposition*}[\ref{prop:optimal-oplax-functor-by-GC}]
  Let $T:\catex\to\CLat$ be an oplax functor
  and $\big(TX \galois{\alpha_X}{\gamma_X} A_X\big)_{X\in\Ob(\mathcal{C})}$
  be an object-wise collection of Galois connections.
Then, the following is an oplax functor.
\begin{align*}
    T^\sharp:
           \catex &\to \CLat \\
               X  &\mapsto A_X \\
       (f:X\to Y) &\mapsto \alpha_Y \circ Tf \circ \gamma_X
  \end{align*}
Moreover, $\alpha:T\Rightarrow T^\sharp$ is an additive oplax natural transformation.
Furthermore, this oplax functor is \emph{initial} among all among all sound abstractions of $T$,
  i.e. among all oplax functors $T':\catex\Rightarrow\CLat_\sqcup$ such that
  $\alpha:T\Rightarrow T'$ is oplax natural.
\end{proposition*}

\begin{proof}
  \begin{itemize}
    \item We show oplax functoriality of $T^\sharp$.
    \begin{itemize}
      \item First, note that $T^\sharp:\catex\to\CLat$ is a prefunctor:
         indeed, for all $f\in\catex(X,Y)$, $T^\sharp(f)=\alpha_Y\circ Tf \circ \gamma_X$
         is monotone as a composition of monotone maps.
      \item For identities, let $X\in\Ob(\catex)$.
Using oplax functoriality of $T$ and monotonicity of $\alpha$, we have
$T^\sharp(\id_X)
           = \alpha_X \circ T(\id_X) \circ \gamma_X
           \pleq \alpha_X \circ \id_X \circ \gamma_X
           = \alpha_X \circ \gamma_X
           \pleq \id_{A_X}$,
where the final equality is by properties of GCs.
      \item For oplax functoriality of $T^\sharp$, let $f\in\catex(X,Y)$ and $g\in\catex(Y,Z)$.
Using oplax functoriality of $T$ and monotonicity of $\alpha$,
        as well as monotonicity of $Tg$, we have
$T^\sharp(g \circ f)
          = \alpha_Z \circ T(g \circ f) \circ \gamma_X
          \pleq \alpha_Z \circ T(g) \circ T(f) \circ \gamma_X
          \pleq \alpha_Z \circ T(g) \circ \gamma_Y \circ \alpha_Y \circ T(f) \circ \gamma_X
          = T^\sharp(g) \circ T^\sharp(f)$.
    \end{itemize}
    \item We show oplax naturality of the additive prenatural transformation $\alpha:T\Rightarrow T^\sharp$.
       Let $f\in\catex(X,Y)$.
We have
       $\alpha_Y \circ T(f)
         \pleq \alpha_Y \circ T(f) \circ \gamma_Y \circ \alpha_X
         = T^\sharp(f) \circ \alpha_X$.
    \item Now, we show that $T^\sharp$ is \emph{initial} among all sound abstractions of $T$.
Let $T':\catex\to\CLat$ be another oplax functor such that $\alpha:T\Rightarrow T'$
    is an oplax natural transformation, i.e. such that all $T'X = A_X$ for all $X\in\Ob(\catex)$
    and $\alpha_Y \circ Tf \pleq T'f \circ \alpha_X$ for all $f\in\catex(X,Y)$.
Since $T'$ and $T^\sharp$ are identical on objects,
    the unique possible oplax natural transformation $T^\sharp f \Rightarrow T'$
    is the identity $1_{\catex\to\CLat}:T^\sharp \Rightarrow T'$.
To show that it is indeed oplax, we need to show that $T^\sharp f \pleq T' f$ for all $f\in\catex(X,Y)$.
For this, simply observe that
    $\alpha_Y \circ Tf \pleq T'f \circ \alpha_X$ implies
    $T^\sharp f = \alpha_Y \circ Tf \circ \gamma_X \pleq T'f \circ \alpha_X \circ \gamma_X \pleq T'f$.
  \end{itemize}
\end{proof}
 Before proving Theorem~\ref{theorem:kan-hom-domain-codomain-abstraction},
we establish a few properties of the map $\exists_m(-)$.

As explained in the body of the paper,
given two posets $(X,\leq_X)$ and $(A,\leq_A)$, a complete lattice $(L,\leq_L)$,
and a monotone map~$m \colon X\to A$,
we have a Galois connection
\[
  \Pos(X, L) \galois{\exists_m(-)}{(-)\circ m} \Pos(A, L).
\]
In particular, this means that for every $f \in \Pos(X,L)$ and $g \in \Pos(A,L)$,
$\exists_m(f) \pleq g \iff f \pleq g\circ m$.

\medskip
In the next lemmas $X,A,L$ and $m$ are defined as above.

\begin{lemma}
  \label{lemma:f-pleq-emf-m}
  For every $f \in \Pos(X,L)$, we have $f \pleq \exists_m(f) \circ m$.
\end{lemma}
\begin{proof}
  Since $\exists_m(f) \pleq \exists_m(f)$, the Galois connection provides
  $f\pleq \exists_m(f) \circ m$.
\end{proof}

\begin{lemma}
  \label{lemma:em-fm-pleq-f}
  For every $g \in \Pos(A,L)$, we have $\exists_m(g\circ m) \pleq g$.
\end{lemma}
\begin{proof}
  Since $g \circ m \pleq g \circ m$, the Galois connection provides
  $\exists_m(g\circ m) \pleq g$.
\end{proof}

Recall that $\Pos$ is the category having posets as objects and monotone maps as morphisms, and $\CLat_\sqcup$ is the category having complete lattices as objects and additive maps as morphisms.
We show that, for any fixed complete lattice $L$, $\exists_{(-)}$ is a functor from the former category to the latter.

\begin{proposition}\label{prop:Kan-is-functorial}
  Let $L$ be a complete lattice.
  Then, the following is a functor
  \begin{align*}
    \exists_{(-)} \colon
      \Pos &\to \CLat_\sqcup\\
         X &\mapsto \Pos(X,L)\\
         m &\mapsto \big(f \mapsto \exists_m(f)\big).
  \end{align*}
\end{proposition}

\begin{proof}
  For the identity, let $X$ be a poset and $f \in \Pos(X,L)$.
      By Lemmas~\ref{lemma:f-pleq-emf-m}~and~\ref{lemma:em-fm-pleq-f},
      \begin{equation*}
        f
        \pleq \exists_{\id_X}(f) \circ \id_X
        = \exists_{\id_X}(f)
        = \exists_{\id_X}(f \circ \id_X)
        \pleq f,
      \end{equation*}
      hence $\exists_{\id_X}(f)=f$ by antisymmetry.

      For composition, this is a direct a consequence of associativity of $\circ$ via the right adjoint.

      Indeed, let $m \colon X\to Y$, $m'\colon Y\to Z$ and $f\colon X\to L$ be monotone functions.
      For all $h\colon Z\to L$,
      \begin{equation*}
        \exists_{m\circ m'}(f) \pleq h
        \iff
        f \pleq h \circ m \circ m'
        \iff
        \exists_{m'}(f) \pleq h \circ m
        \iff
        \exists_m(\exists_{m'}(f)) \pleq h.
      \end{equation*}
      In particular, $\exists_{m\circ m'}(f)\pleq \exists_m(\exists_{m'}(f))$
      and $\exists_m(\exists_{m'}(f))\pleq\exists_{m\circ m'}(f)$, as required.
\end{proof}

We need one more lemma about the interaction between $\exists_m(-)$
and additive maps $\alpha$.

\begin{lemma}\label{lemma:ex-add-into}
  Let $X, A$ be posets, $L_1, L_2$ be complete lattice, $f \colon X\to L_1$, $m \colon X\to A$ be monotone,
  and $\alpha \colon L_1\to L_2$ be additive.
Then, $\alpha \circ \exists_m(f) = \exists_m(\alpha \circ f)$.
\end{lemma}
\begin{proof}
  This is immediate from the definition of $\exists_m$,
  and the fact that $\alpha$ commutes with joins: for every $a \in A$,
  \begin{equation*}
    (\alpha \circ \exists_m(f))(a)
    = \alpha\big(\bigsqcup_{L_1} \{ f(x)\,|\, m(x) \leq_A a \}\big)
    = \bigsqcup_{L_2} \{ \alpha(f(x))\,|\, m(x) \leq_A a \}
    = (\exists_m(\alpha \circ f))(a).
  \end{equation*}
  We can also give an element-less proof that relies on the properties of the Galois connection
  instead of the explicit formula for $\exists_m$.
  Introduce the coadditive right-adjoint $\gamma:L_2\to L_1$
  to $\alpha$, and observe that for all $u:Y\to L_1$, $v:Y\to L_2$,
  we have $\alpha\circ u \pleq v \iff u \pleq \gamma \circ v$.
Hence, for any $h:Y\to L_2$, we have
  $\alpha\circ\exists_m(f)       \pleq h
   \iff \exists_m(f)        \pleq \gamma \circ h
   \iff            f        \pleq \gamma \circ h \circ m
   \iff     \alpha\circ f        \pleq h \circ m
   \iff \exists_m(\alpha\circ f) \pleq h$.

  In particular, we get $\alpha\circ\exists_m(f) \pleq \exists_m(\alpha\circ f) \pleq \alpha\circ\exists_m(f)$,
  which concludes by antisymmetry.
\end{proof}

We are now ready to prove Theorem~\ref{theorem:kan-hom-domain-codomain-abstraction}.

\TheoremKanHomDomainCodomainAbstraction*

\begin{proof}
  For the identity, let $X$ be a poset and $L$ be a complete lattice.
For all $f \colon X\to L$,
we simply have $K(\id_X,\id_L)(f) = \id_L\circ\exists_{\id_x}(f) = f$
    by Proposition~\ref{prop:Kan-is-functorial}.

  For composition, consider $m_1 \colon X\to Y$, $m_2 \colon Y\to Z$ two
    monotone maps, and $\alpha_1 \colon L\to L^\sharp$, $\alpha_2 \colon L^\sharp\to L^\sesquisharp$
    be two additive maps.
Then, for any $f \colon X\to L$, we have
\begin{align*}
          & K(m_2\circ m_1,\alpha_2\circ\alpha_1)(f)\\
        = {}& \alpha_2\circ\alpha_1\circ\exists_{m_2\circ m_1}(f)
            & \mathcmt{by definition}\\
        = {}& \alpha_2\circ\alpha_1\circ\exists_{m_2}(\exists_{m_1}(f))
            & \mathcmt{by Proposition~\ref{prop:Kan-is-functorial}}\\
        = {}& \alpha_2\circ\exists_{m_2}(\alpha_1\circ \exists_{m_1}(f))
            & \mathcmt{by Lemma~\ref{lemma:ex-add-into}}\\
        = {}& \big(K(m_2,\alpha_2)\circ K(m_1,\alpha_1)\big)(f).
            & \mathcmt{by definition}
            & \qedhere
    \end{align*}
\end{proof}

\subsection{Proofs of Section~\ref{subsec:state-space-absop}}
\label{subsec:proof-theorem-size-abstraction-kleisli-oplax}

\begin{theorem*}[\ref{theorem:size-abstraction-kleisli-oplax-endofunctor}]
  Consider an arbitrary mapping $\mathcal{M}:\Ob(\Set)\to\Ob(\Set)$,
  and an arbitrary collection of functions $m = \big(m_X : X \to \mathcal{M}X \big)_{X: \Set}$.
Note that we do not require any functoriality or naturality a priori.
Suppose that $(T(-), \eta_T, \circ_T)$ is a $\Set$-monad,
  pointwise order-enriched via a choice of $\CLat$ structure on $T$-objects,
  and that additionally all arrows $Tf$ are \emph{additive} for this \CLat structure
  (in other words, $T$ factorises as $T=U \circ T_s$ via a functor
  $T_s:\Set\to\CLat_\sqcup$).

  Then, $\mathcal{M}$ induces an \emph{additive} \emph{oplax} endofunctor in the corresponding Kleisli category, given by
  \begin{align*}
    \mathcal{M} :
      \Kleisli_T  &\to \Kleisli_T\\
              X   &\mapsto \mathcal{M}X\\
      (f:X\to TY) &\mapsto K(m_X, T m_Y)(f).
  \end{align*}
Moreover, $m:\Id_{\Kleisli_T} \Rightarrow \mathcal{M}$ is an oplax natural transformation.
\end{theorem*}

\begin{proof}
  We start by recalling that
  $\mathcal{M}_{X,Y}(f) = K(m_X, T m_Y)(f) = T(m_Y)\circ (\exists_{m_X}(f))$
  for all $f:X\to TY$.
Moreover, note that by assumption $T(m_Y)$ is monotone (and even
  additive) for the prescribed order on $T$-objects\footnote{
  To be fully formal, we could explicitly write the factorisation $T=U\circ T_s$,
  and distinguish $T(m_Y)\in\Set(TX, T\mathcal{M}Y)$ from
  $T_s(m_Y)\in\CLat_\sqcup(T_s X, T_s \mathcal{M} Y)$,
  but this is largely unimportant: $T(m_Y)=U(T_s(m_Y))$, so that they
  represent the same \Set-map -- $T_s(m_Y)$ simply serves as a witness
  that this map is monotone and additive for the prescribed order on
  objects.}.

\begin{itemize}
  \item First, observe that the construction is well typed:
since $T$ factors as $T=U\circ T_s$ and $T_s:\Pos \to \CLat_\sqcup$,
Corollary~\ref{cor:size-abstraction-bifunctor} immediately provides a functor
    $K(-,T_s -):\Set\times\Set\to\CLat_\sqcup$

    In particular, $\mathcal{M}_{X,Y}=K(m_X, T_s m_Y)$ is \emph{additive}, since
    it is a morphism in $\CLat_\sqcup$.

  \item Now, we immediately show the \emph{oplax naturality} condition
    (before checking oplax functoriality, since the inequality will be reused).

    We need to show that for all $f:X\to TY$, we have the following square.
    \begin{equation*}
    \begin{tikzcd}[column sep={1.5cm,between origins}, row sep={1.5cm,between origins}]
         X
           \arrow[r, "f"]
           \arrow[d, "m_X"']
      & TY
           \arrow[d, "Tm_Y"]
      \\
        \mathcal{M}X
           \arrow[r, "\mathcal{M}f"']
           \arrow[ru, "\geq\," rotate=50, midway, phantom, no line, node font=\Large]
      & T\mathcal{M}Y
    \end{tikzcd}
    \end{equation*}
For this, observe that $f \pleq \exists_{m_X}(f) \circ m_X$
    by Lemma~\ref{lemma:f-pleq-emf-m}.
Now, since $T(m_Y)$ is monotone, we obtain
    $T_s(m_Y)\circ f \pleq T_s(m_Y) \circ \exists_{m_X}(f) \circ m_X = \mathcal{M}(f)\circ m_X$.

  \item Now, let $f:X\to TY$ and $g:Y\to TZ$.
    We need to show that $\mathcal{M}(g\circ_T f)\pleq\mathcal{M}(g)\circ_T \mathcal{M}(f)$.

    Throughout the proof, we use monotonicity of $\circ_T$ in both arguments,
    monotonicity of $\exists_{m_X}(-)$ in its argument,
    and monotonicity of $T(m_z)$.
\begin{align*}
          & \mathcal{M}(g\circ_T f)\\
= {}& T(m_Z)\circ \exists_{m_X}(g\circ_T f)
          & \mathcmt{by definition}\\
\pleq\;{}& T(m_Z)\circ \exists_{m_X}\big(g\circ_T (\exists_{m_X}(f) \circ m)\big)
          & \mathcmt{by Lemma~\ref{lemma:f-pleq-emf-m}}\\
= {}& T(m_Z)\circ \exists_{m_X}\big((g\circ_T \exists_{m_X}(f)) \circ m\big)
          & \mathcmt{by $\circ_T$/$\circ$ compatibility}\\
 \pleq\;{}& T(m_Z)\circ (g\circ_T \exists_{m_X}(f))
          & \mathcmt{by Lemma~\ref{lemma:em-fm-pleq-f}}\\
       ={}& (T(m_Z)\circ g) \circ_T \exists_{m_X}(f)
          & \mathcmt{by monadicity (see Lemma~\ref{lemma:compatibility-T-circ-circT})}\\
 \pleq\;{}& (\mathcal{M}(g) \circ m_Y) \circ_T \exists_{m_X}(f)
          & \mathcmt{by oplax naturality above}\\
= {}& \mathcal{M}(g) \circ \big(m_Y \circ_T \exists_{m_X}(f)\big)
          & \mathcmt{by monadicity}\\
= {}& \mathcal{M}(g) \circ_T \mathcal{M}(f)
          & \mathcmt{by definition}
    \end{align*}

  \item Finally, we prove that for any $X:\Set$,
    $\mathcal{M}(\eta_{T,X}) \pleq \eta_{T,\mathcal{M}X}$.
For this, observe that
    $\mathcal{M}(\eta_{T,X})
      = T(m_X)\circ \exists_{m_X}(\eta_{T,X})
      = \exists_{m_X}(T(m_X)\circ\eta_{T,X})
      = \exists_{m_X}(\eta_{T,\mathcal{M}X}\circ m_{X})$,
    where the second equality is by Lemma~\ref{lemma:ex-add-into},
    and the final equality is by naturality of $\eta_T$.
Finally, applying Lemma~\ref{lemma:em-fm-pleq-f} and monotonicity of $\exists_{m_X}(-)$,
    we get
    $\mathcal{M}(\eta_{T,X}) = \exists_{m_X}(\eta_{T,\mathcal{M}X}\circ m_{X}) \pleq \eta_{\mathcal{M}X}$.
\end{itemize}
\end{proof}

\subsection{Proofs of Section~\ref{subsec:abstract-monads}}
\label{subsec:proofs-abstract-monads}

\begin{theorem*}[\ref{theorem:optimal-oplax-monad-by-GC}]
  Let $(T(-),\eta_T,\circ_T)$ be an abstract monad,
  and $\big(TX \galois{\alpha_X}{\gamma_X} A_X\big)_{X\in\Ob(\Set)}$ be an object-wise
  collection of Galois connections between complete lattices.
Suppose, furthermore, that these GCs are precise enough to represent units,
  in the sense $\gamma_X\circ\alpha_X\circ\eta_{T,X} = \eta_{T,X}$ for all $X$,
  and that these are Galois insertions, i.e. $\alpha_X\circ\gamma_X=\id$.

  Then, we obtain a structure of abstract monad
  $(T^\sharp(-),\eta_{T^\sharp},\circ_{T^\sharp})$,
by setting
$g \circ_{T^\sharp} f~:=~\alpha_Z \circ \big((\gamma_Z \circ g) \circ_T (\gamma_Y \circ f)\big)$
  for all $f\in\Pos(X,T^\sharp Y)$, $g\in\Pos(Y,T^\sharp Z)$,
and
  $\eta_{T^\sharp,X}~:=~\alpha_X \circ \eta_{T,X}$.

  Moreover, $\alpha$ induces the following (additive, normal) oplax functor between the corresponding
  Kleisli (unital) magmoids.
\begin{align*}
    \mathcal{A}:
       \widetilde{\Kleisli_T} &\Rightarrow \widetilde{\Kleisli_{T^\sharp}}\\
                           X &\mapsto X\\
                 (f:X\to TY) &\mapsto \alpha\circ f,
  \end{align*}

\end{theorem*}

\begin{proof}
  \begin{itemize}
    \item We show that $(T^\sharp(-),\eta_{T^\sharp},\circ_{T^\sharp})$ is an abstract monad,
      i.e. a premonad that is pointwise order-enriched, for a given choice of $\CLat$ order structure on objects
      $T^\sharp(X)$.
\begin{itemize}
      \item The order structure on objects is just the complete lattice structure on the $A_X$
        given in the collection of Galois connections.
      \item $\circ_{T^\sharp}$ is monotone in both arguments, by monotonicity of $\circ_T$, $\alpha$ and $\gamma$.
      \item
        Let $u:W\to W$, $f:X\to T^\sharp Y$, and $g:Y \to T^\sharp Z$.
We have
        \begin{align*}
        (g\circ_{T^\sharp}f) \circ u
        & = \big(\alpha_Z \circ ((\gamma_Z \circ g) \circ_T (\gamma_Y \circ f))\big) \circ u
          & \text{(definition)}\\
        & = \alpha_Z \circ \big(((\gamma_Z \circ g) \circ_T (\gamma_Y \circ f)) \circ u\big)
          & \text{($\circ$-associativity)}\\
        & = \alpha_Z \circ \big((\gamma_Z \circ g) \circ_T ((\gamma_Y \circ f) \circ u)\big)
          & \text{(compatibility $\circ$/$\circ_T$)}\\
        & = \alpha_Z \circ \big((\gamma_Z \circ g) \circ_T (\gamma_Y \circ (f \circ u))\big)
          & \text{($\circ$-associativity)}\\
        & = g \circ_{T^\sharp} (f \circ u)
          & \text{(definition)}.\\
        \end{align*}
     \item Let $f: X \to T^\sharp Y$.
       First, notice that
       \begin{align*}
         \eta_{T^\sharp,Y}\circ_{T^\sharp} f
         & = \alpha_Y \circ \big((\gamma_Y \circ \alpha_Y \circ \eta_{T, Y}) \circ_T (\gamma_Y \circ f)\big)
           & \text{(definition)}\\
         & = \alpha_Y \circ \big(\eta_{T, Y} \circ_T (\gamma_Y \circ f)\big)
           & \text{(assumption $\gamma\circ\alpha\circ\eta_T=\eta_T$)}\\
         & = \alpha_Y \circ \gamma_Y \circ f
           & \text{($\circ_T$-unit)}\\
         & = f
           & \text{(assumption $\alpha\circ\gamma=\id$)}.
       \end{align*}

       Similarly, $f \circ_{T^\sharp} \eta_{T^\sharp,X} = \alpha_Y \circ \gamma_Y \circ f = f$.
    \end{itemize}

    \item We show that $\mathcal{A}:\widetilde{\Kleisli_T}\Rightarrow\widetilde{\Kleisli_{T^\sharp}}$
      is additive normal oplax.
    \begin{itemize}
       \item Additivity of $\mathcal{A}$ (on arrows) is just codomain abstraction.
       \item For composition, let $f:X\to TY$, $g:Y\to TZ$.
We have
        $\mathcal{A}(g \circ_T f)
         = \alpha_Z \circ \big( g \circ_T f \big)
         \pleq \alpha_Z \circ \big( (\gamma_Z \circ \alpha_Z \circ g) \circ_T (\gamma_Y \circ \alpha_Y \circ f) \big)
         = \mathcal{A}(g) \circ_{T^\sharp} \mathcal{A}(f)$,
         where we have used monotonicity of $\alpha$, $\circ_T$, and $\id\pleq\gamma\circ\alpha$.
Hence $\mathcal{A}$ is oplax.
       \item For units, we have
         $\mathcal{A}(\eta_{T,X}) = \eta_{T^\sharp,X} = \eta_{T^\sharp,\mathcal{A}X}$
         hence $\mathcal{A}$ is \emph{normal}.
    \end{itemize}
  \end{itemize}
\end{proof}

One could wonder whether the constructions of optimal abstract monad (Theorem~\ref{theorem:optimal-oplax-monad-by-GC})
  and of optimal oplax functor (Proposition~\ref{prop:optimal-oplax-functor-by-GC}) are compatible.

  This is indeed the case when we start from some order-enriched concrete monad $(T(-),\eta_T,\circ_T)$ where
  the $Tf$ are monotone for the prescribed order structure (i.e. there is a factorisation of functors $T=U\circ
  T_s$ with $T_s:\Set\to\CLat$).

  We provide the following (additional) proposition to account for this.

\begin{proposition}
  \label{prop:oplax-functor-and-abstract-monad-compat}
  Let $(T(-), \eta_T, \circ_T)$ be an monad, pointwise order-enriched via some choice of complete lattice
  structure on $T$-objects, and assume that the $Tf$ are monotone for this order structure, i.e. that $T$
  factorises as $T=U\circ T_S$ for a functor $T_s:\Set\to\CLat$.

  Let $\big(TX \galois{\alpha_X}{\gamma_X} A_X\big)_{X\in\Ob(\catex)}$ be an object-wise
  collection of Galois connections as in Theorem~\ref{theorem:optimal-oplax-monad-by-GC}.

  Then, in the optimal abstract monad constructed by Theorem~\ref{theorem:optimal-oplax-monad-by-GC},
  $(T^\sharp(-),\eta_{T^\sharp},\circ_{T^\sharp})$,
  the construction $T^\sharp(f)=(\eta_{T^\sharp,Y}\circ f)\circ_{T^\sharp} \id_{T^\sharp X}$
  gives $T^\sharp$ the structure of an oplax functor, and this structure is identical
  to that obtained in Proposition~\ref{prop:optimal-oplax-functor-by-GC}
  via by setting instead $T^\sharp(f)=\gamma_Y\circ Tf \circ \alpha_X$.
\end{proposition}

\begin{proof}
  Let $X,Y$ be two sets, and $f:X\to Y$ be a function. We have
\begin{align*}
  (\eta_{T^\sharp,Y}\circ f)\circ_{T^\sharp} \id_{T^\sharp X}
  & = \alpha_Y \circ \big(
         (\gamma_Y \circ \alpha_Y \circ \eta_{T, Y} \circ f)
          \circ_T
          (\gamma_X \circ \id_{T^\sharp X})
         \big)
    & \text{(definition)}\\
  & = \alpha_Y \circ \big(
         (\gamma_Y \circ \alpha_Y \circ \eta_{T, Y} \circ f)
          \circ_T
          (\id_{T X} \circ \gamma_X)
         \big)
    & \text{($\circ$-unit)}\\
  & = \alpha_Y \circ \big(
         (\eta_{T, Y} \circ f)
          \circ_T
          (\id_{T X} \circ \gamma_X)
         \big)
    & \text{(assumption $\gamma\circ\alpha\circ\eta_T=\eta_T$)}\\
  & = \alpha_Y \circ \big(
         ((\eta_{T, Y} \circ f) \circ_T \id_{T X})
          \circ \gamma_X
         \big)
    & \text{(compatibility $\circ/\circ_T$)}\\
  & = \alpha_Y\circ Tf \circ \gamma_X
    & \text{($(T,\eta_T,\circ_T)$ monad)}.
  \end{align*}
\end{proof}

 \begin{proposition*}[\ref{prop:best-bracketing}]
   For all compatible maps $f,g,h$ in $\widetilde{\Kleisli}_\ii$,
   $(h \circ_\ii g) \circ_\ii f {}\mathop{\dot\sqsubseteq}{} h \circ_\ii (g \circ_\ii f)$.
\end{proposition*}
\begin{proof}
   This comes from the fact that $\mu_\pp$ is $(\pp-)$ join, and that $\alpha$ commutes with join,
   so that for any $k:X\to \pp(Y)$ and $E\in\pp(X)$,
   $(\alpha \circ \mu_\pp \circ \pp(\gamma\circ\alpha\circ k))(E) = \alpha\big(\bigcup_{x\in E} k(x)\big)
    = \bigsqcup_{x \in E}\big((\alpha\circ k)(x)\big)$.

   In particular, because $\alpha \circ \gamma \circ \alpha = \alpha$,
   for any $u:X \to \pp(Y)$ and $v:Y\to\pp(Z)$, we have
   $\alpha \circ ((\gamma \circ \alpha \circ v) \circ_\pp  u) = \alpha \circ (v \circ_\pp u)$.
Indeed, for all $x\in X$,
   $\big(\alpha \circ ((\gamma \circ \alpha \circ v) \circ_\pp  u)\big)(x)
    = \big(\alpha \circ \mu_\pp \circ \pp(\gamma \circ \alpha \circ v) \circ u\big)(x)
    = \alpha\Big(\mu_\pp\big(\;\{(\gamma\circ\alpha\circ v)(y) \,|\, y\in u(x)\}\;\big)\Big)
    = \alpha\big(\bigcup_{y\in u(x)} (\gamma\circ\alpha\circ v)(y)\big)
    = \bigsqcup_{y\in u(x)} (\alpha\circ\gamma\circ\alpha\circ v)(y)
    = \bigsqcup_{y\in u(x)} (\alpha\circ v)(y)
    = \alpha\big(\bigcup_{y\in u(x)} v(y)\big)
    = \big(\alpha \circ \mu_\pp \circ \pp(v) \circ u\big)(x)
    = \big(\alpha \circ (v \circ_\pp u)\big)(x)$.

   This can be applied to our question as follows.
   \begin{align*}
   (h \circ_\ii g) \circ_\ii f
      &= \big(\alpha \circ ((\gamma \circ h) \circ_\pp (\gamma \circ g))\big) \circ_\ii f \\
      &= \alpha \circ \Big(
            \big(\gamma \circ \alpha \circ ((\gamma \circ h) \circ_\pp (\gamma \circ g))\big)
             \circ_\pp (\gamma \circ f)\Big)\\
      &= \alpha \circ \Big(
            ((\gamma \circ h) \circ_\pp (\gamma \circ g))
             \circ_\pp (\gamma \circ f)\Big)\\
      &= \alpha \circ \Big( (\gamma \circ h) \circ_\pp (\gamma \circ g) \circ_\pp (\gamma \circ f) \Big).
    \end{align*}
In contrast, the over bracketing leads to
    \begin{align*}
    h \circ_\ii (g \circ_\ii f)
      &= h \circ_\ii \big(\alpha \circ ((\gamma \circ g) \circ_\pp (\gamma \circ f))\big) \\
      &= \alpha \circ \Big(
            \Big((\gamma \circ h) \circ_\pp
             \big(\gamma \circ \alpha \circ ((\gamma \circ g) \circ_\pp (\gamma \circ f))\big)
             \Big)\\
      &\pgeq \alpha \circ \Big(
            \Big((\gamma \circ h) \circ_\pp
             \big((\gamma \circ g) \circ_\pp (\gamma \circ f)\big)
             \Big)\\
      &= \alpha \circ \Big( (\gamma \circ h) \circ_\pp (\gamma \circ g) \circ_\pp (\gamma \circ f) \Big)\\
      &= (h \circ_\ii g) \circ_\ii f
   \end{align*}
   where we have used $\gamma\circ\alpha \geq \id$ to obtain the inequality.
\end{proof}

\end{document}